\documentclass[sigplan,screen]{acmart}

\setcopyright{none}
\acmPrice{15.00}
\acmDOI{10.1145/3453483.3454062}
\acmYear{2021}
\copyrightyear{2021}
\acmSubmissionID{pldi21main-p227-p}
\acmISBN{978-1-4503-8391-2/21/06}
\acmConference[PLDI '21]{Proceedings of the 42nd ACM SIGPLAN International Conference on Programming Language Design and Implementation}{June 20--25, 2021}{Virtual, UK}
\acmBooktitle{Proceedings of the 42nd ACM SIGPLAN International Conference on Programming Language Design and Implementation (PLDI '21), June 20--25, 2021, Virtual, UK}

\usepackage{xspace}
\usepackage{enumerate}
\usepackage[shortlabels]{enumitem}
\usepackage{subcaption}
\usepackage{mathpartir}
\usepackage{turnstile}
\usepackage{mathtools}
\mathtoolsset{showmanualtags}
\usepackage{stmaryrd}
\usepackage{relsize}
\usepackage{multirow}
\usepackage{multicol}
\usepackage{algorithm}
\usepackage[noend]{algpseudocode}
\usepackage{pseudo}
\pseudoset{indent-length=0.4em}
\usepackage{tikz}
\usetikzlibrary{arrows,automata,chains,shapes,trees,fit}
\usetikzlibrary{external}
\usepackage{xfrac}
\usepackage{hhline}
\usepackage{rotating}
\usepackage{wrapfig}
\usepackage{cleveref}
\crefname{theorem}{Thm.}{Thms.}
\crefname{lemma}{Lem.}{Lemmas}
\crefname{corollary}{Cor.}{Cors.}
\crefname{figure}{Fig.}{Figs.}
\crefname{definition}{Defn.}{Defns.}
\crefname{table}{Tab.}{Tabs.}
\crefformat{section}{\S#2#1#3}
\crefmultiformat{section}{\S#2#1#3}{ and~\S#2#1#3}{, \S#2#1#3}{ and~\S#2#1#3}
\crefname{example}{Ex.}{Exs.}
\crefname{item}{item}{items}
\crefname{footnote}{footnote}{footnotes}
\crefname{observation}{Obs.}{Obs.}
\crefname{remark}{Remark}{Remarks}
\crefname{proposition}{Prop.}{Props.}
\crefname{fact}{Fact}{Facts}
\crefname{challenge}{Challenge}{Challenges}
\crefname{principle}{Principle}{Principles}
\crefname{counterexample}{Ex.}{Exs.}
\usepackage{shortcuts}
\usepackage{balance}

\makeatletter 
\AtEndDocument{%
\def\cb@checkPdfxy#1#2#3#4#5{%
\cb@@findpdfpoint{#1}{#2}%
\ifdim#3sp=\cb@pdfx\relax      
\ifdim#4sp=\cb@pdfy\relax      
\ifdim#5=\cb@pdfz\relax
\else
\cb@error
\fi
\else
\cb@error
\fi
\else
\cb@error
\fi
}}
\makeatother 

\usepackage{color}
\usepackage[color]{changebar}
\cbcolor{ACMRed}
\nochangebars

\reversemarginpar
\newcommand{\revision}[1]{}

\newif\iflong
\longtrue

\usepackage[cal=cm]{mathalfa}
\renewcommand{\mathbb}[1]{\vvmathbb{#1}}

\newtheoremstyle{acmremark}%
  {.5\baselineskip\@plus.2\baselineskip\@minus.2\baselineskip}
  {.5\baselineskip\@plus.2\baselineskip\@minus.2\baselineskip}
  {\itshape}
  {\parindent}
  {\itshape}
  {.}
  {.5em}
  {\thmname{#1}\thmnumber{ #2}\thmnote{ {(#3)}}}
\AtEndPreamble{%
  \theoremstyle{acmplain}%
  \newtheorem*{theorem*}{Theorem}%
  \newtheorem*{lemma*}{Lemma}%
  \newtheorem*{proposition*}{Proposition}%
  \theoremstyle{acmdefinition}%
  \newtheorem{counterexample}[theorem]{Counterexample}%
  \theoremstyle{acmremark}%
  \newtheorem{remark}[theorem]{Remark}%
  \theoremstyle{acmplain}%
}

\AtBeginDocument{%
  \setlength\abovedisplayskip{0.5\abovedisplayskip}%
  \setlength\belowdisplayskip{0.5\belowdisplayskip}%
  \setlength\abovedisplayshortskip{0.5\abovedisplayshortskip}%
  \setlength\belowdisplayshortskip{0.5\belowdisplayshortskip}%
  \setlength\floatsep{0.5\floatsep}%
  \setlength\abovecaptionskip{0.5\abovecaptionskip}%
}
\setlist{nosep,leftmargin=\parindent}

\newsavebox{\fminipagebox}
\NewDocumentEnvironment{fminipage}{m O{\fboxsep}}
 {\par\kern#2\noindent\begin{lrbox}{\fminipagebox}
  \begin{minipage}{#1}\ignorespaces}
 {\end{minipage}\end{lrbox}%
  \makebox[#1]{%
    \kern\dimexpr-\fboxsep-\fboxrule\relax
    \fbox{\usebox{\fminipagebox}}%
    \kern\dimexpr-\fboxsep-\fboxrule\relax
  }\par\kern#2
 }

\usepackage{array}
\newcolumntype{H}{>{\setbox0=\hbox\bgroup}c<{\egroup}@{}}

\newcommand{\m}[1]{\mathsf{#1}}

\newcommand{\maxp}[1]{\mathopen{|[} #1 \mathclose{]|}}

\newcommand{\aone}{\underline{1}}
\newcommand{\azero}{\underline{0}}
\newcommand{\aord}{\sqsubseteq}

\newcommand{\prob}{\ensuremath{\bbP}}
\newcommand{\expe}{\ensuremath{\bbE}}
\newcommand{\vari}{\ensuremath{\bbV}}

\newcommand{\lang}{\textsc{Appl}}
\newcommand{\reason}[1]{\dagger~\text{#1}~\dagger}

\newcommand{\many}[1]{\overrightarrow{#1}}
\newcommand{\bind}{\mathbin{\gg\!=}}

\newcommand{\iskip}{\kw{skip}}
\newcommand{\itick}[1]{\kw{tick}(#1)}
\newcommand{\iassign}[2]{#1 \coloneqq #2}
\newcommand{\isample}[2]{#1 \sim #2}
\newcommand{\iprob}[3]{\kw{if}~\kw{prob}(#1)~\kw{then}~#2~\kw{else}~#3~\kw{fi}}
\newcommand{\icond}[3]{\kw{if}~#1~\kw{then}~#2~\kw{else}~#3~\kw{fi}}
\newcommand{\iloop}[2]{\kw{while}~#1~\kw{do}~#2~\kw{od}}
\newcommand{\iinvoke}[1]{\kw{call}~#1}

\newcommand{\ctop}{\kw{true}}
\newcommand{\cneg}[1]{\kw{not}~#1}
\newcommand{\cconj}[2]{#1~\kw{and}~#2}
\newcommand{\cbin}[3]{#2 \mathrel{#1} #3}
\newcommand{\cle}[2]{\cbin{\le}{#1}{#2}}

\newcommand{\eadd}[2]{#1 + #2}

\newcommand{\emul}[2]{#1 \times #2}

\newcommand{\kstop}{\kw{Kstop}}
\newcommand{\kseq}[2]{\kw{Kseq}~#1~#2}
\newcommand{\kloop}[3]{\kw{Kloop}~#1~#2~#3}

\newcommand{\twr}[1]{{\color{ACMGreen}\emph{Tom: #1}}}

\begin{document}

\title{Central Moment Analysis for Cost Accumulators in Probabilistic Programs}
\iflong
\subtitle{Technical Report}
\fi


\author{Di Wang}
\affiliation{
  \institution{Carnegie Mellon University}
  \country{USA}
}

\author{Jan Hoffmann}
\affiliation{
  \institution{Carnegie Mellon University}
  \country{USA}
}

\author{Thomas Reps}
\affiliation{
  \institution{University of Wisconsin}
  \country{USA}
}


\begin{abstract}
For probabilistic programs, it is usually not possible to automatically
derive exact information about their properties, such as the distribution
of states at a given program point.
Instead, one can attempt to derive approximations, such as upper bounds
on \emph{tail probabilities}.
Such bounds can be obtained via concentration inequalities, which rely
on the \emph{moments} of a distribution, such as
the expectation (the first \emph{raw} moment) or
the variance (the second \emph{central} moment).
%
Tail bounds obtained using central moments are often
tighter than the ones obtained using raw moments, but
automatically analyzing central moments is more challenging.
%
%

This paper presents an analysis for probabilistic programs that
automatically derives symbolic upper and lower bounds on
variances, as well as higher central moments, of \emph{cost accumulators}.
To overcome the challenges of higher-moment analysis, it generalizes
analyses for expectations with an algebraic abstraction that
simultaneously analyzes different moments, utilizing relations between
them.
%
%
A key innovation is the notion of \emph{moment-polymorphic recursion}, and a
practical derivation system that handles recursive functions.
%

%
%

The analysis has been implemented using a template-based technique
that reduces the inference of polynomial bounds to linear programming.
Experiments with our prototype central-moment analyzer show that,
despite the analyzer's upper/lower bounds on various quantities,
it obtains tighter tail bounds than an existing system that uses only raw
moments, such as expectations.
\end{abstract}



\maketitle

\tikzexternaldisable

%
%
%

\section{Introduction}
\label{Se:Intro}

Probabilistic programs~\cite{JCSS:Kozen81,book:MM05,FOSE:GHN14} can be used to manipulate \emph{uncertain} quantities modeled by probability distributions and random control flows.
Uncertainty arises naturally in Bayesian networks, which capture statistical dependencies (e.g., for diagnosis of diseases~\cite{JAMIA:JC10}), and cyber-physical systems, which are subject to sensor errors and peripheral disturbances (e.g., airborne collision-avoidance systems~\cite{DASC:MJ16}).
Researchers have used probabilistic programs to implement and analyze
randomized algorithms~\cite{ESOP:BEG18}, cryptographic protocols~\cite{POPL:BGZ09}, and machine-learning algorithms~\cite{NATURE:Ghahramani15}.
A probabilistic program propagates uncertainty through a computation, and produces a distribution over results.

%
In general, it is not tractable to compute the result distributions of probabilistic programs automatically and precisely: Composing simple distributions can quickly complicate the result distribution, and randomness in the control flow can easily lead to state-space explosion.
Monte-Carlo simulation~\cite{book:MonteCarlo16} is a common approach to study the result distributions, but the technique does not provide formal guarantees, and can sometimes be inefficient~\cite{ESOP:BKK18}.

In this paper, we focus on a specific yet important kind of uncertain quantity: \emph{cost accumulators}, which are program variables that can only be incremented or decremented through the program execution.
Examples of cost accumulators include
termination time~\cite{ESOP:KKM16,CAV:CFG16,POPL:CFN16,ESOP:BKK18,LICS:OKK16},
rewards in Markov decision processes (MDPs)~\cite{book:Puterman94},
position information in control systems~\cite{TACAS:BGP16,ATVA:BKS19,PLDI:SCG13},
and cash flow during bitcoin mining~\cite{PLDI:WFG19}.
Recent work~\cite{PLDI:NCH18,TACAS:KUH19,TACAS:BGP16,PLDI:WFG19} has proposed successful static-analysis approaches that leverage
\emph{aggregate} information of a cost accumulator $X$, such as $X$'s \emph{expected} value $\expe[X]$ (i.e., $X$'s ``first moment'').
The intuition why it is beneficial to compute aggregate information---in lieu of distributions---is that aggregate measures like expectations 
\emph{abstract} distributions to a single number, while still indicating non-trivial properties.
Moreover, expectations are transformed by statements in a probabilistic program in a manner similar to the weakest-precondition transformation of formulas in a non-probabilistic program~\cite{book:MM05}.

One important kind of aggregate information is \emph{moments}.
In this paper, we show how to obtain \emph{central} moments (i.e., $\expe[(X-\expe[X])^k]$ for any $k \ge 2$),
whereas most previous work focused on \emph{raw} moments (i.e., $\expe[X^k]$ for any $k \ge 1$).
Central moments can provide more information about distributions.
For example, the \emph{variance} $\vari[X]$ (i.e., $\expe[(X-\expe[X])^2]$, $X$'s ``second central moment'') indicates how $X$ can deviate from its mean,
the \emph{skewness} (i.e., $\frac{ \expe[(X-\expe[X])^3] }{ (\vari[X])^{\sfrac{3}{2}} }$, $X$'s ``third standardized moment'') indicates how lopsided the distribution of $X$ is,
and the \emph{kurtosis} (i.e., $\frac{ \expe[(X - \expe[X])^4 ] }{ (\vari[X])^2 }$, $X$'s ``fourth standardized moment'') measures the heaviness of the tails of the distribution of $X$.
One application of moments is to answer queries about \emph{tail bounds}, e.g., the assertions about probabilities of the form $\prob[X \ge d]$, via \emph{concentration-of-measure} inequalities from probability theory~\cite{book:Dubhashi09}.
With central moments, we find an opportunity to obtain more precise tail bounds of the form $\prob[X \ge d]$, and become able to derive bounds on tail probabilities of the form $\prob[|X - \expe[X]| \ge d]$.

Central moments $\expe[(X-\expe[X])^k]$ can be seen as polynomials of raw moments $\expe[X],\cdots,\expe[X^k]$, e.g., the variance $\vari[X] = \expe[(X-\expe[X])^2]$ can be rewritten as $\expe[X^2]-\expe^2[X]$, where $\expe^k[X]$ denotes $(\expe[X])^k$.
To derive bounds on central moments, we need both \emph{upper} and \emph{lower} bounds on the raw moments, because of the presence of \emph{subtraction}.
For example, to upper-bound $\vari[X]$, a static analyzer needs to have an \emph{upper} bound on $\expe[X^2]$ and a \emph{lower} bound on $\expe^2[X]$.

In this work, we present and implement the first fully automatic analysis for deriving symbolic \emph{interval}
bounds on higher central moments for cost accumulators in probabilistic programs with general recursion
and continuous distributions.
One challenge is to support \emph{inter}procedural reasoning to reuse analysis results for functions.
Our solution makes use of a ``lifting'' technique from the natural-language-processing community.
That technique derives an algebra for second moments from an algebra for first moments~\cite{EMNLP:LE09}.
We generalize the technique to develop \emph{moment semirings},
and use them to derive a novel \emph{frame} rule to handle function calls with \emph{moment-polymorphic recursion} (see \cref{Se:SoundnessCriteria}).

Recent work has successfully automated inference of upper~\cite{PLDI:NCH18} or lower bounds~\cite{PLDI:WFG19} on the expected cost of probabilistic programs.
\citet{TACAS:KUH19} developed a system to derive upper bounds on higher \emph{raw} moments of program runtimes.
\revision{5}
\begin{changebar}
However, even in combination, existing approaches can\emph{not} solve tasks such as
deriving a lower bound on the second raw moment of runtimes, or
deriving an upper bound on the variance of accumulators that count live heap cells.
\end{changebar}
\cref{Fi:MotivatingExample}(a) summarizes the features of related work on moment inference for probabilistic programs.
To the best of our knowledge, our work is the first moment-analysis tool that supports all of the listed programming and analysis features.
\cref{Fi:MotivatingExample}(b) and (c) compare our work with related work in terms of tail-bound analysis
on a concrete program (see \cref{Se:TailBound}).
The bounds are derived for the cost accumulator \id{tick} in a random-walk program that we will present in \cref{Se:Overview}.
It can be observed that for $d \ge 20$, the most precise tail bound for \id{tick} is the one obtained via an upper bound on the variance $\vari[\id{tick}]$ ($\id{tick}$'s second central moment).

\begin{figure}
\centering
\begin{subfigure}{\columnwidth}
  \centering
  \begin{small}
  \begin{tabular}{c|c|c|c|c|c}
    \hline
    feature & \cite{TACAS:BGP16} & \cite{PLDI:NCH18} & \cite{TACAS:KUH19} & \cite{PLDI:WFG19} & this work \\ \hline
    loop & & \checkmark & \checkmark & \checkmark & \checkmark \\
    recursion & & \checkmark & & & \checkmark \\
    continuous distributions & \checkmark & & \checkmark & \checkmark & \checkmark \\
    non-monotone costs & \checkmark & & & \checkmark & \checkmark \\
    higher moments & \checkmark & & \checkmark & & \checkmark \\
    interval bounds & \checkmark & & & \checkmark & \checkmark \\ \hline
  \end{tabular}  
  \end{small}
  \vspace{-5pt}
  \caption*{(a)}
  \vspace{8pt}
\end{subfigure}
\begin{subfigure}{\columnwidth}
  \centering
  \begin{small}
  \begin{tabular}{@{\hspace{2pt}}c@{\hspace{2pt}}|@{\hspace{2pt}}c@{\hspace{2pt}}|@{\hspace{2pt}}c@{\hspace{2pt}}|@{\hspace{2pt}}c@{\hspace{2pt}}}
    \hline
    & \cite{PLDI:NCH18,PLDI:WFG19} & { \cite{TACAS:KUH19}} & {this work} \\ \hline
    Derived  & \multirow{2}{*}{$\expe[\id{tick}] \le 2d + 4$} & $\expe[\id{tick}^2] \le {}$ & $\vari[\id{tick}] \le {} $ \\
    bound & & $4d^2+22d+28$ & $22d+28$ \\ \hline
    Moment type & raw & raw & central \\ \hline
    Concentration & {Markov} & {Markov} & \multirow{2}{*}{Cantelli} \\
    inequality & ($\text{degree} = 1$) & ($\text{degree} = 2$) & \\ \hline
    Tail bound  & \multirow{2}{*}{ $\approx \frac{1}{2}$} & \multirow{2}{*}{$\approx \frac{1}{4} $} & \multirow{2}{*}{$\xrightarrow{d \to \infty} 0 $ } \\
    $\prob[\id{tick} \ge 4d]$ &  & & \\ \hline
  \end{tabular}
  \end{small}
  \vspace{-5pt}
  \caption*{(b)}
  \vspace{2pt}
\end{subfigure}
\begin{subfigure}{\columnwidth}
  \centering
  \small
  \input{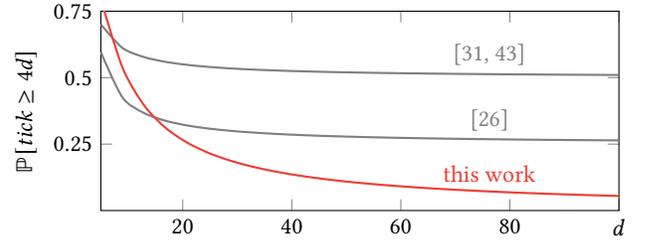}
  \vspace{-18pt}
  \caption*{(c)}
  \vspace{-2pt}
\end{subfigure}
\caption{(a) Comparison in terms of supporting features. (b) Comparison in terms of moment bounds for the running example. (c) Comparison in terms of derived tail bounds.}
\label{Fi:MotivatingExample}
\vspace{-10pt}
\end{figure}

\emph{Our work incorporates ideas known from the literature:}
\begin{itemize}[-]
  \item Using the expected-potential method (or ranking super-martingales) to derive upper bounds on the expected program runtimes or monotone costs~\cite{PLDI:NCH18,TACAS:KUH19,CAV:CFG16,CAV:CS13,POPL:CFN16,POPL:FH15}.
  \item Using the \emph{Optional Stopping Theorem} from probability theory to ensure the soundness of lower-bound inference for probabilistic programs~\cite{PLDI:WFG19,POPL:HKG20,CAV:BEF16,misc:SO19}.
  \item Using \emph{linear programming} (LP) to efficiently automate the (expected) potential method for (expected) cost analysis~\cite{ICFP:WKH20,POPL:HJ03,APLAS:HH10}.
\end{itemize}
%
\emph{The contributions of our work are as follows:}
\begin{itemize}
  \item We develop moment semirings to compose the moments for a cost accumulator from two computations, and to enable interprocedural reasoning about higher moments.
  
  \item
  
  We instantiate moment semirings with the symbolic interval domain, use that to develop a derivation system for interval bounds on higher central moments for cost accumulators, and automate the derivation via LP solving.
  %
  %
  
  \item
  We prove the soundness of our derivation system for programs that satisfy the criterion of our recent extension to the Optional Stopping Theorem,
  and develop an algorithm for checking this criterion automatically.

  

  \item
We implemented our analysis
and evaluated it
on a broad suite of benchmarks from the literature.
%
%
Our experimental results show that on a variety of examples, our analyzer is able to use higher central moments to obtain tighter tail bounds on program runtimes than the system of~\citet{TACAS:KUH19}, which uses only upper bounds on raw moments.
\end{itemize}
%



\section{Overview}
\label{Se:Overview}

In this section, we demonstrate the expected-potential method for both first-moment analysis (previous work) and higher central-moment analysis (this work) (\cref{Se:ExpectedPotential}),
and discuss the challenges to supporting interprocedural reasoning and to ensuring the soundness of our approach (\cref{Se:SoundnessCriteria}).
%

\begin{example}\label{Exa:RecursiveRandomWalk}
  The program in \cref{Fi:RecursiveRandomWalk} implements a bounded, biased random walk.
  The main function consists of a single statement ``$\iinvoke{\mathsf{rdwalk}}$'' that invokes a recursive function.
  The variables $x$ and $d$ represent the current position and the ending position of the random walk, respectively.
  We assume that $d>0$ holds initially.
  In each step, the program samples the length of the current move from a uniform distribution on the interval $[{-1},2]$.
  The statement $\itick{1}$ adds one to a cost accumulator that counts the number of steps before
  the random walk ends. 
  We denote this accumulator by \id{tick} in the rest of this section.
  The program terminates with probability one and its expected accumulated cost is bounded by
  $2d+4$.
  %
\end{example}

\begin{figure}
  \centering
  \begin{subfigure}{0.49\columnwidth}
  \centering
  \begin{small}
  \begin{pseudo}
    \kw{func} $\mathsf{rdwalk}$() \kw{begin} \\+
      ${ \color{ACMDarkBlue} \{ \; 2(d-x)+4 \;\} }$ \\
      \kw{if} $\cbin{<}{x}{d}$ \kw{then} \\+
        ${\color{ACMDarkBlue} \{\;  2(d-x)+4 \;\}  }$ \\
        $\isample{t}{ \kw{uniform}({-1},2)}$; \\
        ${\color{ACMDarkBlue} \{\;  2(d - x - t) + 5  \;\} }$ \\
        $ \iassign{x}{  \eadd{x}{t} }$; \\
        ${\color{ACMDarkBlue} \{\;  2(d - x) + 5 \;\} }$ \\
        \kw{call} $\mathsf{rdwalk}$; \\
        ${\color{ACMDarkBlue} \{ \; 1 \; \} }$ \\
        \kw{tick}(1) \\
        ${\color{ACMDarkBlue} \{ \; 0 \; \} } $ \\-
      \kw{fi} \\-
    \kw{end}
  \end{pseudo}
  \end{small}
  \end{subfigure}
  \begin{subfigure}{0.49\columnwidth}
  \centering
  \begin{small}
  \begin{pseudo}
    {\color{gray} \# $\mathsf{XID} \defeq \{x,d,t\}$ } \\
    {\color{gray} \# $\mathsf{FID} \defeq \{\mathsf{rdwalk} \}$ } \\
    {\color{gray} \# \emph{pre-condition}: $\{d > 0 \}$ } \\
    \kw{func} $\mathsf{main}$() \kw{begin} \\+
      $\iassign{x}{0}$; \\
      \kw{call} $\mathsf{rdwalk}$ \\-
    \kw{end}
  \end{pseudo}
  \end{small}
  \end{subfigure}
  \caption{A bounded, biased random walk, implemented using recursion. The annotations show the derivation of an \emph{upper} bound on the expected accumulated cost.}
  \label{Fi:RecursiveRandomWalk}
\end{figure}

\subsection{The Expected-Potential Method for Higher-Moment Analysis}
\label{Se:ExpectedPotential}

Our approach to higher-moment analysis is inspired by the \emph{expected-potential method}~\cite{PLDI:NCH18}, which is also known as \emph{ranking super-martingales}~\cite{CAV:CS13,PLDI:WFG19,TACAS:KUH19,CAV:CFG16}, for expected-cost bound analysis of probabilistic programs.
%

The classic \emph{potential method} of amortized analysis~\cite{JADM:Tarjan85} can be automated to derive symbolic cost bounds for non-probabilistic programs~\cite{POPL:HJ03,APLAS:HH10}.
The basic idea is to define a \emph{potential function} $\phi: \Sigma \to \bbR^+$ that maps program states $\sigma \in \Sigma$ to nonnegative numbers,
where we assume each state $\sigma$ contains a cost-accumulator component $\sigma.\alpha$.
If a program executes with initial state $\sigma$ to final state $\sigma'$,
then it holds that $\phi(\sigma) \ge (\sigma'.\alpha - \sigma.\alpha) + \phi(\sigma')$,
where $(\sigma'.\alpha - \sigma.\alpha)$ describes the accumulated cost from $\sigma$ to $\sigma'$.
%
%
The potential method also enables \emph{compositional} reasoning: if a statement $S_1$ executes from $\sigma$ to $\sigma'$ and a statement $S_2$ executes from $\sigma'$ to $\sigma''$, then we have $\phi(\sigma) \ge (\sigma'.\alpha - \sigma.\alpha) + \phi(\sigma')$ and $\phi(\sigma') \ge (\sigma''.\alpha - \sigma'.\alpha) + \phi(\sigma'')$; therefore, we derive $\phi(\sigma) \ge (\sigma''.\alpha - \sigma.\alpha) + \phi(\sigma'')$ for the sequential composition $S_1;S_2$.
For non-probabilistic programs, the initial potential provides an \emph{upper} bound on the accumulated cost.
%

This approach has been adapted to reason about expected costs of probabilistic programs~\cite{PLDI:NCH18,PLDI:WFG19}.
To derive upper bounds on the \emph{expected} accumulated cost of a program $S$ with initial state $\sigma$, one needs to take into consideration the \emph{distribution} of all possible executions.
More precisely, the potential function should satisfy the following property:
\begin{equation}\label{Eq:ExpectedPotentialInequality}
\phi(\sigma) \ge \expe_{\sigma' \sim \interp{S}(\sigma)}[ C(\sigma,\sigma') + \phi(\sigma')  ],
\end{equation}
where the notation $\expe_{x \sim \mu}[f(x)]$ represents the expected value of $f(x)$, where $x$ is drawn from the distribution $\mu$, $\interp{S}(\sigma)$ is the distribution over final states of executing $S$ from $\sigma$, and $C(\sigma,\sigma') \defeq \sigma'.\alpha - \sigma.\alpha$ is the execution cost from $\sigma$ to $\sigma'$.

\begin{example}\label{Exa:FstMomRecursiveRandomWalk}
  \cref{Fi:RecursiveRandomWalk} annotates the $\mathsf{rdwalk}$ function from \cref{Exa:RecursiveRandomWalk} with the derivation of an upper bound on the expected accumulated cost.
  %
  %
  %
  The annotations, taken together, define an expected-potential function $\phi : \Sigma \to \bbR^+$ where a program state $\sigma \in \Sigma$ consists of a program point and a valuation for program variables.
  To justify the upper bound $2(d-x)+4$ for the function $\mathsf{rdwalk}$, one has to show that the potential right before the $\itick{1}$ statement should be at least $1$.
  This property is established by \emph{backward} reasoning on the function body:
  \begin{itemize}
    \item For $\iinvoke{\mathsf{rdwalk}}$, we apply the ``induction hypothesis'' that the expected cost of the function $\mathsf{rdwalk}$ can be upper-bounded by $2(d-x)+4$.
    Adding the $1$ unit of potential need by the tick statement, we obtain $2(d-x)+5$ as the pre-annotation of the function call.
    
    \item For $\iassign{x}{\eadd{x}{t}}$, we substitute $x$ with $\eadd{x}{t}$ in the post-annotation of this statement to obtain the pre-annotation.
    
    \item For $\isample{t}{\kw{uniform}({-1},2)}$, because its post-annotation is $2(d-x-t)+5$, we compute its pre-annotation as
    \begin{center}\small      
    $
    \begin{aligned}
     & \expe_{t \sim \kw{uniform}({-1},2)}[2(d-x-t)+5] \\
    ={} & 2(d-x)+5 - 2 \cdot \expe_{t \sim \kw{uniform}({-1},2)}[t] \\
    ={} & 2(d-x)+5-2 \cdot \sfrac{1}{2} = 2(d-x)+4, 
    \end{aligned}
    $
    \end{center}
    which is exactly the upper bound we want to justify.
  \end{itemize}
\end{example}

\paragraph{Our approach}
%
In this paper, we focus on derivation of higher central moments.
Observing that a central moment $\expe[(X-\expe[X])^k]$ can be rewritten as a polynomial of raw moments $\expe[X],\cdots,\expe[X^k]$,
we reduce the problem of bounding central moments to reasoning about upper and lower bounds on raw moments.
For example, the variance can be written as $\vari[X] = \expe[X^2] - \expe^2[X]$, so it suffices to analyze the \emph{upper} bound of
the second moment $\expe[X^2]$ and the \emph{lower} bound on the square of the first moment $\expe^2[X]$.
For higher central moments, this approach requires both upper and lower bounds on higher raw moments.
For example, consider the fourth central moment of a \emph{nonnegative} random variable $X$:
\begin{small}
$
  \expe[(X-\expe[X])^4] = \expe[X^4]-4\expe[X^3]\expe[X]+6\expe[X^2]\expe^2[X] - 3\expe^4[X].
$
\end{small}
Deriving an upper bound on the fourth central moment requires lower bounds on the first ($\expe[X]$) and third ($\expe[X^3]$) raw moments.
%

We now sketch the development of \emph{moment semirings}.
We first consider only the upper bounds on higher moments of \emph{nonnegative} costs.
To do so, we extend the range of the expected-potential function $\phi$ to real-valued vectors $(\bbR^+)^{m+1}$, where $m \in \bbN$ is the degree of the target moment.
We update the potential inequality \eqref{Eq:ExpectedPotentialInequality} as follows:
\begin{equation}\label{Eq:OverApproxIntervalPotentialInequality}
  \phi(\sigma) \ge \expe_{\sigma' \sim \interp{S}(\sigma)} [ \many{\langle C(\sigma,\sigma')^k\rangle_{0\le k\le m}} \otimes \phi(\sigma')    ],
\end{equation}
where $\many{\tuple{v_k}_{0 \le k \le m}}$ denotes an $(m+1)$-dimensional vector, the order $\le$ on vectors is defined pointwise, and $\otimes$ is some \emph{composition} operator.
Recall that $\interp{S}(\sigma)$ denotes the distribution over final states of executing $S$ from $\sigma$, and $C(\sigma,\sigma')$ describes the cost for the execution from $\sigma$ to $\sigma'$.
Intuitively, for $\phi(\sigma) = \many{ \tuple{\phi(\sigma)_k}_{0 \le k  \le m}}$ and each $k$, the component $\phi(\sigma)_k$ is an upper bound on the $k$-th moment of the cost for the computation starting from $\sigma$.
The $0$-th moment is the \emph{termination probability} of the computation, and we assume it is always one for now.
We \emph{cannot} simply define $\otimes$ as pointwise addition because, for example, $(a+b)^2 \neq a^2+b^2$ in general.
If we think of $b$ as the cost for some probabilistic computation, and we prepend a constant cost $a$ to the computation, then by linearity of expectations, we have $\expe[(a+b)^2] = \expe[a^2+2ab+ b^2] = a^2 + 2 \cdot a \cdot \expe[b] + \expe[b^2]$, i.e., reasoning about the second moment requires us to keep track of the first moment.
Similarly, we should have
\begin{equation*}
  \phi(\sigma)_2 \ge \expe_{\sigma' \sim \interp{S}(\sigma)}[ C(\sigma,\sigma')^2 + 2 \cdot C(\sigma,\sigma') \cdot \phi(\sigma')_1 + \phi(\sigma')_2 ],
\end{equation*}
for the second-moment component,
where $\phi(\sigma')_1$ and $\phi(\sigma')_2$ denote $\expe[b]$ and $\expe[b^2]$, respectively.
Therefore, the composition operator $\otimes$ for second-moment analysis (i.e., $m=2$) should be defined as
\begin{equation}\label{Eq:SecondMomentComposition}
  \tuple{ 1, r_1,s_1 } \otimes \tuple{ 1, r_2,s_2 } \defeq \tuple{ 1, r_1 + r_2, s_1 + 2r_1r_2 + s_2 }.
\end{equation}

\begin{figure}
\centering
  \begin{small}
  \begin{pseudo}
    \kw{func} $\mathsf{rdwalk}$() \kw{begin} \\+
      ${\color{ACMDarkBlue} \{ \; \langle 1, \; 2(d-x) + 4, \; 4(d - x)^2 + 22(d - x) + 28 \rangle \; \} }$ \\
      \kw{if} $\cbin{<}{x}{d}$ \kw{then} \\+
        ${\color{ACMDarkBlue} \{\; \langle 1, \; 2(d - x) + 4, \; 4(d - x)^2 + 22(d -x ) + 28 \rangle \; \}  }$ \\
        $\isample{t}{ \kw{uniform}({-1},2)}$; \\
        ${\color{ACMDarkBlue} \{  \; \langle 1, \; 2(d-x-t)+5, \; 4(d-x-t)^2 + 26(d - x - t) + 37 \rangle \;  \} }$ \\
        $ \iassign{x}{  \eadd{x}{t} }$; \\
        ${\color{ACMDarkBlue} \{\; \langle 1, \; 2(d - x) + 5, \; 4(d  - x)^2 + 26(d-x) + 37 \rangle \; \} }$ \\
        \kw{call} $\mathsf{rdwalk}$; \\
        ${\color{ACMDarkBlue} \{  \; \tuple{1, \; 1, \; 1}  \; \} }$ \\
        \kw{tick}(1) \\
        ${\color{ACMDarkBlue} \{ \; \tuple{1, \; 0, \; 0} \; \} } $ \\-
      \kw{fi} \\-
    \kw{end}
  \end{pseudo}
  \end{small}
\caption{Derivation of an \emph{upper} bound on the first and second moment of the accumulated cost.}
\label{Fi:AnnotatedRecursiveRandomWalk}
\vspace{-10pt}
\end{figure}

\begin{example}\label{Exa:SndMomRecursiveRandomWalk}
  \cref{Fi:AnnotatedRecursiveRandomWalk} annotates the $\mathsf{rdwalk}$ function from \cref{Exa:RecursiveRandomWalk} with the derivation of an upper bound on both the first and second moment of the accumulated cost.
  %
  %
  To justify the first and second moment of the accumulated cost for the function $\mathsf{rdwalk}$,
  we again perform backward reasoning:
  \begin{itemize}
    \item For $\itick{1}$, it transforms a post-annotation $a$ by $\lambda a. (\tuple{1,1,1} \!\otimes\! a)$; thus, the pre-annotation is $\tuple{1,1,1} \otimes \tuple{1,0,0} = \tuple{1,1,1}$.
  
    \item For $\iinvoke{\mathsf{rdwalk}}$, we apply the ``induction hypothesis'', i.e., the upper bound shown on line 2. We use the $\otimes$ operator to compose the induction hypothesis with the post-annotation of this function call:
    \begin{center}\small
    $
    \begin{aligned}
      & \tuple{ 1, 2(d\!-\!x)\!+\!4 ,4(d\!-\!x)^2\!+\!22(d\!-\!x)\!+\!28 } \otimes \tuple{1,1,1} \\[-3pt]
      ={} & \tuple{1, 2(d\!-\!x)\!+\!5, (4(d\!-\!x)^2\!+\!22(d\!-\!x)\!+\!28)\!+\!2\!\cdot\! (2(d\!-\!x)\!+\!4) \!+\! 1  } \\[-3pt]
      ={} & \tuple{1, 2(d\!-\!x)\!+\!5, 4(d\!-\!x)^2\!+\!26(d\!-\!x)\!+\!37 }.
    \end{aligned}
    $
    \end{center}
    
    \item For $\iassign{x}{\eadd{x}{t}}$, we substitute $x$ with $x+t$ in the post-annotation of this statement to obtain the pre-annotation.
    
    \item For $\isample{t}{\kw{uniform}({-1},2)}$, because the post-annotation involves both $t$ and $t^2$, we compute from the definition of uniform distributions that 
    \[
    \expe_{t \sim \kw{uniform}({-1},2)}[t]=\sfrac{1}{2}, \quad \expe_{t \sim \kw{uniform}({-1},2)}[t^2] = 1.
    \]
    Then the upper bound on the second moment is derived as follows:
    \begin{center}\small
    $
    \begin{aligned}
      & \expe_{t \sim \kw{uniform}({-1},2)}[4(d\!-\!x\!-\!t)^2\!+\!26(d\!-\!x\!-\!t)\!+\!37] \\
      ={} & (4(d\!-\!x)^2\!+\!26(d\!-\!x)\!+\!37)  \!-\! (8(d\!-\!x)\!+\!26) \!\cdot\! \expe_{t \sim \kw{uniform}({-1},2)}[t ] \\[-3pt]
       & + 4 \!\cdot\! \expe_{t \sim \kw{uniform}({-1},2)}[t^2] \\
      ={} & 4(d\!-\!x)^2 \!+\!22(d\!-\!x) \!+\!28, 
    \end{aligned}
    $
    \end{center}
    which is the same as the desired upper bound on the second moment of the accumulated cost for the function $\mathsf{rdwalk}$. (See \cref{Fi:AnnotatedRecursiveRandomWalk}, line 2.)
  \end{itemize}
  %
\end{example}

We generalize the composition operator $\otimes$ to moments with arbitrarily high degrees, via a family of algebraic structures, which we name \emph{moment semirings} (see \cref{Se:MomentSemirings}). 
These semirings are \emph{algebraic} in the sense that they can be instantiated with any partially ordered semiring, not just $\bbR^+$.

\paragraph{Interval bounds}
Moment semirings not only provide a general method to analyze higher moments, but also enable reasoning about upper and lower bounds on moments \emph{simultaneously}.
The simultaneous treatment is also essential for analyzing programs with \emph{non-monotone} costs (see \cref{Se:InferenceRules}).

We instantiate moment semirings with the standard interval semiring $\calI = \{ [a,b] \mid a \le b \}$.
%
The algebraic approach allows us to systematically incorporate the interval-valued bounds, by \emph{reinterpreting} operations in \cref{Eq:SecondMomentComposition} under $\calI$:
\begin{center}\small
\vspace{-8pt}
$
\begin{aligned}
  & \tuple{ [1,1], [r_1^{\m{L}},r_1^{\m{U}}], [s_1^{\m{L}},s_1^{\m{U}}] } \otimes \tuple{ [1,1] , [r_2^{\m{L}},r_2^{\m{U}}], [s_2^{\m{L}},s_2^{\m{U}}] } \\
  {} \defeq {} & \langle [1,1], [r_1^{\m{L}},r_1^{\m{U}}] +_\calI  [r_2^{\m{L}},r_2^{\m{U}}], \\[-3pt]
  & \hphantom{\langle} [s_1^{\m{L}},s_2^{\m{U}}]+_\calI 2 \cdot ([r_1^{\m{L}},r_1^{\m{U}}] \cdot_\calI [r_2^{\m{L}},r_2^{\m{U}}]) +_\calI [s_2^{\m{L}},s_2^{\m{U}}] \rangle \\
   ={} & \langle [1,1], [r_1^{\m{L}}+r_2^{\m{L}},r_1^{\m{U}}+r_2^{\m{U}}], [s_1^{\m{L}}+2 \cdot \min S+s_2^{\m{L}}, s_1^{\m{U}}+2 \cdot \max S+s_2^{\m{U}}] \rangle , 
\end{aligned}
$
\end{center}
where $S \defeq \{ r_1^{\m{L}}r_2^{\m{L}}, r_1^{\m{L}}r_2^{\m{U}}, r_1^{\m{U}}r_2^{\m{L}}, r_1^{\m{U}}r_2^{\m{U}} \}$.
We then update the potential inequality \cref{Eq:OverApproxIntervalPotentialInequality} as follows:
\begin{equation*}
  \phi(\sigma) \sqsupseteq \expe_{\sigma' \sim \interp{S}(\sigma)} [ \many{\langle [C(\sigma,\sigma')^k,C(\sigma,\sigma')^k] \rangle_{0\le k \le m}} \otimes \phi(\sigma')    ],
\end{equation*}
where the order $\sqsubseteq$ is defined as pointwise interval inclusion.

\begin{example}\label{Exa:RandomWalkVar}
Suppose that the interval bound on the first moment of the accumulated cost of the $\m{rdwalk}$ function from \cref{Exa:RecursiveRandomWalk} is $[2(d-x),2(d-x)+4]$.
We can now derive the upper bound on the variance $\vari[\id{tick}] \le 22d+28$ shown in \cref{Fi:MotivatingExample}(b) (where we substitute $x$ with $0$ because the main function initializes $x$ to $0$ on line 5 in \cref{Fi:RecursiveRandomWalk}):
\begin{center}\small 
$
\begin{aligned}
  \vari[\id{tick}] & = \expe[\id{tick}^2] - \expe^2[\id{tick}] \\[-3pt]
  & \le (\text{upper bnd. on $\expe[\id{tick}^2]$}) \!- \! (\text{lower bnd. on $\expe[\id{tick}]$})^2 \\[-3pt]
  & = (4d^2+22d+28) - (2d)^2 = 22d+28.
\end{aligned}
$
\end{center}

In \cref{Se:TailBound}, we describe how we use moment bounds to derive the tail bounds shown in \cref{Fi:MotivatingExample}(c).
\end{example}

\subsection{Two Major Challenges}
\label{Se:SoundnessCriteria}

\paragraph{Interprocedural reasoning}
Recall that in the derivation of \cref{Fi:AnnotatedRecursiveRandomWalk}, we use the $\otimes$ operator to compose the upper bounds on moments for $\iinvoke{\m{rdwalk}}$ and its post-annotation $\tuple{1,1,1}$.
However, this approach does \emph{not} work in general, because the post-annotation might be symbolic (e.g., $\tuple{1,x,x^2}$) and the callee might mutate referenced program variables (e.g., $x$).
One workaround is to derive a \emph{pre}-annotation for each possible \emph{post}-annotation of a recursive function, i.e., the moment annotations for a recursive function is \emph{polymorphic}.
This workaround would \emph{not} be effective for non-tail-recursive functions:
for example, we need to reason about the $\m{rdwalk}$ function in \cref{Fi:AnnotatedRecursiveRandomWalk} with \emph{infinitely} many post-annotations $\tuple{1,0,0}$, $\tuple{1,1,1}$, $\tuple{1,2,4}$, \ldots, i.e., $\tuple{1,i,i^2}$ for all $i \in \bbZ^+$.

Our solution to \emph{moment-polymorphic recursion} is to
introduce a \emph{combination} operator $\oplus$ in a way that
if $\phi_1$ and $\phi_2$ are two expected-potential functions, then
\[
\phi_1(\sigma) \oplus \phi_2(\sigma) \!\ge\! \expe_{\sigma'  \!\sim\! \interp{S}(\sigma)}[ \many{ \tuple{C(\sigma,\sigma')^k}_{0 \!\le\! k \!\le\! m} } \otimes (\phi_1(\sigma') \oplus \phi_2(\sigma')) ].
\]
We then use the $\oplus$ operator to derive a \emph{frame} rule:
\begin{mathpar}\small
\inferrule
{ {\color{ACMDarkBlue} \{\; Q_1 \; \} } ~ S ~ {\color{ACMDarkBlue} \{ \; Q_1' \; \} } \\
  {\color{ACMDarkBlue} \{\; Q_2 \; \} } ~ S ~ {\color{ACMDarkBlue} \{ \; Q_2' \; \} }
}
{ {\color{ACMDarkBlue} \{\; Q_1 \oplus Q_2 \; \} } ~ S ~ {\color{ACMDarkBlue} \{ \; Q_1' \oplus Q_2' \; \} } }
\end{mathpar}
We define $\oplus$ as pointwise addition, i.e., for second moments,
\begin{equation}\label{Eq:SndMomentCombine}
\tuple{p_1,r_1,s_1} \oplus \tuple{p_2,r_2,s_2} \defeq \tuple{p_1+p_2,r_1+r_2,s_1+s_2},
\end{equation}
and because the $0$-th-moment (i.e., termination-probability) component is no longer guaranteed to be one, we redefine $\otimes$ to consider the termination probabilities:
\begin{equation}\label{Eq:SndMomentExtend}
\tuple{p_1,r_1,s_1} \otimes \tuple{p_2,r_2,s_2} \defeq \tuple{p_1 p_2, p_2r_1+p_1r_2, p_2s_1+2r_1r_2 + p_1s_2}.
\end{equation}

\revision{3}
\begin{changebar}
\begin{remark}
  As we will show in \cref{Se:MomentSemirings},
  the composition operator $\otimes$ and combination operator $\oplus$ form a \emph{moment semiring};
  consequently, we can use \emph{algebraic} properties of semirings (e.g., distributivity) to aid higher-moment analysis.
  For example, a vector $\tuple{0,r_1,s_1}$ whose termination-probability component is zero does not seem to make sense, because moments with respect to a zero distribution should also be zero.
  However, by distributivity, we have
  \begin{center}\small
  $
  \begin{aligned}
  & \tuple{1,r_3,s_3} \otimes \tuple{1,r_1+r_2,s_1+s_2} \\
  ={} & \tuple{1,r_3,s_3} \otimes (\tuple{0,r_1,s_1} \oplus \tuple{1,r_2,s_2}) \\
  ={} & (\tuple{1,r_3,s_3} \otimes \tuple{0,r_1,s_1}) \oplus ( \tuple{1,r_3,s_3} \oplus \tuple{1,r_2,s_2}).
  \end{aligned}
  $
  \end{center}
  If we think of $\tuple{1,r_1+r_2,s_1+s_2}$ as a post-annotation of a computation whose moments are bounded by $\tuple{1,r_3,s_3}$,
  the equation above indicates that we can use $\oplus$ to \emph{decompose} the post-annotation into subparts, and then reason about each subpart separately.
  This fact inspires us to develop a decomposition technique for moment-polymorphic recursion.
\end{remark}
\end{changebar}

\begin{example}\label{Exa:MomentPolymorphism}
  With the $\oplus$ operator and the frame rule, we only need to analyze the $\m{rdwalk}$ function from \cref{Exa:RecursiveRandomWalk} with \emph{three} post-annotations: $\tuple{1,0,0}$, $\tuple{0,1,1}$, and $\tuple{0,0,2}$, which form a kind of ``elimination sequence.''
  \revision{4}
  \begin{changebar}
  We construct this sequence in an \emph{on-demand} manner;
  the first post-annotation is the identity element $\tuple{1,0,0}$ of the moment semiring.
  \end{changebar}
  
  For post-annotation $\tuple{1,0,0}$, as shown in \cref{Fi:AnnotatedRecursiveRandomWalk}, we need to know the moment bound for $\m{rdwalk}$ with the post-annotation $\tuple{1,1,1}$.
  \revision{4}
  \begin{changebar}
  Instead of reanalyzing $\m{rdwalk}$ with the post-annotation $\tuple{1,1,1}$,
  we use the $\oplus$ operator to compute the ``difference'' between it and the previous post-annotation $\tuple{1,0,0}$.
  \end{changebar}
  Observing that $\tuple{1,1,1} = \tuple{1,0,0} \oplus \tuple{0,1,1}$, we now analyze $\m{rdwalk}$ with $\tuple{0,1,1}$ as the post-annotation:
  \begin{small}
  \begin{pseudo}
    $\iinvoke{\m{rdwalk}}$; {\color{ACMDarkBlue} $\{ \;  \tuple{0, \; 1, \; 3} \;  \}$   } {\color{gray}\# $=\tuple{1, \; 1, \; 1} \otimes \tuple{0, \; 1, \; 1} $}  \\
    $\itick{1}$ {\color{ACMDarkBlue} $\{ \; \tuple{0,\; 1, \;1} \; \}$ }
  \end{pseudo}
  \end{small}
  Again, because $\tuple{0,1,3} = \tuple{0,1,1} \oplus \tuple{0,0,2}$, we need to further analyze $\m{rdwalk}$ with $\tuple{0,0,2}$ as the post-annotation:
  \begin{small}
  \begin{pseudo}
    $\iinvoke{\m{rdwalk}}$; {\color{ACMDarkBlue} $\{ \;  \tuple{0, \; 0, \; 2} \;  \}$   } {\color{gray}\# $=\tuple{1, \; 1, \; 1} \otimes \tuple{0, \; 0, \; 2}$} \\
    $\itick{1}$ {\color{ACMDarkBlue} $\{ \; \tuple{0,\; 0, \;2} \; \}$ }
  \end{pseudo}
  \end{small}
  With the post-annotation $\tuple{0,0,2}$, we can now reason \emph{monomorphically} without analyzing any new post-annotation!
  We can perform a succession of reasoning steps similar to what we have done in \cref{Exa:FstMomRecursiveRandomWalk} to justify the following bounds (``unwinding'' the elimination sequence):
  \revision{4}
  \begin{changebar}
  \begin{itemize}
    \item ${\color{ACMDarkBlue} \{  \tuple{ 0,  0 , 2 }  \} } ~ {\m{rdwalk}} ~ { \color{ACMDarkBlue} \{  \tuple{0,  0,  2}  \} }$:
    Directly by backward reasoning with the post-annotation $\tuple{0,0,2}$.
    
    \item ${\color{ACMDarkBlue} \{  \tuple{ 0,  1 , 4(d-x) + 9  }  \} } ~ {\m{rdwalk}} ~ { \color{ACMDarkBlue} \{  \tuple{0,  1, 1}  \} }$:
    To analyze the recursive call with post-annotation $\tuple{0,1,3}$, we use the frame rule with the post-call-site annotation $\tuple{0,0,2}$ to derive $\tuple{0,1,4(d-x)+11}$ as the pre-annotation:
    \begin{small}
    \begin{pseudo}
    {\color{ACMDarkBlue} $\{ \; \tuple{0, \;1,\; 4(d-x)+ 11} \; \}$ } {\color{gray}\# $=\tuple{0, \; 1, \; 4(d-x)+9} \oplus \tuple{0, \; 0, \; 2}$} \\
    $\iinvoke{\m{rdwalk}}$; \\
    {\color{ACMDarkBlue} $\{ \;  \tuple{0, \; 1, \; 3} \;  \}$   } {\color{gray}\# $=\tuple{0, \; 1, \; 1} \oplus \tuple{0, \; 0, \; 2}$} \\
    \end{pseudo}
    \end{small}
    
    \item ${\color{ACMDarkBlue} \{  \tuple{ 1,  2(d\!-\!x)\!+\!4 , 4(d\!-\!x)^2 \!+\! 22(d\!-\!x)\!+\!28  }  \} } ~ {\m{rdwalk}} ~{ \color{ACMDarkBlue} \{  \tuple{1,  0, 0}  \} }$: To analyze the recursive call with post-annotation $\tuple{1,1,1}$, we use the frame rule with the post-call-site annotation $\tuple{0,1,1}$ to derive $\tuple{1,2(d-x)+5,4(d-x)^2+26(d-x)+37}$ as the pre-annotation:
    \begin{small}
    \begin{pseudo}
    {\color{ACMDarkBlue} $\{ \; \tuple{1, 2(d\!-\!x)\!+\!5, 4(d\!-\!x)^2\!+\! 26(d\!-\!x)\!+\!37} \; \}$ } \\
    {\color{gray}\# $=\tuple{1, 2(d\!-\!x)\!+\!4, 4(d\!-\!x)^2\!+\!22(d\!-\!x)\!+\!28} \oplus \tuple{0, 1,  4(d\!-\!x)\!+\!9}$} \\
    $\iinvoke{\m{rdwalk}}$; \\
    {\color{ACMDarkBlue} $\{ \;  \tuple{1, \; 1, \; 1} \;  \}$   } {\color{gray}\# $=\tuple{1, \; 0, \; 0} \oplus \tuple{0, \; 1, \; 1}$} \\
    \end{pseudo}
    \end{small}
  \end{itemize}
  \end{changebar}
\end{example}

In \cref{Se:InferenceRules}, we present an automatic inference system for the expected-potential method that is extended with interval-valued bounds on higher moments, with support for moment-polymorphic recursion.

\paragraph{Soundness of the analysis}
Unlike the classic potential method, the expected-potential method is \emph{not} always sound when reasoning about the moments for cost accumulators in probabilistic programs.

\begin{figure}
  \centering
  \begin{small}
  \begin{pseudo}
    \kw{func} $\m{geo}$() \kw{begin} {\color{ACMDarkBlue} $\{ \; \tuple{1, \; 2^x } \;\}$ } \\+
    $\iassign{x}{\eadd{x}{1}}$; {\color{ACMDarkBlue} $\{ \; \tuple{1, \; 2^{x-1} } \; \}$ } \\
    {\color{gray}\# expected-potential method for \emph{lower} bounds: } \\
    {\color{gray}\# $2^{x-1} < \sfrac{1}{2} \cdot (2^x+1) + \sfrac{1}{2} \cdot 0$ } \\
    \kw{if} $\kw{prob}(\sfrac{1}{2})$ \kw{then} {\color{ACMDarkBlue} $\{ \; \tuple{1, \; 2^{x} +1  } \;\}$ } \\+
      $\itick{1}$; {\color{ACMDarkBlue} $\{ \; \tuple{1, \; 2^x} \; \}$ }  \\
      $\iinvoke{\m{geo}}$ {\color{ACMDarkBlue} $\{ \; \tuple{1, \; 0} \}$ } \\-
    \kw{fi} \\-
    \kw{end}
  \end{pseudo}
  \end{small}
  \caption{A purely probabilistic loop with annotations for a \emph{lower} bound on the first moment of the accumulated cost.}
  \label{Fi:ProbLoop}
  \vspace{-10pt}
\end{figure}

\begin{counterexample}\label{Exa:UnsoundPotentialFunction}
  Consider the program in \cref{Fi:ProbLoop} that describes a purely probabilistic loop that exits the loop with probability $\sfrac{1}{2}$ in each iteration.
  The expected accumulated cost of the program should be \emph{one}~\cite{POPL:HKG20}.
  However, the annotations in \cref{Fi:ProbLoop} justify a potential function $2^x$ as a \emph{lower} bound on the expected accumulated cost, no matter what value $x$ has at the beginning, which is apparently \emph{unsound}.  
\end{counterexample}

Why does the expected-potential method fail in this case?
\revision{9}
\begin{changebar}
The short answer is that dualization only works for some problems:
upper-bounding the sum of nonnegative ticks is equivalent to
lower-bounding the sum of nonpositive ticks;
lower-bounding the sum of nonnegative ticks---the issue in
\cref{Fi:ProbLoop}---is equivalent to upper-bounding the
sum of nonpositive ticks; 
however, the two kinds of problems are \emph{inherently different}~\cite{POPL:HKG20}.
\end{changebar}
\Omit{
\citet{POPL:HKG20} presented a thorough study of the difficulty of lower-bound inference for probabilistic programs (which, also implies the difficulty of handling non-monotone costs, because upper-bounding the sum of nonnegative ticks is \emph{equivalent} to lower-bounding the sum of nonpositive ticks).
}
Intuitively, the classic potential method for bounding the costs of non-probabilistic programs is a \emph{partial}-correctness\Omit{ reasoning} method, i.e., derived upper/lower bounds are sound if the analyzed program terminates~\cite{SP:NDF17}.
With probabilistic programs, many programs do not terminate \emph{definitely}, but only \emph{almost surely}, i.e., they terminate with probability one, but have some execution traces that are non-terminating.
The programs in \cref{Fi:RecursiveRandomWalk,Fi:ProbLoop} are both almost-surely terminating.
For the expected-potential method, the potential at a program state can be seen as an \emph{average} of potentials needed for all possible computations that continue from the state.
If the program state can lead to a non-terminating execution trace, the potential associated with that trace might be problematic, and as a consequence, the expected-potential method might fail.

Recent research~\cite{POPL:HKG20,PLDI:WFG19,CAV:BEF16,misc:SO19} has employed the \emph{Optional Stopping Theorem} (OST) from probability theory to address this soundness issue.
The classic OST provides a collection of \emph{sufficient} conditions for reasoning about expected gain \emph{upon termination} of stochastic processes, where the expected gain at any time is \emph{invariant}.
By constructing a stochastic process for executions of probabilistic programs and setting the expected-potential function as the invariant, one can apply the OST to justify the soundness of the expected-potential function.
%
%
In a companion paper~\cite{Companion}, we study and propose an extension to the classic OST with a \emph{new} sufficient condition that is suitable for reasoning about higher moments;
in this work, we prove the soundness of our central-moment inference for programs that satisfy this condition, and develop an algorithm to check this condition automatically (see \cref{Se:SemanticOptionalStopping}).
%




\begin{figure}
\centering
\begin{small}
\begin{align*}
  S & \Coloneqq \iskip \mid \itick{c} \mid \iassign{x}{E} \mid \isample{x}{D} \mid \iinvoke{f} \mid \iloop{L}{S}  \\
  & \mid \iprob{p}{S_1}{S_2} \mid \icond{L}{S_1}{S_2}  \mid S_1 ; S_2 \\
  L & \Coloneqq \ctop \mid \cneg{L} \mid \cconj{L_1}{L_2} \mid \cle{E_1}{E_2} \\
  E & \Coloneqq x \mid c \mid \eadd{E_1}{E_2} \mid \emul{E_1}{E_2} \\
  D & \Coloneqq \kw{uniform}(a,b) \mid \cdots
\end{align*}
\end{small}
\vspace{-10pt}
\caption{Syntax of \lang{}, where $p \in [0,1]$, $a,b,c \in \bbR$, $a < b$, $x\in \mathsf{VID}$ is a variable, and $f \in \mathsf{FID}$ is a function identifier.}
\label{Fi:Syntax}
\vspace{-10pt}
\end{figure}

\section{Derivation System for Higher Moments}
\label{Se:DerivationSystem}

In this section, we describe the inference system used by our analysis.
We first present a probabilistic programming language (\cref{Se:ProbLang}).
We then introduce \emph{moment semirings} to compose higher moments for a cost accumulator from two computations (\cref{Se:MomentSemirings}).
We use moment semirings to develop our derivation system, which is presented as a declarative program logic (\cref{Se:InferenceRules}).
Finally, we sketch how we reduce the inference of a derivation to LP solving (\cref{Se:Automation}).

\subsection{A Probabilistic Programming Language}
\label{Se:ProbLang}

This paper uses an imperative arithmetic probabilistic programming language \lang{} that supports general recursion and continuous distributions, where program variables are real-valued.
We use the following notational conventions.
Natural numbers $\bbN$ \emph{exclude} $0$, i.e., $\bbN \defeq \{ 1,2,3,\cdots\} \subseteq \bbZ^+ \defeq \{0,1,2,\cdots\}$.
%
%
%
The \emph{Iverson brackets} $[\cdot]$ are defined by $[\varphi]=1$ if $\varphi$ is true and otherwise $[\varphi]=0$.
We denote updating an existing binding of $x$ in a finite map $f$ to $v$ by $f[x \mapsto v]$.
We will also use the following standard notions from probability theory: $\sigma$-algebras, measurable spaces, measurable functions, random variables, probability measures, and expectations.
\iflong
\Cref{Se:AppendixProbTheory} provides a review of those notions.
\else
We include a review of those notions in the technical report~\cite{Techreport}.
\fi

\cref{Fi:Syntax} presents the syntax of \lang{}, where the metavariables $S$, $L$, $E$, and $D$ stand for statements, conditions, expressions, and distributions, respectively.
Each distribution $D$ is associated with a probability measure $\mu_D \in \bbD(\bbR)$.
We write $\bbD(X)$ for the collection of all probability measures on the measurable space $X$.
%
For example, $\kw{uniform}(a,b)$ describes a uniform distribution on the interval $[a,b]$, and its corresponding probability measure is the integration of its density function $\mu_{\kw{uniform}(a,b)}(O) \defeq \int_O \frac{[a \le x \le b]}{b - a} dx $.
The statement ``$\isample{x}{D}$'' is a \emph{random-sampling} assignment, which draws from the distribution $\mu_D$ to obtain a sample value and then assigns it to $x$.
The statement ``$\iprob{p}{S_1}{S_2}$'' is a \emph{probabilistic-branching} statement, which executes $S_1$ with probability $p$, or $S_2$ with probability $(1-p)$.

The statement ``$\iinvoke{f}$'' makes a (possibly recursive) call to the function with identifier $f \in \mathsf{FID}$.
In this paper, we assume that the functions only manipulate states that consist of global program variables.
The statement $\itick{c}$, where $c \in \bbR$ is a constant, is used to define the \emph{cost model}.
It adds $c$ to an anonymous global cost accumulator.
Note that our implementation supports local variables, function parameters, return statements, as well as accumulation of non-constant costs;
the restrictions imposed here are not essential, and are introduced solely to simplify the presentation.

We use a pair $\tuple{\scrD, S_{\mathsf{main}}}$ to represent an \lang{} program, where $\scrD$ is a finite map from function identifiers to their bodies and $S_{\mathsf{main}}$ is the body of the main function.
We present an operational semantics for \lang{} in \cref{Se:ExpectedCostBoundAnalysis}. 

\subsection{Moment Semirings}
\label{Se:MomentSemirings}

As discussed in \cref{Se:ExpectedPotential}, we want to design a \emph{composition} operation $\otimes$ and a \emph{combination} operation $\oplus$ to compose and combine higher moments of accumulated costs such that
\begin{align*}
  \phi(\sigma) & \!\sqsupseteq\! \expe_{\sigma' \!\sim\! \interp{S}(\sigma)} \![ \many{ \tuple{ C(\sigma,\sigma')^k}_{0\!\le\! k \!\le\! m}} \!\otimes\! \phi(\sigma')    ], \\
  \phi_1(\sigma) \!\oplus\! \phi_2(\sigma) & \!\sqsupseteq\! \expe_{\sigma' \!\sim\! \interp{S}(\sigma)} \! [ \many{ \tuple{C(\sigma,\sigma')^k}_{0 \!\le\! k \!\le\! m} } \!\otimes\! (\phi_1(\sigma') \!\oplus\! \phi_2(\sigma')) ],
\end{align*}
where the expected-potential functions $\phi,\phi_1,\phi_2$ map program states to interval-valued vectors, $C(\sigma,\sigma')$ is the cost for the computation from $\sigma$ to $\sigma'$, and $m$ is the degree of the target moment.
In \cref{Eq:SndMomentExtend,Eq:SndMomentCombine}, we gave a definition of $\otimes$ and $\oplus$ suitable for first and second moments, respectively.
In this section, we generalize them to reason about upper and lower bounds of higher moments.
Our approach is inspired by the work of~\citet{EMNLP:LE09}, which develops a method to ``lift'' techniques for first moments to those for second moments.
Instead of restricting the elements of semirings to be vectors of numbers, we propose \emph{algebraic} moment semirings that can also be instantiated with vectors of intervals, which we need for the interval-bound analysis that was demonstrated in \cref{Se:ExpectedPotential}.

\begin{definition}\label{De:MomSemiring}
  The $m$-th order \emph{moment semiring} $\calM_\calR^{(m)} = (|\calR|^{m+1},\oplus, \otimes, \azero, \aone)$ is parametrized by a partially ordered semiring $\calR = (|\calR|,{\le}, {+},{\cdot},0 ,1 )$,
  where
  \begin{align}
    \many{\tuple{u_k}_{0 \le k \le m} } \!\oplus\! \many{\tuple{v_k}_{0 \le k \le m}} & \defeq \many{\tuple{u_k + v_k}_{0 \le k \le m}}, \\
    \many{\tuple{u_k}_{0 \le k \le m}} \!\otimes\! \many{\tuple{v_k}_{0 \le k \le m}} & \defeq \many{ \tuple{ \textstyle \sum_{i=0}^k  \binom{k}{i} \!\times\! (u_i \cdot v_{k-i}) }_{0 \le k \le m} }, \label{Eq:MomentGenerating} 
  \end{align}
  $\binom{k}{i}$ is the binomial coefficient; the scalar product $n \times u$ is an abbreviation for $\sum_{i=1}^n u$, for $n \in \bbZ^+, u \in \calR$; $\azero \defeq \tuple{0,0,\cdots,0}$; and $\aone \defeq \tuple{1 , 0 ,\cdots, 0 }$.
  We define the partial order $\aord$ as the pointwise extension of the partial order $\le$ on $\calR$.
\end{definition}

\revision{9}
\begin{changebar}
Intuitively, the definition of $\otimes$ in \cref{Eq:MomentGenerating} can be seen as the multiplication of two moment-generating functions for distributions with moments $\many{\tuple{u_k}_{0 \le k \le m}}$ and $\many{\tuple{v_k}_{0 \le k \le m}}$, respectively.
\end{changebar}
We prove a composition property for moment semirings.

\begin{lemma}\label{Lem:EisnerBasicProperty}
  For all $u,v \in \calR$, it holds that
  \[
  \many{\tuple{(u+v)^k}_{0 \le k \le m}} = \many{\tuple{u^k}_{0 \le k \le m}} \otimes \many{\tuple{v^k}_{0 \le k \le m}},
  \]
  where $u^n$ is an abbreviation for $\prod_{i=1}^n u$, for $n \in \bbZ^+, u \in \calR$.
\end{lemma}

\subsection{Inference Rules}
\label{Se:InferenceRules}

%
%
We present the derivation system as a declarative program logic that uses moment semirings to enable compositional reasoning and moment-polymorphic recursion.

\paragraph{Interval-valued moment semirings}
Our derivation system infers upper and lower bounds simultaneously, rather than separately,
which is essential for \emph{non-monotone} costs.
Consider a program ``$\itick{-1} ; S$'' and suppose that we have $\tuple{1,2,5}$ and $\tuple{1,{-2},5}$ as the upper and lower bound on the first two moments of the cost for $S$, respectively.
If we only use the upper bound, we derive $\tuple{1,{-1},1} \otimes \tuple{1,2,5} = \tuple{1,1,2}$, which is \emph{not} an upper bound on the moments of the cost for the program;
if the \emph{actual} moments of the cost for $S$ are $\tuple{1,0,5}$,
then the \emph{actual} moments of the cost for ``$\itick{-1}; S$'' are $\tuple{1,{-1},1} \otimes \tuple{1,0,5} = \tuple{1,{-1}, 4} \not\le \tuple{1,1,2}$.
%
%
Thus, in the analysis, we instantiate moment semirings with the interval domain $\calI$.
For the program ``$\itick{-1}; S$'', its interval-valued bound on the first two moments is
$
  \tuple{[1,1], [{-1},{-1}],[1,1]} \otimes \tuple{[1,1], [{-2},2], [5,5]} 
  =  \tuple{[1,1], [{-3},1], [ 2, 10]}. 
$

\paragraph{Template-based expected-potential functions}
The basic approach to automated inference using potential functions is to introduce a \emph{template} for the expected-potential functions.
Let us fix $m \in \bbN$ as the degree of the target moment.
Because we use $\calM_\calI^{(m)}$-valued expected-potential functions whose range is vectors of intervals, the templates are vectors of intervals whose ends are represented \emph{symbolically}.
In this paper, we represent the ends of intervals by \emph{polynomials} in $\bbR[\mathsf{VID}]$ over program variables.
%

More formally, we lift the interval semiring $\calI$ to a \emph{symbolic} interval semiring $\calP\calI$ by representing the ends of the $k$-th interval by polynomials in $\bbR_{kd}[\mathsf{VID}]$ up to degree $kd$ for some fixed $d \in \bbN$.
%
%
%
Let $\calM_{\calP\calI}^{(m)}$ be the $m$-th order moment semiring instantiated with the symbolic interval semiring.
Then the potential annotation is represented as $Q = \many{\tuple{[L_k,U_k]}_{0 \le k \le m}} \in \calM_{\calP\calI}^{(m)}$, where $L_k$'s and $U_k$'s are polynomials in $\bbR_{kd}[\mathsf{VID}]$.
$Q$ defines an $\calM_\calI^{(m)}$-valued expected-potential function $\phi_Q(\sigma) \defeq \many{ \tuple{[ \sigma(L_k),\sigma(U_k) ]}_{0 \le k \le m} }$, where $\sigma$ is a program state, and $\sigma(L_k)$ and $\sigma(U_k)$ are $L_k$ and $U_k$ evaluated over $\sigma$, respectively.

\paragraph{Inference rules}
We formalize our derivation system for moment analysis in a Hoare-logic style.
The judgment has the form $\Delta \vdash_h \{\Gamma;Q\}~S~\{\Gamma';Q'\}$, where
$S$ is a statement, $\{\Gamma;Q\}$ is a precondition, $\{\Gamma';Q'\}$ is a postcondition,
$\Delta = \tuple{\Delta_k}_{0 \le k \le m}$ is a context of function specifications,
and $h \in \bbZ^+$ specifies some restrictions put on $Q,Q'$ that we will explain later.
The \emph{logical context} $\Gamma\!:\!(\mathsf{VID} \!\to\! \bbR) \!\to\! \set{\top,\bot}$ is a predicate\Omit{on program states} that describes reachable states at a program point. 
%
%
The \emph{potential annotation} $Q \in \calM_{\calP\calI}^{(m)}$ specifies a map from program states to the moment semiring that is used to define interval-valued expected-potential functions.
The semantics of the triple $\{\cdot;Q\}~S~\{\cdot;Q'\}$ is that if the rest of the computation after executing $S$ has its moments of the accumulated cost bounded by $\phi_{Q'}$, then the whole computation has its moments of the accumulated cost bounded by $\phi_Q$.
The parameter $h$ restricts all $i$-th-moment components in $Q,Q'$, such that $i < h$, to be $[0,0]$.
We call such potential annotations \emph{$h$-restricted};
\revision{6}
\begin{changebar}
this construction is motivated by an observation from \cref{Exa:MomentPolymorphism}, where we illustrated the benefits of carrying out interprocedural analysis using an ``elimination sequence'' of annotations for recursive function calls, where the successive annotations have a greater number of zeros, filling from the left.
\end{changebar}
\emph{Function specifications} are valid pairs of pre- and post-conditions for all declared functions in a program.
For each $k$, such that $0 \!\le\! k \!\le\! m$, and each function $f$,
a valid specification $(\Gamma;Q,\Gamma';Q') \in \Delta_k(f)$ is justified by the judgment $\Delta \vdash_k \{\Gamma;Q\}~\scrD(f)~\{\Gamma';Q'\}$, where $\scrD(f)$ is the function body of $f$, and $Q,Q'$ are $k$-restricted.
%
%
To perform context-sensitive\Omit{interprocedural} analysis, a function can have multiple specifications.

\begin{figure}
\begin{mathpar}\small
  \Rule{Q-Tick}{    Q = \many{\tuple{[c^k,c^k]}_{0 \le k \le m}} \otimes Q' }{ \Delta \vdash_h \{ \Gamma;Q \}~\itick{c}~\{ \Gamma;Q' \} }
  \and
  \Rule{Q-Sample}{ \Gamma =\Forall{x \in \mathrm{supp}(\mu_D)} \Gamma' \\\\  Q = \expe_{x \sim \mu_D}[Q'] }{ \Delta \vdash_h \{  \Gamma;Q   \}~\isample{x}{D}~\{ \Gamma' ;Q' \} }
  \\
  \Rule{Q-Loop}{ \Delta \vdash_h \{ \Gamma \wedge L;Q \}~S~\{ \Gamma;Q \}  }{ \Delta \!\vdash_h\! \{ \Gamma;Q \} ~\iloop{L}{S}~ \{ \Gamma \!\wedge\! \neg L; Q \} }
  \enskip
  \Rule{Q-Call-Mono}
  { (\Gamma;Q , \Gamma';Q') \in \Delta_m(f)
  }
  { \Delta \!\vdash_m\! \{ \Gamma; Q \} ~\iinvoke{f}~ \{ \Gamma'; Q' \} }
  \\
  \Rule{Q-Seq}{ \Delta \vdash_h \{\Gamma;Q \}~S_1~\{\Gamma';Q' \} \\\\ \Delta \vdash_h \{ \Gamma';Q' \}~S_2~\{\Gamma'';Q''\}  }{ \Delta \vdash_h \{ \Gamma;Q \}~S_1;S_2~\{\Gamma'';Q''\} }
  \enskip
  \Rule{Q-Call-Poly}{ h < m \\  \Delta_h(f) = (\Gamma;Q_1,\Gamma';Q_1') \\\\ \Delta \vdash_{h+1} \{ \Gamma;Q_2\}~\scrD(f)~\{ \Gamma';Q_2' \}   }{ \Delta \vdash_h \{ \Gamma; Q_1 \oplus Q_2 \}~\iinvoke{f}~\{ \Gamma'; Q_1' \oplus Q_2' \} }
  \\
  \Rule{Q-Prob}{  \Delta \vdash_h \{ \Gamma;Q_1 \}~S_1~\{ \Gamma';Q' \} \\ \Delta \vdash_h \{ \Gamma;Q_2 \} ~S_2~\{ \Gamma';Q' \} \\ Q = {P \oplus R} \\ P = \tuple{[p,p],[0,0],\cdots,[0,0]} \otimes Q_1 \\ R = \tuple{[1-p,1-p],[0,0],\cdots,[0,0]} \otimes Q_2 }{ \Delta \vdash_h \{ \Gamma;Q \} ~\iprob{p}{S_1}{S_2}~\{ \Gamma'; Q' \} }
  \and
\iflong
  \Rule{Q-Weaken}{ \Delta \!\vdash_h\! \{ \Gamma_0;Q_0 \}~S~\{ \Gamma_0';Q_0' \} \enskip \Gamma \models \Gamma_0 \enskip \Gamma_0' \models \Gamma' \enskip \Gamma \models Q \sqsupseteq Q_0 \enskip \Gamma_0' \models Q_0' \sqsupseteq Q'  }{ \Delta \vdash_h \{\Gamma;Q \}~S~\{\Gamma';Q'\} }
\fi
\end{mathpar}
\caption{Selected inference rules of the derivation system.}
\label{Fi:InferenceRules}
\vspace{-8pt}
\end{figure}

\cref{Fi:InferenceRules} presents some of the inference rules.
The rule \textsc{(Q-Tick)} is the only rule that deals with costs in a program.
%
%
To accumulate the moments of the cost, we use the $\otimes$ operation in the moment semiring $\calM_{\calP\calI}^{(m)}$.
The rule \textsc{(Q-Sample)} accounts for sampling statements.
Because ``$x \!\sim\! D$'' randomly assigns a value to $x$ in the support of distribution $D$, we quantify $x$ out universally from the logical context.
To compute $Q = \expe_{x \sim \mu_D}[Q']$, where $x$ is drawn from distribution $D$, we assume the moments for $D$ are well-defined and computable,
and substitute $x^i$, $i \in \bbN$ with the corresponding moments in $Q'$.
We make this assumption because every component of $Q'$ is a polynomial over program variables.
For example, if $D = \kw{uniform}({-1},2)$, we know the following facts
\[
\expe_{x \sim \mu_D}[x^0] \!=\! 1, \expe_{x \sim \mu_D}[x^1] \!=\! \sfrac{1}{2},\expe_{x \sim \mu_D}[x^2] \!=\! 1,\expe_{x \sim \mu_D}[x^3] \!=\! \sfrac{5}{4}.
\]
Then for $Q' = \langle [1,1], [1+x^2,xy^2+x^3y] \rangle$, by the linearity of expectations, we compute $Q=\expe_{x \sim \mu_D}[Q']$ as follows:
\begin{center}\small
\vspace{-8pt}
$
\begin{aligned}
  \expe_{x \sim \mu_D}[Q'] & = \langle [1,1], [\expe_{x \sim \mu_D}[1+x^2], \expe_{x \sim \mu_D}[xy^2+x^3y] \rangle \\[-3pt]
  & = \langle [1,1], [1 + \expe_{x \sim \mu_D}[x^2], y^2 \expe_{x \sim \mu_D}[x] + y \expe_{x \sim \mu_D}[x^3]] \rangle \\[-3pt]
  & = \langle [1,1], [2, \sfrac{1}{2} \cdot  y^2+\sfrac{5}{4} \cdot y] \rangle. 
\end{aligned}
$
\end{center}
The other probabilistic rule \textsc{(Q-Prob)} deals with probabilistic branching.
Intuitively, if the moments of the
execution of $S_1$ and $S_2$ are $q_1$ and $q_2$, respectively, and those of the accumulated cost of the computation after the branch statement is bounded by $\phi_{Q'}$,
then the moments for the whole computation should be bounded by a ``weighted average'' of $(q_1 \otimes \phi_{Q'})$ and $(q_2 \otimes \phi_{Q'})$, with respect to the branching probability $p$.
We implement the weighted average by the combination operator $\oplus$ applied to
$
\tuple{ [p,p],[0,0],\cdots,[0,0] } \otimes q_1 \otimes \phi_{Q'}
$
and
$\tuple{ [1-p,1-p],[0,0],\cdots,[0,0] } \otimes q_2 \otimes \phi_{Q'}$,
because the $0$-th moments denote probabilities.

The rules \textsc{(Q-Call-Poly)} and \textsc{(Q-Call-Mono)} handle function calls.
Recall that in \cref{Exa:MomentPolymorphism}, we use the $\oplus$ operator to combine multiple potential functions for a function to reason about recursive function calls.
The restriction parameter $h$ is used to ensure that the derivation system only needs to reason about \emph{finitely} many post-annotations for each call site.
In rule \textsc{(Q-Call-Poly)}, where $h$ is smaller than the target moment $m$,
we fetch the pre- and post-condition $Q_1,Q_1'$ for the function $f$ from the specification context $\Delta_h$.
We then combine it with a \emph{frame} of $(h+1)$-restricted potential annotations $Q_2,Q_2'$ for the function $f$.
%
%
The frame is used to account for the interval bounds on the moments for the computation after the function call for most non-tail-recursive programs.
When $h$ reaches the target moment $m$, we use the rule \textsc{(Q-Call-Mono)} to reason \emph{moment-monomorphically}, because setting $h$ to $m+1$ implies that the frame can only be $\tuple{[0,0],[0,0],\cdots,[0,0]}$.

\iflong
The structural rule \textsc{(Q-Weaken)} is used to strengthen the pre-condition and relax the post-condition.
The entailment relation $\Gamma \models \Gamma'$ states that the logical implication $\Gamma \implies \Gamma'$ is valid.
In terms of the bounds on higher moments for cost accumulators, if the triple $\{\cdot;Q\}~S~\{\cdot;Q'\}$ is valid, then we can safely widen the intervals in the pre-condition $Q$ and narrow the intervals in the post-condition $Q'$.
\fi

\begin{figure}
  \centering
  \begin{small}
  \begin{pseudo}
    \kw{func} $\mathsf{rdwalk}$() \kw{begin} \\+
      ${\color{ACMDarkBlue} \{ \; x \!<\! d \!+\! 2; \; \langle [1,1], \; [2(d \!-\! x), 2(d\!-\!x)\!+\!4], }$ \\
      ${\color{ACMDarkBlue} \enskip [4(d \!-\! x)^2 + 6(d \!-\! x) \!-\! 4, 4(d\!-\!x)^2\!+\!22(d\!-\!x)\!+\!28] \rangle  \; \} }$ \\
      \kw{if} $\cbin{<}{x}{d}$ \kw{then} \\+
        ${\color{ACMDarkBlue} \{ \;  x  \!<\!  d;  \; \langle [1,1], \; [2(d \!-\! x), 2(d\!-\!x)\!+\!4], }$ \\
        ${\color{ACMDarkBlue} \enskip [4(d \!-\! x)^2 \!+\! 6(d\!-\! x) \!-\! 4, 4(d\!-\!x)^2\!+\!22(d\!-\!x)\!+\!28 ] \rangle \; \}  }$ \\
        $\isample{t}{ \kw{uniform}({-1},2)}$; \\
        ${\color{ACMDarkBlue} \{ \; x  \!<\!  d  \!\wedge\!  t  \!\le\!  2 ; \; \langle [1,1], \; [2(d \!-\! x \!-\! t) \!+\! 1, 2(d\!-\!x\!-\!t)\!+\!5], }$ \\
        ${\color{ACMDarkBlue} \enskip [4(d \!-\! x \!-\! t)^2 + 10(d \!-\! x \!-\! t) \!-\! 3, 4(d\!-\!x\!-\!t)^2\!+\!26(d\!-\!x\!-\!t)\!+\!37) ] \rangle \;  \} }$ \\
        $ \iassign{x}{  \eadd{x}{t} }$; \\
        ${\color{ACMDarkBlue} \{ \; x  \!<\!  d  \!+\!  2; \; \langle [1,1], \; [2(d \!-\! x) \!+\! 1, 2(d\!-\!x)\!+\!5], }$ \\
        ${\color{ACMDarkBlue} \enskip [4(d \!-\! x)^2 \!+\! 10(d \!-\! x) \!-\! 3, 4(d\!-\!x)^2\!+\!26(d\!-\!x)\!+\!37 ] \rangle \; \} }$ \\
        \kw{call} $\mathsf{rdwalk}$; \\
        ${\color{ACMDarkBlue} \{ \; \top; \; \tuple{ [1,1], \; [1,1], \; [1,1] } \; \} }$ \\
        \kw{tick}(1) \\
        ${\color{ACMDarkBlue} \{ \; \top; \; \tuple{[1,1],\;[0,0],\;[0,0]} \; \} } $ \\-
      \kw{fi} \\-
    \kw{end}
  \end{pseudo}
  \end{small}
  \caption{The $\mathsf{rdwalk}$ function with annotations for the interval-bounds on the first and second moments.}
  \label{Fi:AnnotatedRunningExample}
  \vspace{-10pt}
\end{figure}

\begin{example}\label{Exa:Derivation}
\cref{Fi:AnnotatedRunningExample} presents the logical context and the complete potential annotation for the first and second moments for the cost accumulator $\id{tick}$ of the $\mathsf{rdwalk}$ function from \cref{Exa:RecursiveRandomWalk}.
Similar to the reasoning in \cref{Exa:MomentPolymorphism}, we can justify the derivation using moment-polymorphic recursion and the moment bounds for $\m{rdwalk}$ with post-annotations
$\tuple{[0,0],[1,1],[1,1]}$ and $\tuple{[0,0],[0,0],[2,2]}$.
%
\end{example}

\subsection{Automatic Linear-Constraint Generation}
\label{Se:Automation}

We adapt existing techniques~\cite{CAV:CHR17,PLDI:NCH18} to automate our inference system by
(i) using an abstract interpreter to infer logical contexts,
(ii) generating templates and linear constraints by inductively applying the derivation rules to the analyzed program, and
(iii) employing an off-the-shelf LP solver to discharge the linear constraints.
During the generation phase, the coefficients of monomials in the polynomials from the ends of the intervals in every qualitative context $Q \in \calM_{\calP\calI}^{(m)}$ are recorded as symbolic names, and the inequalities among those coefficients---derived from the inference rules in \cref{Fi:InferenceRules}---are emitted to the LP solver.

\iflong

\paragraph{Generating linear constraints}
\cref{Fi:LPGeneration} demonstrates the generation process for some of the bounds in \cref{Fi:AnnotatedRecursiveRandomWalk}.
Let $B_k$ be a vector of \emph{monomials} over program variables $\mathsf{VID}$ of degree up to $k$.
Then a polynomial $\sum_{b \in B_k} q_b \cdot b$, where $q_b \in \bbR$ for all $b \in B_k$, can be represented as a vector of its coefficients $(q_b)_{b \in B_k}$.
We denote coefficient vectors by uppercase letters, while we use lowercase letters as names of the coefficients.
We also assume that the degree of the polynomials for the $k$-th moments is up to $k$.

\begin{figure*}
\begin{mathpar}\small
  \inferrule*[right=\textsc{(Q-Tick)}]
  { p^{tk}_1 = u^{tk}_1 \\
    q^{tk}_1 = u^{tk}_1 + v^{tk}_1 \\
    q^{tk}_x = v^{tk}_x \\
    q^{tk}_N = v^{tk}_N \\
    q^{tk}_r = v^{tk}_r \\\\
    t^{tk}_1 = w^{tk}_1 + 2v^{tk}_1 + u^{tk}_1 \\
    t^{tk}_x = w^{tk}_x + 2v^{tk}_x \\
    t^{tk}_{x^2} = w^{tk}_{x^2} \\
    \cdots
  }
  { \Delta \vdash \{ \top; (P^{tk},Q^{tk},T^{tk}) \}~\itick{1}~\{\top; (U^{tk},V^{tk},W^{tk}) \} }
  \and
  \inferrule*[right=\textsc{(Q-Sample)}]
  { p^{sa}_1 = u^{sa}_1 \\
    q^{sa}_1 = v^{sa}_1 + \sfrac{1}{2} \cdot v^{sa}_r \\
    q^{sa}_x = v^{sa}_x \\
    q^{sa}_N = v^{sa}_N \\
    q^{sa}_r = 0 \\\\
    t^{sa}_1 = w^{sa}_1 + \sfrac{1}{2} \cdot w^{sa}_r + 1 \cdot w^{sa}_{r^2} \\
    t^{sa}_x = w^{sa}_x + \sfrac{1}{2} \cdot w^{sa}_{r \cdot x} \\
    t^{sa}_{x^2} = w^{sa}_{x^2} \\
    \cdots
  }
  { \Delta \vdash \{ x<N; (P^{sa},Q^{sa},T^{sa}) \}~\isample{r}{\kw{uniform}({-1},2) }~\{x < N \wedge r \le 2; (U^{sa},V^{sa},W^{sa}) \} }
\end{mathpar}
\caption{Generate linear constraints, guided by inference rules.}
\label{Fi:LPGeneration}
\end{figure*}

For \textsc{(Q-Tick)}, we generate constraints that correspond to the composition operation $\otimes$ of the moment semiring.
For example, the second-moment component should satisfy
\begin{align*}
\sum_{b \in B_2} t^{tk}_b \cdot b & = \text{the second-moment component of}~ \big( (1,1,1) \\
& \otimes (\sum_{b \in B_0} u^{tk}_b \cdot b,\sum_{b \in B_1} v^{tk}_b \cdot b,\sum_{b \in B_2} w^{tk}_b \cdot b) \big) \\
& = \sum_{b \in B_2} w^{tk}_b \cdot b + 2 \cdot \sum_{b \in B_1} v^{tk}_b \cdot b + \sum_{b \in B_0} u^{tk}_b \cdot b.
\end{align*}
Then we extract $t^{tk}_1 = w^{tk}_1 + 2v^{tk}_1 + u^{tk}_1$ for $b=1$, and $t^{tk}_x = w^{tk}_x + 2v^{tk}_x$ for $b=x$, etc.
For \textsc{(Q-Sample)}, we generate constraints to perform ``partial evaluation'' on the polynomials by substituting $r$ with the moments of $\kw{uniform}({-1},2)$.
As we discussed in \cref{Se:InferenceRules}, let $D$ denote $\kw{uniform}({-1},2)$, then $\expe_{r\sim D}[w^{sa}_r \cdot r] = w^{sa}_r \cdot \expe_{r \sim D}[r] = \sfrac{1}{2} \cdot w^{sa}_r$, $\expe_{r \sim D}[w^{sa}_{r^2} \cdot r^2] = w^{sa}_{r^2} \cdot \expe_{r \sim D}[r^2] = 1 \cdot w^{sa}_{r^2}$.
Then we generate a constraint $t^{sa}_1 = w^{sa}_1 + \sfrac{1}{2} \cdot w^{sa}_r + 1 \cdot w^{sa}_{r^2}$ for $t^{sa}_1$.

The loop rule \textsc{(Q-Loop)} involves constructing loop \emph{invariants} $Q$, which is in general a non-trivial problem for automated static analysis.
Instead of computing the loop invariant $Q$ explicitly, our system represents $Q$ directly as a template with unknown coefficients,
then uses $Q$ as the post-annotation to analyze the loop body and obtain a pre-annotation,
and finally generates linear constraints that indicate the pre-annotation equals to $Q$.

The structural rule \textsc{(Q-Weaken)} can be applied at any point during the derivation.
In our implementation, we apply it where the control flow has a branch, because different branches might have different costs.
To handle the judgment $\Gamma \models Q \sqsupseteq Q'$, i.e., to generate constraints that ensure one interval is always contained in another interval, where the ends of the intervals are polynomials,
we adapt the idea of \emph{rewrite functions} \cite{PLDI:NCH18,CAV:CHR17}.
Intuitively, to ensure that $[L_1,U_1] \sqsupseteq [L_2,U_2]$, i.e., $L_1 \le L_2$ and $U_2 \le U_1$, under the logical context $\Gamma$, we generate constraints indicating that there exist two polynomials $T_1,T_2$ that are always nonnegative under $\Gamma$, such that $L_1=L_2+T_1$ and $U_1=U_2-T_2$.
Here, $T_1$ and $T_2$ are like slack variables, except that because all quantities are polynomials, they are too (i.e., slack polynomials).
In our implementation, $\Gamma$ is a set of linear constraints over program variables of the form $\calE \ge 0$ , then we can represent $T_1,T_2$ by \emph{conical} combinations (i.e., linear combinations with nonnegative scalars) of expressions $\calE$ in $\Gamma$.

\else

\begin{example}\label{Exa:LinearConstraintGeneration}
  We demonstrate linear-constraint generation for the upper bound on the first moment for the sampling statement $\isample{x}{\kw{uniform}({-1},2)}$ with a pre-annotation $\tuple{[0,0], [0,q_{x^2} \cdot x^2 + q_{x} \cdot x + q_{1} \cdot 1]}$ and a post-annotation $\tuple{[0,0], [0,q_{x^2}' \cdot x^2 + q_{x}' \cdot x + q_{1}' \cdot 1]}$,
  where we use polynomials of $x$ up to degree 2 as the templates, and $q_{x^2},q_x,q_1,q_{x^2}',q_x',q_1'$ are unknown numeric coefficients.
  By \textsc{(Q-Sample)}, we generate constraints to perform ``partial evaluation'' on the polynomials by substituting $x$ with the moments of $\kw{uniform}({-1},2)$.
  Let $D$ denote $\kw{uniform}({-1},2)$. Because
  \[
  \expe_{x \sim \mu_D}[q_{x^2}' \cdot x^2 + q_{x}' \cdot x + q_{1}' \cdot 1] = (q_{x^2}' \cdot 1 + q_{x}' \cdot \sfrac{1}{2} + q_{1}' \cdot 1),
  \]
  we generate these linear constraints:
  \[
  q_{x^2}=0, \quad q_x=0, \quad q_1 = q'_{x^2} + q'_x \cdot \sfrac{1}{2} + q_1'.
  \]
\end{example}

\revision{2}
\begin{changebar}
Constraint generation for other inference rules is similar to the process we describe in \cref{Exa:LinearConstraintGeneration}.
For example, let us consider the loop rule \textsc{(Q-Loop)}.
Instead of computing the loop invariant $Q$ explicitly, our system represents $Q$ directly as a template with unknown coefficients,
then uses $Q$ as the post-annotation to analyze the loop body and obtain a pre-annotation,
and finally generates linear constraints that indicate that the pre-annotation equals $Q$.
\end{changebar}
Details of the linear-constraint generation of our system are included in the technical report~\cite{Techreport}.

\fi

\iflong
\paragraph{Solving linear constraints}
\fi
The LP solver not only finds assignments to the coefficients that satisfy the constraints, it can also optimize a linear objective function.
In our central-moment analysis, we construct an objective function that tries to minimize imprecision.
For example, let us consider upper bounds on the variance.
We randomly pick a concrete valuation of program variables that satisfies the pre-condition (e.g., $d>0$ in \cref{Fi:RecursiveRandomWalk}), and then substitute program variables with the concrete valuation in the polynomial for the upper bound on the variance (obtained from bounds on the raw moments).
The resulting linear combination of coefficients, which we set as the objective function, stands for the variance under the concrete valuation.
Thus, minimizing the objective function produces the most precise upper bound on the variance under the specific concrete valuation.
Also, we can extract a \emph{symbolic} upper bound on the variance using the assignments to the coefficients.
Because the derivation of the bounds only uses the given pre-condition, the symbolic bounds apply to all valuations\Omit{of program variables} that satisfy the pre-condition.



\section{Soundness of Higher-Moment Analysis}
\label{Se:SemanticOptionalStopping}

%
%
%
%
%
%
%
%
%
%

%
In this section, we study the soundness of our derivation system for higher-moment analysis.
We first present a Markov-chain semantics for the probabilistic programming language \lang{} to reason about how \emph{stepwise} costs contribute to the \emph{global} accumulated cost (\cref{Se:ExpectedCostBoundAnalysis}).
We then formulate higher-moment analysis with respect to the semantics and prove the soundness of our derivation system for higher-moment analysis based on a recent extension to the \emph{Optional Stopping Theorem} (\cref{Se:AnExtendedOST}).
%
%
%
Finally, we sketch the algorithmic approach for ensuring the soundness of our analysis (\cref{Se:AutomaticTerminationCheck}).
%

\subsection{A Markov-Chain Semantics}
\label{Se:ExpectedCostBoundAnalysis}

\paragraph{Operational semantics}
We start with a small-step operational semantics with continuations, which we will use later to construct the Markov-chain semantics.
We follow a distribution-based approach~\cite{ICFP:BLG16,JCSS:Kozen81}
to define an operational cost semantics for \lang{}.
Full details of the semantics are included in
\iflong
\Cref{Se:OperationalCostSemantics}.
\else
the technical report~\cite{Techreport}.
\fi
A \emph{program configuration} $\sigma \in \Sigma$ is a quadruple $\tuple{\gamma,S,K,\alpha}$ where $\gamma : \mathsf{VID} \to \bbR$ is a program state that maps variables to values, $S$ is the statement being executed, $K$ is a continuation that describes what remains to be done after the execution of $S$, and $\alpha \in \bbR$ is the global cost accumulator.
%
%
%
%
An execution of an \lang{} program $\tuple{\scrD,S_\mathsf{main}}$ is initialized with $\tuple{\lambda\_. 0, S_{\mathsf{main}}, \kstop,0}$, and the termination configurations have the form $\tuple{\_,\iskip,\kstop,\_}$,
where $\kstop$ is an empty continuation.  

Different from a standard semantics where each program configuration steps to at most one new configuration, a probabilistic semantics may pick several different new configurations.
The evaluation relation for \lang{} has the form $\sigma \mapsto \mu$ where $\mu \in \bbD(\Sigma)$ is a probability measure over configurations.
Below are two example rules.
The rule \textsc{(E-Prob)} constructs a distribution whose support has exactly two elements, which stand for the two branches of the probabilistic choice.
We write $\delta(\sigma)$ for the \emph{Dirac measure} at $\sigma$, defined as $\lambda A.[\sigma\in A]$ where $A$ is a measurable subset of $\Sigma$.
We also write $p \cdot \mu_1 + (1-p)\cdot \mu_2$ for a convex combination of measures $\mu_1$ and $\mu_2$ where $p \in [0,1]$, defined as $\lambda A. p \cdot \mu_1(A) + (1-p) \cdot \mu_2(A)$.
The rule \textsc{(E-Sample)} ``pushes'' the probability distribution of $D$ to a distribution over post-sampling program configurations.
\begin{mathpar}\small
  \inferrule*[right=(E-Prob)]{S=\iprob{p}{S_1}{S_2} }{ \tuple{\gamma, S, K,\alpha } \mapsto p \cdot \delta(\tuple{\gamma,S_1,K,\alpha}) + (1-p) \cdot \delta(\tuple{\gamma,S_2,K,\alpha}) }
  \and
  \inferrule*[right=(E-Sample)]{ }{ \tuple{\gamma, x\!\sim\!D, K,\alpha} \mapsto \lambda A. \mu_D(\{ r \mid \tuple{\gamma[x \!\mapsto\! r],\iskip,K,\alpha} \in A\}) }  
\end{mathpar}

\begin{example}\label{Exa:EvaluationRelation}
  Suppose that a random sampling statement is being executed, i.e., the current configuration is
  \begin{center}\small
  $\tuple{\{ t \mapsto t_0 \}, (\isample{t}{\kw{uniform}(-1,2)}),K_0,\alpha_0}.$
  \end{center}
  The probability measure for the uniform distribution is $\lambda O. \int_O \frac{[-1 \le x \le 2]}{3}   dx$.
  \revision{7}
  \begin{changebar}
  Thus, by the rule \textsc{(E-Sample)}, we derive the post-sampling probability measure over configurations:
  \begin{center}\small
  $
  \lambda A. \int_{\bbR} [\tuple{ \{ t \mapsto r \}, \iskip, K_0,\alpha_0  } \in A] \cdot \frac{[-1 \le r \le 2]}{3} dr.
  $
  \end{center}
  \end{changebar}
\end{example}

\paragraph{A Markov-chain semantics}
In this work, we harness Markov-chain-based reasoning~\cite{ESOP:KKM16,LICS:OKK16} to develop a Markov-chain cost semantics for \lang{}, based on the evaluation relation $\sigma \mapsto \mu$.
An advantage of this approach is that it allows us to study how the cost of every single evaluation step contributes to the accumulated cost at the exit of the program.
Details of this semantics are included in 
\iflong
\Cref{Se:TraceBasedCostSemantics}.
\else
the technical report~\cite{Techreport}.
\fi

Let $(\Omega,\calF,\prob)$ be the probability space where $\Omega \defeq \Sigma^{\bbZ^+}$ is the set of all \emph{infinite} traces over program configurations, $\calF$ is a $\sigma$-algebra on $\Omega$, and $\prob$ is a probability measure on $(\Omega,\calF)$ obtained by the evaluation relation $\sigma \mapsto \mu$ and the initial configuration $\tuple{\lambda\_.0, S_{\mathsf{main}},\kstop,0}$.
Intuitively, $\prob$ specifies the probability distribution over all possible executions of a probabilistic program.
The probability of an assertion $\theta$ with respect to $\prob$, written $\prob[\theta]$, is defined as $\prob(\{ \omega \mid \theta(\omega)~\text{is true} \})$.

To formulate the accumulated cost at the exit of the program, we define the \emph{stopping time} $T:\Omega \to \bbZ^+ \cup \{\infty\}$ of a probabilistic program as a random variable on the probability space $(\Omega,\calF,\prob)$ of program traces:
\[
T(\omega) \defeq \inf \{ n \in \bbZ^+ \mid \omega_n = \tuple{\_,\iskip,\kstop,\_} \},
\]
i.e., $T(\omega)$ is the number of evaluation steps before the trace $\omega$ reaches some termination configuration $\tuple{\_,\iskip,\kstop,\_}$.
We define the accumulated cost $A_T : \Omega \to \bbR$ with respect to the stopping time $T$ as
\[
A_T(\omega) \defeq A_{T(\omega)}(\omega),
\]
where $A_n : \Omega \to \bbR$ captures the accumulated cost at the $n$-th evaluation step for $n \in \bbZ^+$, which is defined as
\[
A_n(\omega) \defeq \alpha_n ~\text{where}~\omega_n =\tuple{\_,\_,\_,\alpha_n}.
\]
%
The $m$-th moment of the accumulated cost is given by the expectation $\expe[A_T^m]$ with respect to $\prob$.

\subsection{Soundness of the Derivation System}
\label{Se:AnExtendedOST}

%
Proofs for this section are included in 
\iflong
\Cref{Se:MarkovChainSemantics,Se:ProofHigherMoments,Se:SoundnessProof}.
\else
the technical report~\cite{Techreport}.
\fi

\paragraph{The expected-potential method for moment analysis}
We fix a degree $m \in \bbN$ and let $\calM_\calI^{(m)}$ be the $m$-th order moment semiring instantiated with the interval semiring $\calI$.
We now define $\calM_\calI^{(m)}$-valued expected-potential functions.

\begin{definition}\label{De:IntervalRankingFunction}
  A measurable map $\phi : \Sigma \to \calM_\calI^{(m)}$ is said to be an \emph{expected-potential function} if
  \begin{enumerate}[(i)]
    \item $\phi(\sigma) = \aone$ if $\sigma=\tuple{\_,\iskip,\kstop,\_}$, and
    \item $\phi(\sigma) \sqsupseteq \expe_{\sigma' \sim {\mapsto}(\sigma)}[\many{\big([(\alpha'-\alpha)^k,(\alpha'-\alpha)^k]\big)_{0 \le k \le m}} \otimes \phi(\sigma')]$ where $\sigma=\tuple{\_,\_,\_,\alpha},\sigma'=\tuple{\_,\_,\_,\alpha'}$ for all $\sigma \in \Sigma$.
  \end{enumerate}
\end{definition}

Intuitively, $\phi(\sigma)$ is an interval bound on the moments for the accumulated cost of the computation that \emph{continues from} the configuration $\sigma$.
%
We define $\mathbf{\Phi}_n$ and $\mathbf{Y}_n$, where $n \in \bbZ^+$, to be $\calM_\calI^{(m)}$-valued random variables on the probability space $(\Omega,\calF,\prob)$ of the Markov-chain semantics as
\[
  \mathbf{\Phi}_n(\omega) \defeq \phi(\omega_n), 
  \mathbf{Y}_n(\omega) \defeq \many{ \tuple{ [A_n(\omega)^k, A_n(\omega)^k] }_{0 \le k \le m} } \otimes \mathbf{\Phi}_n(\omega)  .
\]
In the definition of $\mathbf{Y}_n$, we use $\otimes$ to compose the powers of the accumulated cost at step $n$ and the expected potential function that stands for the moments of the accumulated cost for the rest of the computation.

\begin{lemma}\label{Lem:MomentInvariant}
  By the properties of potential functions, we can prove that $\expe[\mathbf{Y}_{n+1} \mid \mathbf{Y}_n] \sqsubseteq \mathbf{Y}_n$ almost surely, for all $n \in \bbZ^+$.
\end{lemma}

We call $\{\mathbf{Y}_n\}_{n \in \bbZ^+}$ a \emph{moment invariant}.
Our goal is to establish that $\expe[\mathbf{Y}_T] \sqsubseteq \expe[\mathbf{Y}_0]$, i.e., the initial interval-valued potential $\expe[\mathbf{Y}_0] = \expe[\aone \otimes \mathbf{\Phi}_0] = \expe[\mathbf{\Phi}_0]$ brackets the higher moments of the accumulated cost $\expe[\mathbf{Y}_T] = \expe[\many{ \tuple{ [A_T^k, A_T^k] }_{0 \le k \le m}  } \otimes \aone] = \many{\tuple{ [\expe[A_T^k], \expe[A_T^k] ]  }_{0 \le k \le m} }$.
%

\paragraph{Soundness}
The soundness of the derivation system is proved with respect to the Markov-chain semantics.
%
Let $\norm{\many{\tuple{[a_k,b_k]}_{0 \le k \le m} }}_\infty \defeq \max_{0 \le k \le m} \{ \max\{ |a_k|,|b_k| \} \} $.

\begin{theorem}\label{The:Soundness}
  Let $\tuple{\scrD,S_\mathsf{main}}$ be a probabilistic program.
  Suppose $\Delta \vdash \{\Gamma;Q\}~S_{\mathsf{main}}~\{\Gamma';\aone\}$, where $Q \in \calM_{\calP\calI}^{(m)}$ and the ends of the $k$-th interval in $Q$ are polynomials in $\bbR_{kd}[\mathsf{VID}]$.
  Let $\{\mathbf{Y}_n\}_{n \in \bbZ^+}$ be the moment invariant extracted from the Markov-chain semantics with respect to the derivation of $\Delta \vdash \{\Gamma;Q\}~S_{\mathsf{main}}~\{\Gamma';\aone\}$.
  If the following conditions hold:
  \begin{enumerate}[(i)]
    \item\label{Item:SoundnessTermination} $\expe[T^{md}] < \infty$, and
    \item\label{Item:SoundnessBounded} there exists $C \ge 0$ such that for all $n \in \bbZ^+$, $\norm{\mathbf{Y}_n}_\infty \le C \cdot (n+1)^{md}$ almost surely,
  \end{enumerate}
  Then
  $\many{\tuple{[\expe[A_T^k], \expe[A_T^k] ]}_{0 \le k \le m}} \aord \phi_Q(\lambda\_.0)$.
\end{theorem}

The intuitive meaning of $\many{\tuple{[\expe[A_T^k], \expe[A_T^k]]}_{0 \le k \le m}} \aord \phi_Q(\lambda\_.0)$ is that the moment $\expe[A_T^k]$ of the accumulated cost upon program termination is bounded by the interval in the $k$\textsuperscript{th}-moment component of $\phi_Q(\lambda\_.0)$, where $Q$ is the quantitative context and $\lambda\_.0$ is the initial state.

As we discussed in \cref{Se:SoundnessCriteria} and \cref{Exa:UnsoundPotentialFunction}, the expected-potential method is \emph{not} always sound for deriving bounds on higher moments for cost accumulators in probabilistic programs.
The extra conditions \cref{The:Soundness}\ref{Item:SoundnessTermination} and \ref{Item:SoundnessBounded} impose
constraints on the analyzed program and the expected-potential function, which allow us to reduce the soundness to the \emph{optional stopping problem} from probability theory.


\paragraph{Optional stopping}
Let us represent the moment invariant $\{\mathbf{Y}_n\}_{n \in \bbZ^+}$ as
\[
\{ \tuple{[L^{(0)}_n, U^{(0)}_n], [L^{(1)}_n, U^{(1)}_n], \cdots, [L^{(m)}_n, U^{(m)}_n]  }  \}_{n \in \bbZ^+},
\]
where $L^{(k)}_n, U^{(k)}_n : \Omega \to \bbR$ are real-valued random variables on the probability space $(\Omega,\calF,\prob)$ of the Markov-chain semantics, for $n \in \bbZ^+$, $0 \le k \le m$.
We then have the observations below as direct corollaries of \cref{Lem:MomentInvariant}:
\begin{itemize}
  \item For any $k$, the sequence $\{U_n^{(k)}\}_{n \in \bbZ^+}$ satisfies $\expe[U_{n+1}^{(k)} \mid U_n^{(k)}] \le U_n^{(k)}$ almost surely, for all $n \in \bbZ^+$, and we want to find sufficient conditions for $\expe[U_T^{(k)}] \le \expe[U_0^{(k)}]$.
  \item For any $k$, the sequence $\{L_n^{(k)}\}_{n \in \bbZ^+}$ satisfies $\expe[L_{n+1}^{(k)} \mid L_n^{(k)}] \ge L_n^{(k)}$ almost surely, for all $n \in \bbZ^+$, and we want to find sufficient conditions for $\expe[L_T^{(k)}] \ge \expe[L_0^{(k)}]$.
\end{itemize}

These kinds of questions can be reduced to \emph{optional stopping problem} from probability theory.
Recent research~\cite{PLDI:WFG19,POPL:HKG20,CAV:BEF16,misc:SO19} has used the \emph{Optional Stopping Theorem} (OST) from probability theory to establish \emph{sufficient} conditions for the soundness for analysis of probabilistic programs.
However, the classic OST turns out to be \emph{not} suitable for higher-moment analysis.
We extend OST with a \emph{new} sufficient condition that allows us to prove \cref{The:Soundness}.
We discuss the details of our extended OST in a companion paper~\cite{Companion};
in this work, we focus on the derivation system for central moments.
\iflong
We also include a brief discussion on OST in \Cref{Se:ProofHigherMoments}.
\fi

\subsection{An Algorithm for Checking Soundness Criteria}
\label{Se:AutomaticTerminationCheck}

\paragraph{Termination Analysis}
We reuse our system for automatically deriving higher moments, which we developed in \cref{Se:InferenceRules,Se:Automation}, for checking if $\expe[T^{md}] < \infty$ (\cref{The:Soundness}\ref{Item:SoundnessTermination}).
To reason about termination time, we assume that every program statement increments the cost accumulator by one.
For example, the inference rule 
\textsc{(Q-Sample)} becomes
\begin{mathpar}\small
  \inferrule
  { \Gamma = \Forall{x \in \mathrm{supp}(\mu_D)} \Gamma' \\
    Q = \tuple{1,1,\cdots,1} \otimes \expe_{x \sim \mu_D}[Q']
  }
  { \Delta \vdash \{\Gamma;Q\}~\isample{x}{D}~\{\Gamma';Q'\} }
\end{mathpar}
However, we cannot apply \cref{The:Soundness} for the soundness of the termination-time analysis, because that would introduce a circular dependence.
Instead, we use a different proof technique to reason about $\expe[T^{md}]$, taking into account the \emph{monotonicity} of the runtimes.
Because upper-bound analysis of higher moments of runtimes has been studied by~\citet{TACAS:KUH19}, we skip the details, but include them in
\iflong
\Cref{Se:TerminationAnalysis}.
\else
the technical report~\cite{Techreport}.
\fi

\paragraph{Boundedness of $\norm{Y_n}_\infty$, $n \in \bbZ^+$}
To ensure that the condition in \cref{The:Soundness}\ref{Item:SoundnessBounded} holds, we check if the analyzed program satisfies the \emph{bounded-update} property: every (deterministic or probabilistic) assignment to a program variable updates the variable with a change bounded by a constant $C$ almost surely.
Then the absolute value of every program variable at evaluation step $n$ can be bounded by $C \cdot n = O(n)$.
Thus, a polynomial up to degree $\ell \in \bbN$ over program variables can be bounded by $O(n^\ell)$ at evaluation step $n$.
As observed by~\citet{PLDI:WFG19}, bounded updates are common in practice.
%


\section{Tail-Bound Analysis}
\label{Se:TailBound}

One application of our central-moment analysis is\Omit{ to use the bounds on higher (central) moments with \emph{concentration-of-measure} inequalities} to bound the probability that the accumulated cost deviates from some given quantity.
In this section, we sketch how we produce the tail bounds shown in \cref{Fi:MotivatingExample}(c).

There are a lot of \emph{concentration-of-measure} inequalities in probability theory~\cite{book:Dubhashi09}.
Among those, one of the most important is \emph{Markov's inequality}:

\begin{proposition}\label{Prop:Markov}
	If $X$ is a nonnegative random variable and $a > 0$, then $\prob[X \ge a] \le \frac{\expe[X^k]}{a^k}$ for any $k \in \bbN$.
\end{proposition}

Recall that \cref{Fi:MotivatingExample}(b) presents upper bounds on the raw moments $\expe[\id{tick}] \le 2d+4$ and $\expe[\id{tick}^2] \le 4d^2+22d+28$ for the cost accumulator $\id{tick}$.
With Markov's inequality, we derive the following tail bounds:
{\small
\begin{align}
  \prob[\id{tick} \ge 4d] & \le \frac{\expe[\id{tick}]}{4d} \le \frac{2d+4}{4d} \xrightarrow{d \to \infty} \frac{1}{2}, \label{Eq:FstMomMarkov} \\
  \prob[\id{tick} \ge 4d] & \le \frac{\expe[\id{tick}^2]}{(4d)^2} \le \frac{4d^2+22d+28}{16d^2} \xrightarrow{d \to \infty} \frac{1}{4}. \label{Eq:SndMomMarkov}
\end{align}
}
Note that \eqref{Eq:SndMomMarkov} provides an asymptotically more precise bound on $\prob[\id{tick} \ge 4d]$ than \eqref{Eq:FstMomMarkov} does, when $d$ approaches infinity.

Central-moment analysis can obtain an even more precise tail bound.
As presented in \cref{Fi:MotivatingExample}(b), our analysis derives $\vari[\id{tick}] \le 22d+28$ for the variance of \id{tick}.
We can now employ concentration inequalities that involve variances of random variables.
Recall \emph{Cantelli's inequality}:

\begin{proposition}\label{Prop:Cantelli}
	If $X$ is a random variable and $a > 0$, then $\prob[X \!-\! \expe[X] \!\ge\! a] \le \frac{\vari[X]}{\vari[X]+a^2}$ and $\prob[X \!-\! \expe[X] \!\le\! {-a}] \le \frac{\vari[X]}{\vari[X]+a^2}$.
\end{proposition}

With Cantelli's inequality, we obtain the following tail bound, where we assume $d \ge 2$:
\begin{equation}\label{Eq:VariCantelli}\small
\begin{split}
  & \prob[\id{tick} \ge 4d] = \prob[\id{tick} - (2d+4) \ge 2d - 4] \\
  {}\le{} & \prob[\id{tick} - \expe[\id{tick}] \ge 2d-4] \le \frac{\vari[\id{tick}]}{\vari[\id{tick}] + (2d-4)^2} \\
  ={} & 1 - \frac{(2d-4)^2}{\vari[\id{tick}] + (2d-4)^2} \le \frac{22d+28}{4d^2+ 6d+44} \xrightarrow{d \to \infty} 0.
\end{split}
\end{equation}
For all $d \ge 15$, \eqref{Eq:VariCantelli} gives a more precise bound than both \eqref{Eq:FstMomMarkov} and \eqref{Eq:SndMomMarkov}.
It is also clear from \cref{Fi:MotivatingExample}(c), where we plot the three tail bounds \eqref{Eq:FstMomMarkov}, \eqref{Eq:SndMomMarkov}, and \eqref{Eq:VariCantelli}, that the asymptotically most precise bound is the one obtained via variances.

In general, for higher central moments, we employ \emph{Chebyshev's inequality} to derive tail bounds:

\begin{proposition}\label{Prop:Chebyshev}
  If $X$ is a random variable and $a>0$, then $\prob[\abs{X - \expe[X]} \ge a] \le \frac{\expe[(X-\expe[X])^{2k}]}{ a^{2k} }$ for any $k \in \bbN$.
\end{proposition}

In our experiments, we use Chebyshev's inequality to derive tail bounds from the fourth central moments.
We will show in \cref{Fi:ComparisonWithTACAS} that these tail bounds can be more precise than those obtained from both raw moments and variances.


\section{Implementation and Experiments}
\label{Se:ImplementationAndExperiments}


\paragraph{Implementation}
%
Our tool is implemented in OCaml, and consists of about 5,300 LOC.
The tool works on imperative arithmetic probabilistic programs using a CFG-based IR~\cite{PLDI:WHR18}.
The language supports recursive functions, continuous distributions, unstructured control-flow, and local variables.
To infer the bounds on the central moments for a cost accumulator in a program, the user needs to specify the order of the analyzed moment, and a maximal degree for the polynomials to be used in potential-function templates.
Using APRON~\cite{CAV:JM09}, we implemented an interprocedural numeric analysis to infer the logical contexts used in the derivation.
%
%
%
%
%
%
We use the off-the-shelf solver Gurobi~\cite{misc:GUROBI} for LP solving.

\paragraph{Evaluation setup}
We evaluated our tool to answer the following three research questions:
\begin{enumerate}
  \item
    \label{RQ:1} How does the raw-moment inference part of our tool compare to existing techniques for expected-cost bound analysis~\cite{PLDI:NCH18,PLDI:WFG19}?
  \item
    \label{RQ:2} How does our tool compare to the state of the art in tail-probability analysis (which is based only on higher \emph{raw} moments~\cite{TACAS:KUH19})?
  
  \revision{1}
  \begin{changebar}
  \item
    \label{RQ:3} How scalable is our tool? Can it analyze programs with many recursive functions?
  \end{changebar}
\end{enumerate}

For the first question, we collected a broad suite of challenging examples from related work~\cite{PLDI:NCH18,PLDI:WFG19,TACAS:KUH19} with different loop and recursion patterns, as well as probabilistic branching, discrete sampling, and continuous sampling.
Our tool achieved comparable precision and efficiency with the prior work on expected-cost bound analysis~\cite{PLDI:NCH18,PLDI:WFG19}.
%
The details are included in
\iflong
\Cref{Se:ExperimentDetails}.
\else
the technical report~\cite{Techreport}.
\fi

For the second question, we evaluated our tool on the complete benchmarks from Kura el al.~\cite{TACAS:KUH19}.
We also conducted a case study of a timing-attack analysis for a program provided by DARPA during engagements of the STAC program~\cite{misc:STAC}, where central moments are more useful than raw moments to bound the success probability of an attacker.
We include the case study in
\iflong
\Cref{Se:TimingAttackAnalysis}.
\else
the technical report~\cite{Techreport}.
\fi

\revision{1}
\begin{changebar}
For the third question, we conducted case studies on two sets of synthetic benchmark programs:
\begin{itemize}
  \item coupon-collector programs with $N$ coupons ($N \in [1,10^3]$), where each program is implemented as a set of tail-recursive functions, each of which represents a state of coupon collection, i.e., the number of coupons collected so far; and
  
  \item 
  random-walk programs with $N$ consecutive one-dimensional random walks ($N \in [1,10^3]$), each of which starts at a position that equals the number of steps taken by the previous random walk to reach the ending position (the origin).
  Each program is implemented as a set of non-tail-recursive functions, each of which represents a random walk.
  The random walks in the same program can have different transition probabilities.
\end{itemize}
The largest synthetic program has nearly 16,000 LOC.
We then ran our tool to derive an upper bound on the fourth (resp., second) central moment of the runtime for each coupon-collector (resp., random-walk) program.
\end{changebar}

The experiments were performed on a machine with an Intel Core i7 3.6GHz processor and 16GB of RAM under macOS Catalina 10.15.7.

\paragraph{Results.}
Some of the evaluation results to answer the second research question are presented in \cref{Ta:ComparisonWithTACAS}.
The programs (1-1) and (1-2) are coupon-collector problems with a total of two and four coupons, respectively.
The programs (2-1) and (2-2) are one-dimensional random walks with integer-valued and real-valued coordinates, respectively.
We omit three other programs here but include the full results in
\iflong
\Cref{Se:ExperimentDetails}.
\else
the technical report~\cite{Techreport}.
\fi
%
The table contains the inferred upper bounds on the moments for runtimes of these programs, and the running times of the analyses.
We compared our results with Kura et al.'s inference tool for raw moments~\cite{TACAS:KUH19}.
Our tool is as precise as, and sometimes more precise than the prior work on all the benchmark programs.
Meanwhile, our tool is able to infer an upper bound on the raw moments of degree up to four on all the benchmarks, while the prior work reports failure on some higher moments for the random-walk programs.
In terms of efficiency, our tool completed each example in less than 10 seconds, while the prior work took more than a few minutes on some programs.
One reason why our tool is more efficient is that we always reduce higher-moment inference with non-linear polynomial templates to efficient LP solving, but the prior work requires semidefinite programming (SDP) for polynomial templates.

\newcommand{\bd}[1]{{\textbf{#1}}}
\begin{table}
  \centering
  \caption{Upper bounds on the raw/central moments of runtimes, with comparison to~\citet{TACAS:KUH19}. ``T/O'' stands for timeout after 30 minutes. ``N/A'' means that the tool is not applicable. ``-'' indicates that the tool fails to infer a bound. Entries with more precise results or less analysis time are marked in bold. Full results are included in
  \iflong
  \Cref{Se:ExperimentDetails}.
  \else
  the technical report~\cite{Techreport}.
  \fi}
  \label{Ta:ComparisonWithTACAS}
  \resizebox{\columnwidth}{!}{%
  \begin{tabular}{@{\hspace{1pt}}c@{\hspace{1pt}}|@{\hspace{1pt}}c@{\hspace{1pt}}||@{\hspace{1pt}}c@{\hspace{1pt}}|@{\hspace{1pt}}c@{\hspace{1pt}}||@{\hspace{1pt}}c@{\hspace{1pt}}|@{\hspace{1pt}}c@{\hspace{1pt}}}
    \hline
    \multirow{2}{*}{program} & \multirow{2}{*}{moment} & \multicolumn{2}{c@{\hspace{1pt}}||@{\hspace{1pt}}}{this work} & \multicolumn{2}{c}{\citet{TACAS:KUH19}} \\ \hhline{~~----}
    & & upper bnd. & time (sec) & upper bnd. & time (sec) \\ \hline
    \multirow{5}{*}{(1-1)} & 2\textsuperscript{nd} raw & \bd{201} & 0.019 & \bd{201
    } & \bd{0.015}  \\ \hhline{~*{5}{-}}
    & 3\textsuperscript{rd} raw & \bd{3829} & \bd{0.019} & \bd{3829} & 0.020 \\ \hhline{~*{5}{-}}
    & 4\textsuperscript{th} raw & \bd{90705} & \bd{0.023} & \bd{90705} & 0.027 \\ \hhline{~*{5}{-}}
    & 2\textsuperscript{nd} central & \bd{32} & \bd{0.029} & N/A & N/A \\ \hhline{~*{5}{-}}
    & 4\textsuperscript{th} central & \bd{9728} & \bd{0.058} & N/A & N/A \\ \hline
    \multirow{5}{*}{(1-2)} & 2\textsuperscript{nd} raw & \bd{2357} & 1.068 & 3124 & \bd{0.037} \\ \hhline{~*{5}{-}}
    & 3\textsuperscript{rd} raw & \bd{148847} & 1.512 & 171932 & \bd{0.062} \\ \hhline{~*{5}{-}}
    & 4\textsuperscript{th} raw  & \bd{11285725} & 1.914 & 12049876 & \bd{0.096} \\ \hhline{~*{5}{-}}
    & 2\textsuperscript{nd} central & \bd{362} & \bd{3.346} & N/A & N/A \\ \hhline{~*{5}{-}}
    & 4\textsuperscript{th} central & \bd{955973} & \bd{9.801} & N/A & N/A \\ \hline
    \multirow{5}{*}{(2-1)} & 2\textsuperscript{nd} raw & \bd{2320} & \bd{0.016} & \bd{2320} & 11.380 \\ \hhline{~*{5}{-}}
    & 3\textsuperscript{rd} raw & \bd{691520} & \bd{0.018} & - & 16.056 \\ \hhline{~*{5}{-}}
    & 4\textsuperscript{th} raw & \bd{340107520} & \bd{0.021} & - & 23.414 \\ \hhline{~*{5}{-}}
    & 2\textsuperscript{nd} central & \bd{1920} & \bd{0.026} & N/A & N/A \\ \hhline{~*{5}{-}}
    & 4\textsuperscript{th} central & \bd{289873920} & \bd{0.049} & N/A & N/A \\ \hline
    \multirow{5}{*}{(2-2)} & 2\textsuperscript{nd} raw & \bd{8375} & \bd{0.022} & \bd{8375} & 38.463 \\ \hhline{~*{5}{-}}
    & 3\textsuperscript{rd} raw & \bd{1362813} & \bd{0.028} & - & 73.408  \\ \hhline{~*{5}{-}}
    & 4\textsuperscript{th} raw & \bd{306105209} & \bd{0.035} & - & 141.072 \\ \hhline{~*{5}{-}}
    & 2\textsuperscript{nd} central &  \bd{5875} & \bd{0.029} & N/A & N/A \\ \hhline{~*{5}{-}}
    & 4\textsuperscript{th} central & \bd{447053126} & \bd{0.086} & N/A & N/A \\ \hline
  \end{tabular}%
  }
  \vspace{-10pt}
\end{table}

Besides raw moments, our tool is also capable of inferring upper bounds on the central moments of runtimes for the benchmarks.
To evaluate the quality of the inferred central moments, \cref{Fi:ComparisonWithTACAS} plots the upper bounds of tail probabilities on runtimes $T$ obtained by~\citet{TACAS:KUH19}, and those by our central-moment analysis.
Specifically, the prior work uses Markov's inequality (\cref{Prop:Markov}), while we are also able to apply Cantelli's and Chebyshev's inequality (\cref{Prop:Cantelli,Prop:Chebyshev}) with central moments.
Our tool outperforms the prior work on programs (1-1) and (1-2)\Omit{, (2-3), and (2-5)}, and derives better tail bounds when $d$ is large on program (2-2)\Omit{ and (2-4)}, while it obtains similar curves on program (2-1).
%

\tikzexternalenable
\begin{figure}
  \centering
  \begin{tabular}{c@{\hspace{0.5ex}}c@{\hspace{0.5ex}}c@{\hspace{0.5ex}}c}
    \includegraphics{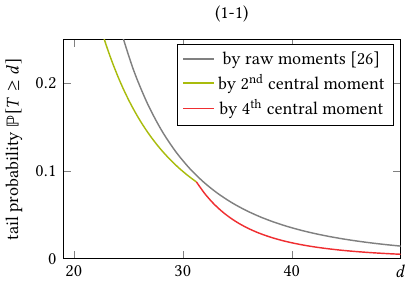}
    &
    \includegraphics{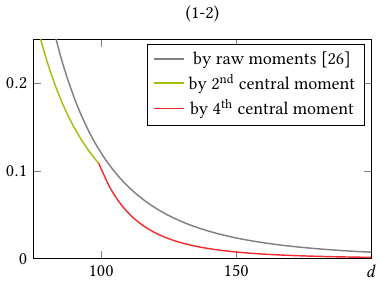}
    \\[-3pt]
    \includegraphics{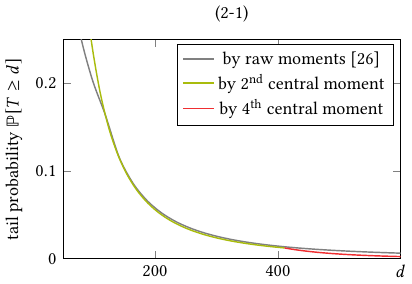}
    &
    \includegraphics{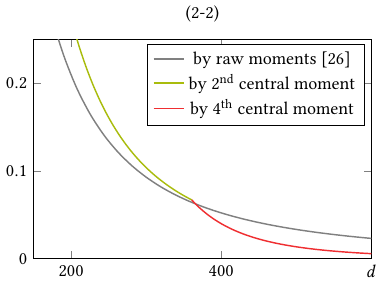}
  \end{tabular}
  \caption{Upper bound of the tail probability $\prob[T \ge d]$ as a function of $d$, with comparison to~\citet{TACAS:KUH19}. Each gray line is the minimum of tail bounds obtained from the raw moments of degree up to four inferred by~\citet{TACAS:KUH19}. Green lines and red lines are the tail bounds given by 2\textsuperscript{nd} and 4\textsuperscript{th} central moments inferred by our tool, respectively. We include the plots for other programs in
    \iflong
    \Cref{Se:ExperimentDetails}.
    \else
    the technical report~\cite{Techreport}.
    \fi}
  \label{Fi:ComparisonWithTACAS}
  \vspace{10pt}
\end{figure}
\tikzexternaldisable

\revision{1}
\begin{changebar}
\paragraph{Scalability.}
In \cref{Fi:Scalability}, we demonstrate the running times of our tool on the two sets of synthetic benchmark programs;
\cref{Fi:ScalabilityCoupons} plots the analysis times for coupon-collector programs as a function of the independent variable $N$ (the total number of coupons),
and \cref{Fi:ScalabilityRdwalks} plots the analysis times for random-walk programs as a function of $N$ (the total number of random walks).
The evaluation statistics show that our tool achieves good scalability in both case studies:
the runtime is almost a linear function of the program size, which is the number of recursive functions for both case studies.
Two reasons why our tool is scalable on the two sets of programs are
(i) our analysis is compositional and uses function summaries to analyze function calls, and
(ii) for a fixed set of templates and a fixed diameter of the call graph, the number of linear constraints generated by our tool grows linearly with the size of the program, and the LP solvers available nowadays can handle large problem instances efficiently.

\begin{figure}
\centering
\begin{subfigure}{0.49\columnwidth}
\centering
\includegraphics[width=\textwidth]{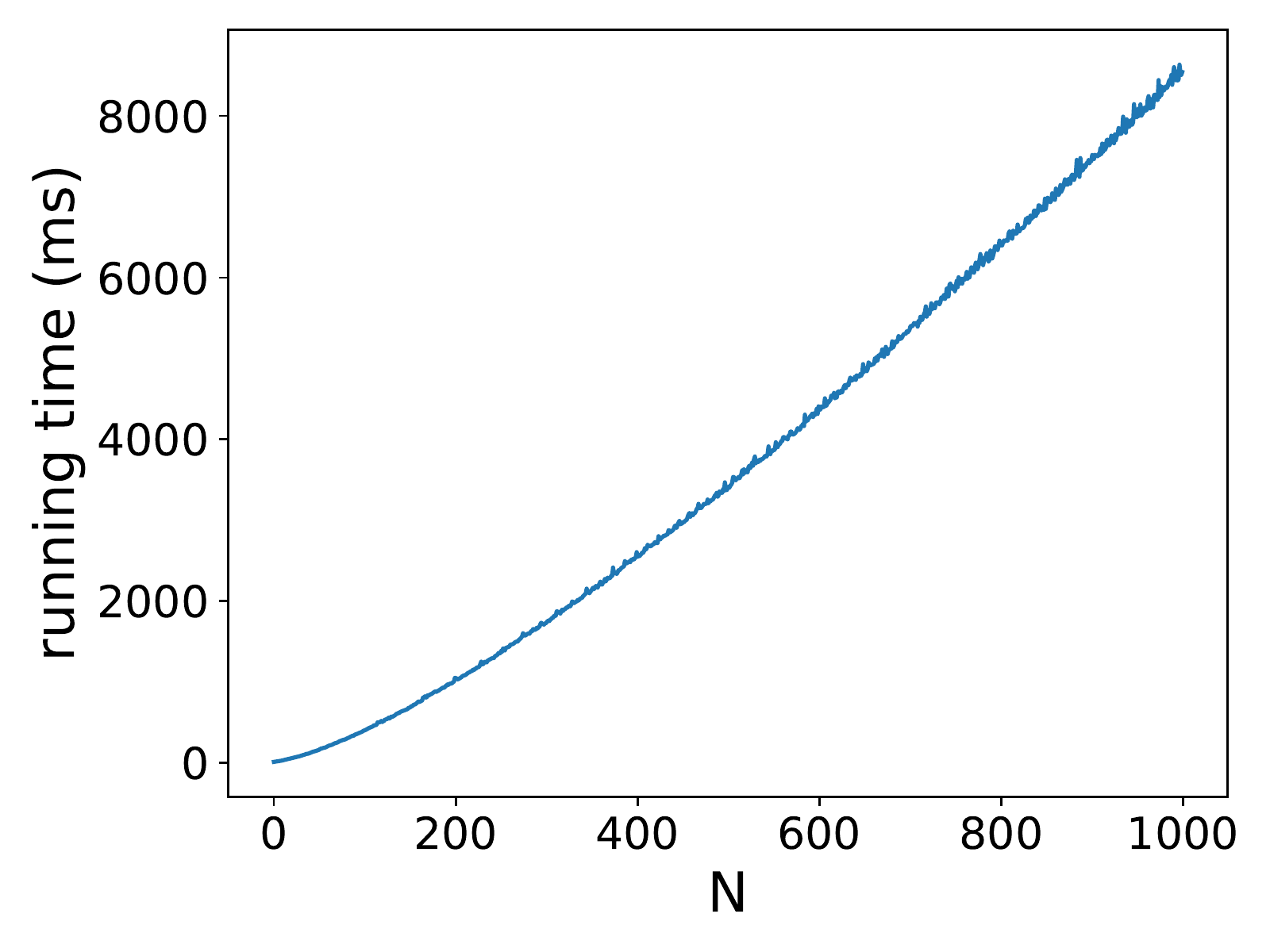}
\caption{Coupon Collector}\label{Fi:ScalabilityCoupons}
\end{subfigure}
\begin{subfigure}{0.49\columnwidth}
\centering
\includegraphics[width=\textwidth]{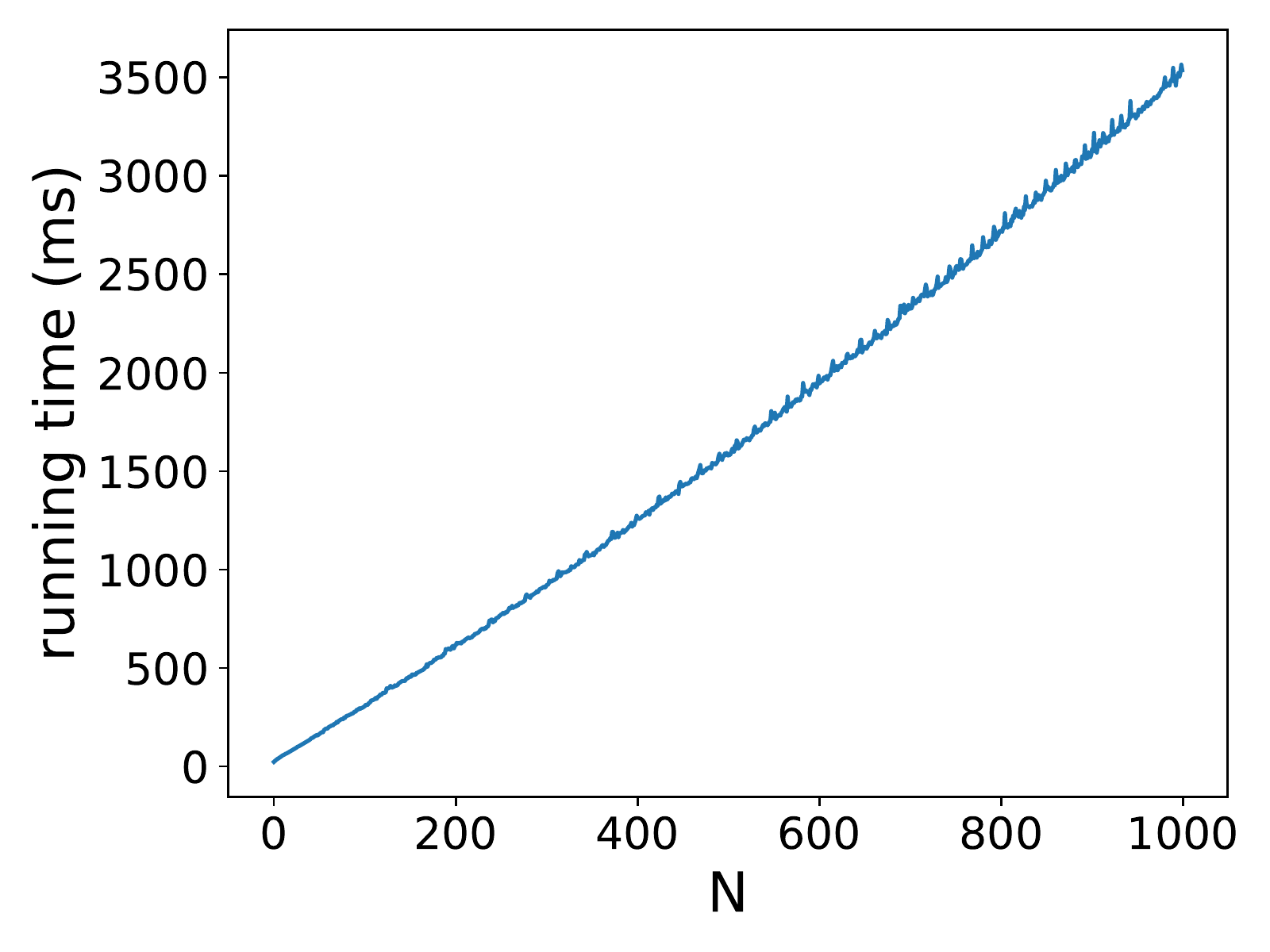}
\caption{Random Walk}\label{Fi:ScalabilityRdwalks}
\end{subfigure}
\caption{Running times of our tool on two sets of synthetic benchmark programs. Each figure plots the runtimes as a function of the size of the analyzed program.}\label{Fi:Scalability}
\end{figure}
\end{changebar}

\paragraph{Discussion.}
\begingroup
\setlength{\columnsep}{7pt}%
Higher central moments can also provide more information about the \emph{shape} of a distribution, e.g.,
the \emph{skewness} (i.e., $\frac{ \expe[(T-\expe[T])^3] }{ (\vari[T])^{\sfrac{3}{2}} }$) indicates how lopsided the distribution of $T$ is,
and the \emph{kurtosis} (i.e., $\frac{ \expe[(T - \expe[T])^4 ] }{ (\vari[T])^2 }$) measures the heaviness of the tails of the distribution of $T$.
We used our tool to analyze two variants of the random-walk program (2-1).
The two random walks have different transition probabilities and step length, but they have the same expected runtime $\expe[T]$.
\cref{Tab:SkewnessKurtosis} presents the skewness and kurtosis derived from the moment bounds inferred by our tool.
\begin{wraptable}{r}{4.2cm}
\centering
\vspace{-8pt}
\caption{Skewness \& kurtosis.}
\vspace{-5pt}
\label{Tab:SkewnessKurtosis}
\begin{small}
\begin{tabular}{@{\hspace{1pt}}c@{\hspace{1pt}}|@{\hspace{1pt}}c@{\hspace{1pt}}|@{\hspace{1pt}}c@{\hspace{1pt}}}
  \hline
  program & skewness & kurtosis \\ \hline
  \textsf{rdwalk-1} & 2.1362 & 10.5633 \\ \hline
  \textsf{rdwalk-2} &  2.9635 &  17.5823 \\ \hline
\end{tabular}
\end{small}
\vspace{-8pt}
\end{wraptable}
A positive skew indicates that the mass of the distribution is concentrated on the \emph{left}, and larger skew means the concentration is more \emph{left}.
A larger kurtosis, on the other hand, indicates that the distribution has \emph{fatter} tails.
Therefore, as the derived skewness and kurtosis indicate, the distribution of the runtime $T$ for \textsf{rdwalk-2} should be more left-leaning and have fatter tails than the distribution for \textsf{rdwalk-1}.
%
Density estimates for the runtime $T$, obtained by simulation, are shown in \cref{Fi:DensityEstimation}.

\endgroup

\newcommand{\gpmax}[2]{(#1) > (#2) ? (#1) : (#2)}
\newcommand{\df}[3]{ gamma(((#2)+(#3))/2.) / (gamma((#2)/2.)*gamma((#3)/2.)) * (((#2)/(#3))^((#2)/2.)) * ((#1)^((#2)/2.-1)) * ((1+((#2)/(#3))*(#1))^(-((#2)+(#3))/2.)) }
\newcommand{\dpearsonVI}[5]{ 1/abs(#5*#2/#3) * (\df{(#1-#4)/(#5*#2/#3)}{2*#2}{2*#3}) }
\newcommand{\dgamma}[3]{ 1./( (#3)^(#2) * gamma(#2) ) * (#1)^((#2)-1) * exp(-((#1)/(#3))) }
\newcommand{\dpearsonV}[4]{ \dgamma{1./(\gpmax{0.0}{sgn(#4) * (#1 - #3)})}{#2}{1./(abs(#4))} / (((#1)-(#3))^2) }
\newcommand{\dpearsonIII}[4]{ \dgamma{sgn(#4) * ((#1) - (#3))}{#2}{abs(#4)} }

\tikzexternalenable
\begin{figure}
\centering
\begin{subfigure}{0.49\columnwidth}
\centering
\includegraphics{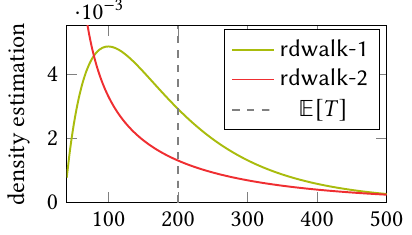}
\end{subfigure}
\begin{subfigure}{0.49\columnwidth}
\centering  
\includegraphics{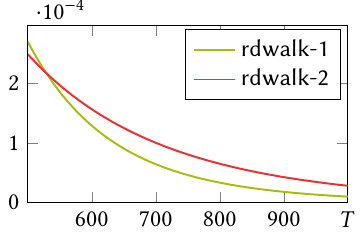}
\end{subfigure}
\caption{Density estimation for the runtime $T$ of two variants \textsf{rdwalk-1} and \textsf{rdwalk-2} of (2-1).}
\label{Fi:DensityEstimation}
\vspace{-10pt}
\end{figure}
\tikzexternaldisable

Our tool can also derive \emph{symbolic} bounds on higher moments.
The table below presents the inferred upper bounds on the variances for the random-walk benchmarks, where we replace the concrete inputs with symbolic pre-conditions.
%
%
\begin{center}\small
  \begin{tabular}{c|c|c}
    \hline
    program & pre-condition & upper bnd. on the variance \\ \hline
    (2-1) & $x \ge 0$ & $1920x$ \\ \hline
    (2-2) & $x \ge 0$ & $2166.6667x+1541.6667$ \\ \hline
  \end{tabular}%
\end{center}

\begin{changebar}
\section{Conclusion}
\label{Se:Conclusion}

We have presented an automatic central-moment analysis for probabilistic programs that support general recursion and continuous distributions,
by deriving symbolic interval bounds on higher raw moments for cost accumulators.
We have proposed \emph{moment semirings} for compositional reasoning and \emph{moment-polymorphic recursion} for interprocedural analysis.
We have proved soundness of our technique using a recent extension to the Optional Stopping Theorem.
The effectiveness of our technique has been demonstrated with our prototype implementation and the analysis of a broad suite of benchmarks, as well as a tail-bound analysis.

In the future, we plan to go beyond arithmetic programs and add support for more datatypes, e.g., Booleans and lists.
We will also work on other kinds of uncertain quantities for probabilistic programs.
Another research direction is to apply our analysis to higher-order functional programs.
\end{changebar}


\begin{acks}
This article is based on research supported, in part, by a gift from Rajiv and Ritu Batra;
by ONR under grants N00014-17-1-2889 and N00014-19-1-2318;
by DARPA under AA contract FA8750-18-C0092;
and by the NSF under SaTC award 1801369, SHF awards 1812876 and 2007784, and CAREER
award 1845514.
Any opinions, findings, and conclusions or recommendations
expressed in this publication are those of the authors,
and do not necessarily reflect the views of the sponsoring
agencies.

We thank Jason Breck and John Cyphert for formulating, and providing
us with an initial formalization of the problem of timing-attack
analysis.
\end{acks}

\balance
\bibliography{db,misc}

\iflong
\clearpage
\appendix
\allowdisplaybreaks

\section{Preliminaries on Measure Theory}
\label{Se:AppendixProbTheory}

Interested readers can refer to textbooks and notes in the literature~\cite{book:Billingsley12,book:Williams91} for more details.

\subsection{Basics}
\label{Se:MeasureTheory}

A \emph{measurable space} is a pair $(S,\calS)$, where $S$ is a nonempty set, and $\calS$ is a \emph{$\sigma$-algebra} on $S$, i.e., a family of subsets of $S$ that contains $\emptyset$ and is closed under complement and countable unions.
The smallest $\sigma$-algebra that contains a family $\calA$ of subsets of $S$ is said to be \emph{generated} by $\calA$, denoted by $\sigma(\calA)$.
Every topological space $(S,\tau)$ admits a \emph{Borel $\sigma$-algebra}, given by $\sigma(\tau)$.
This gives canonical $\sigma$-algebras on $\bbR$, $\bbQ$, $\bbN$, etc.
A function $f : S \to T$, where $(S,\calS)$ and $(T,\calT)$ are measurable spaces, is said to be \emph{$(\calS,\calT)$-measurable}, if $f^{-1}(B) \in \calS$ for each $B \in \calT$.
If $T = \bbR$, we tacitly assume that the Borel $\sigma$-algebra is defined on $T$, and we simply call $f$ \emph{measurable}, or a \emph{random variable}.
Measurable functions form a vector space, and products, maxima, and limiting operations preserve measurability.

A \emph{measure} $\mu$ on a measurable space $(S,\calS)$ is a mapping from $\calS$ to $[0,\infty]$ such that
(i) $\mu(\emptyset) = 0$, and
(ii) for all pairwise-disjoint $\{A_n\}_{n \in \bbZ^+}$ in $\calS$, it holds that $\mu(\bigcup_{n \in \bbZ^+} A_i) = \sum_{n \in \bbZ^+} \mu(A_i)$.
The triple $(S,\calS,\mu)$ is called a \emph{measure space}.
A measure $\mu$ is called a \emph{probability} measure, if $\mu(S) = 1$.
We denote the collection of probability measures on $(S,\calS)$ by $\bbD(S,\calS)$.
For each $x \in S$, the \emph{Dirac measure} $\delta(x)$ is defined as $\lambda A. [x \in A]$.
For measures $\mu$ and $\nu$, we write $\mu + \nu$ for the measure $\lambda A. \mu(A) + \nu(A)$.
For measure $\mu$ and scalar $c \ge 0$, we write $c \cdot \mu$ for the measure $\lambda A. c \cdot \mu(A)$.

The \emph{integral} of a measurable function $f$ on $A \in \calS$ with respect to a measure $\mu$ on $(S,\calS)$ is defined following Lebesgue's theory and is denoted by $\mu(f;A)$, $\int_A f d\mu$, or $\int_A f(x) \mu(dx)$.
If $\mu$ is a probability measure, we call the integral as the \emph{expectation} of $f$, written $\expe_{x \sim \mu}[f; A]$, or simply $\expe[f; A]$ when the scope is clear in the context.
If $A = S$, we tacitly omit $A$ from the notations.
For each $A \in \calS$, it holds that $\mu(f; A) = \mu(f \mathrm{I}_A)$, where $\mathrm{I}_A$ is the indicator function for $A$. 
We denote the collection of \emph{integrable} functions on $(S,\calS,\mu)$ by $\calL^1(S,\calS,\mu)$.

A \emph{kernel} from a measurable space $(S,\calS)$ to another $(T,\calT)$ is a mapping from $S \times \calT$ to $[0,\infty]$ such that:
(i) for each $x \in S$, the function $\lambda B. \kappa(x,B)$ is a measure on $(T,\calT)$, and
(ii) for each $B \in \calT$, the function $\lambda x. \kappa(x,B)$ is measurable.
We write $\kappa : (S,\calS) \rightsquigarrow (T,\calT)$ to declare that $\kappa$ is a kernel from $(S,\calS)$ to $(T,\calT)$.
Intuitively, kernels describe measure transformers from one measurable space to another.
A kernel $\kappa$ is called a \emph{probability} kernel, if $\kappa(x,T) = 1$ for all $x \in S$.
We denote the collection of probability kernels from $(S,\calS)$ to $(T,\calT)$ by $\bbK((S,\calS),(T,\calT))$.
If the two measurable spaces coincide, we simply write $\bbK(S,\calS)$.
We can ``push-forward'' a measure $\mu$ on $(S,\calS)$ to a measure on $(T,\calT)$ through a kernel $\kappa: (S,\calS) \rightsquigarrow (T,\calT)$ by integration:\footnote{
We use a \emph{monad bind} notation $\bind$ here.
Indeed, the category of measurable spaces admits a monad with sub-probability measures~\cite{CATA:Giry82,ENTCS:Panangaden99}.
}
$
\mu \bind \kappa \defeq  \lambda B. \int_S  \kappa(x,B) \mu(dx).
$

\subsection{Product Measures}
\label{Se:ProductMeasures}

The \emph{product} of two measurable spaces $(S,\calS)$ and $(T,\calT)$ is defined as $(S,\calS) \otimes (T,\calT) \defeq (S \times T, \calS \otimes \calT)$, where $\calS \otimes \calT$ is the smallest $\sigma$-algebra that makes coordinate maps measurable, i.e., $\sigma( \{ \pi_1^{-1}(A) \mid A \in \calS\} \cup \{ \pi_2^{-1}(B) \mid B \in \calT \}  ) )$, where $\pi_i$ is the $i$-the coordinate map.
If $\mu_1$ and $\mu_2$ are two probability measures on $(S,\calS)$ and $(T,\calT)$, respectively, then there exists a unique probability measure $\mu$ on $(S,\calS) \otimes (T,\calT)$, by Fubini's theorem, called the \emph{product measure} of $\mu_1$ and $\mu_2$, written $\mu_1 \otimes \mu_2$, such that $\mu(A \times B) = \mu_1(A) \cdot \mu_2(B)$ for all $A \in \calS$ and $B \in \calT$.

If $\mu$ is a probability measure on $(S,\calS)$ and $\kappa : (S,\calS) \rightsquigarrow (T,\calT)$ is a probability kernel, then we can construct the a probability measure on $(S,\calS) \otimes (T,\calT)$ that captures all transitions from $\mu$ through $\kappa$:
$
\mu \otimes \kappa \defeq \lambda(A,B). \int_A \kappa(x,B) \mu (dx).
$
If $\mu$ is a probability measure on $(S_0,\calS_0)$ and $\kappa_i : (S_{i-1},\calS_{i-1}) \rightsquigarrow (S_i,\calS_i)$ is a probability kernel for $i=1,\cdots,n$,where $n \in \bbZ^+$, then we can construct a probability measure on $\bigotimes_{i=0}^n (S_i,\calS_i)$, i.e., the space of sequences of $n$ transitions by iteratively applying the kernels to $\mu$:
\begin{align*}
\mu \otimes \bigotimes\nolimits_{i=1}^{0} \kappa_i & \defeq \mu, \\
\mu \otimes \bigotimes\nolimits_{i=1}^{k+1} \kappa_i & \defeq (\mu \otimes \bigotimes\nolimits_{i=1}^{k} \kappa_i) \otimes \kappa_{k+1}, & 0 \le k < n.
\end{align*}

Let $(S_i,\calS_i)$, $i \in \calI$ be a family of measurable spaces.
Their product, denoted by $\bigotimes_{i \in \calI} (S_i,\calS_i) = (\prod_{i \in \calI} S_i, \bigotimes_{i \in \calI} \calS_i)$, is the product space with the smallest $\sigma$-algebra such that for each $i \in \calI$, the coordinate map $\pi_i$ is measurable.
The theorem below is widely used to construct a probability measures on an infinite product via probability kernels.

\begin{proposition}[Ionescu-Tulcea]\label{Prop:IonescuTulcea}
  Let $(S_i,\calS_i)$, $i \in \bbZ^+$ be a sequence of measurable spaces.
  Let $\mu_0$ be a probability measure on $(S_0,\calS_0)$.
  For each $i \in \bbN$, let $\kappa_i : \bigotimes_{k=0}^{i-1} (S_k,\calS_k) \rightsquigarrow (S_i,\calS_i)$ be a probability kernel.
  Then there exists a sequence of probability measures $\mu_i \defeq \mu_0 \otimes \bigotimes_{k=1}^i \kappa_k$, $i \in \bbZ^+$, and there exists a uniquely defined probability measure $\mu$ on $\bigotimes_{k=0}^\infty (S_k,\calS_k)$ such that $\mu_i(A) = \mu(A \times \prod_{k=i+1}^\infty S_k)$ for all $i \in \bbZ^+$ and $A \in \bigotimes_{k=0}^i \calS_i$.
\end{proposition}

\subsection{Conditional Expectations}
\label{Se:ConditionalExpectations}

Let $(\Omega,\calF,\prob)$ be a measure space, where $\prob$ is a probability measure.
We write $\calL^1$ for $\calL^1(\Omega,\calF,\prob)$.
Let $X \in \calL^1$ be an integrable random variable, and $\calG \subseteq \calF$ be a sub-$\sigma$-algebra of $\calF$.
Then there exists a random variable $Y \in \calL^1$ such that:
(i) $Y$ is $\calG$-measurable, and
(ii) for each $G \in \calG$, it holds that $\expe[Y; G] = \expe[X; G]$.
Such a random variable $Y$ is said to be a version of \emph{conditional expectation} $\expe[X \mid \calG]$ of $X$ given $\calG$.
Conditional expectations admit almost-sure uniqueness.
On the other hand, if $X$ is a nonnegative random variable, the conditional expectation $\expe[X \mid \calG]$ is also defined and almost-surely unique.

Intuitively, for some $\omega \in \Omega$, $\expe[X \mid \calG](\omega)$ is the expectation of $X$ given the set of values $Z(\omega)$ for every $\calG$-measurable random variable $Z$.
For example, if $\calG = \{\emptyset,\Omega\}$, which contains no information, then $\expe[X \mid \calG](\omega) = \expe[X]$.
If $\calG = \calF$, which contains full information, then $\expe[X \mid \calG](\omega) = X(\omega)$.

We review some useful properties of conditional expectations below.

\begin{proposition}\label{Prop:ConditionalExpectations}
  Let $X \in \calL^1$, and $\calG,\calH$ be sub-$\sigma$-algebras of $\calF$.
  \begin{enumerate}[(1)]
    \item If $Y$ is any version of $\expe[X \mid \calG]$, then $\expe[Y] = \expe[X]$.
    
    \item If $X$ is $\calG$-measurable, then $\expe[X \mid \calG] = X$, a.s.
    
    \item If $\calH$ is a sub-$\sigma$-algebra of $\calG$, then $\expe[\expe(X \mid \calG) \mid \calH] = \expe[X \mid \calH]$, a.s.
    
    \item If $Z$ is $\calG$-measurable and $Z \cdot X \in \calL^1$, then $\expe[Z \cdot X \mid \calG] = Z \cdot \expe[X \mid \calG]$, a.s.
    On the other hand, if $X$ is a nonnegative random variable and $Z$ is a nonnegative $\calG$-measurable random variable, then the property also holds, with the understanding that both sides might be infinite. 
  \end{enumerate}
\end{proposition}

\subsection{Convergence Theorems}
\label{Se:OtherUsefulResults}

We review two important convergence theorems for series of random variables.

\begin{proposition}[Monotone convergence theorem]\label{Prop:MON}
  If $\{f_n\}_{n \in \bbZ^+}$ is a non-decreasing sequence of nonnegative measurable functions on a measure space $(S,\calS,\mu)$, and $\{f_n\}_{n \in \bbZ^+}$ converges to $f$ pointwise, then $f$ is measurable and
  $
  \lim_{n \to \infty} \mu(f_n) = \mu(f) \le \infty.
  $
  
  Further, the theorem still holds if $f$ is chosen as a measurable function and ``$\{f_n\}_{n \in \bbZ^+}$ converges to $f$ pointwise'' holds \emph{almost everywhere}, rather than \emph{everywhere}.
\end{proposition}

\begin{proposition}[Dominated convergence theorem]\label{Prop:DOM}
  If $\{f_n\}_{n \in \bbZ^+}$ is a sequence of measurable functions on a measure space $(S,\calS,\mu)$, $\{f_n\}_{n \in \bbZ^+}$ converges to $f$ pointwise, and $\{f_n\}_{n \in \bbZ^+}$ is dominated by a nonnegative integrable function $g$ (i.e., $|f_n(x)| \le g(x)$ for all $n \in \bbZ^+$, $x \in S$), then $f$ is integrable and
  $
  \lim_{n \to \infty} \mu(f_n) = \mu(f).
  $
  
  Further, the theorem still holds if $f$ is chosen as a measurable function and ``$\{f_n\}_{n \in \bbZ^+}$ converges to $f$ pointwise and is dominated by $g$'' holds \emph{almost everywhere}.
\end{proposition}



\section{Operational Cost Semantics}
\label{Se:OperationalCostSemantics}

We follow a distribution-based approach~\cite{ICFP:BLG16,JCSS:Kozen81} to define reduction rules for the semantics of \lang{}.
A probabilistic semantics steps a program configuration to a probability distribution on configurations.
To formally describe these distributions, we need to construct a measurable space of program configurations.
Our approach is to construct a measurable space for each of the four components of configurations, and then use their product measurable space as the semantic domain.
\begin{itemize}
  \item Valuations $\gamma : \mathsf{VID} \to \bbR$ are finite real-valued maps, so we define $(V,\calV) \defeq (\bbR^{\mathsf{VID}},\calB(\bbR^{\mathsf{VID}}))$ as the canonical structure on a finite-dimensional space.
  
  \item The executing statement $S$ can contain real numbers, so we need to ``lift'' the Borel $\sigma$-algebra on $\bbR$ to program statements.
  Intuitively, statements with exactly the same structure can be treated as vectors of parameters that correspond to their real-valued components.
  Formally, we achieve this by constructing a metric space on statements and then extracting a Borel $\sigma$-algebra from the metric space.
  \cref{Fi:MetricsForSemantics} presents a recursively defined metric $d_S$ on statements, as well as metrics $d_E$, $d_L$, and $d_D$ on expressions, conditions, and distributions, respectively, as they are required by $d_S$.
  We denote the result measurable space by $(S,\calS)$.
  
  \item
  A \emph{continuation} $K$ is either an empty continuation $\kstop$, a loop continuation $\kloop{L}{S}{K}$, or a sequence continuation $\kseq{S}{K}$.
  Similarly, we construct a measurable space $(K,\calK)$ on continuations by extracting from a metric space.
  \cref{Fi:MetricsForSemantics} shows the definition of a metric $d_K$ on continuations.
  
  \item The cost accumulator $\alpha \in \bbR$ is a real number, so we define $(W,\calW) \defeq (\bbR,\calB(\bbR))$ as the canonical measurable space on $\bbR$.
\end{itemize}
Then the semantic domain is defined as the product measurable space of the four components: $(\Sigma,\calO) \defeq (V,\calV) \otimes (S,\calS) \otimes (K,\calK) \otimes (W,\calW)$.

\begin{figure*}
  \centering
  \small
  
  \noindent\hrulefill
  
  \begin{minipage}{0.49\textwidth}
  \begin{align*}
    d_E(x,x) & \defeq 0 \\
    d_E(c_1,c_2) & \defeq \abs{c_1-c_2} \\
    d_E(\eadd{E_{11}}{E_{12}},\eadd{E_{21}}{E_{22}}) & \defeq d_E(E_{11},E_{21}) + d_E(E_{12},E_{22}) \\
    d_E(\emul{E_{11}}{E_{12}},\emul{E_{21}}{E_{22}}) & \defeq d_E(E_{11},E_{21}) + d_E(E_{12},E_{22}) \\
    d_E(E_1,E_2) & \defeq \infty~\text{otherwise}
    \end{align*}
  \end{minipage}
  \vrule
  \begin{minipage}{0.49\textwidth}
  \begin{align*}
    d_L(\ctop,\ctop) & \defeq 0 \\
    d_L(\cneg{L_1},\cneg{L_2}) & \defeq d_L(L_1,L_2) \\
    d_L(\cconj{L_{11}}{L_{12}},\cconj{L_{21}}{L_{22}}) & \defeq d_L(L_{11},L_{21}) + d_L(L_{12},L_{22}) \\
    d_L(\cle{E_{11}}{E_{12}},\cle{E_{21}}{E_{22}}) & \defeq d_E(E_{11},E_{21}) + d_E(E_{12},E_{22}) \\
    d_L(L_1,L_2) & \defeq \infty~\text{otherwise}
  \end{align*}
  \end{minipage}
  
  \noindent\hrulefill
  
  \begin{minipage}{1.0\textwidth}
    \begin{align*}
      d_D(\kw{uniform}(a_1,b_1),\kw{uniform}(a_2,b_2)) & \defeq |a_1-a_2| + |b_1-b_2| \\
      d_D(D_1,D_2) & \defeq \infty~\text{otherwise}
    \end{align*}
  \end{minipage}
  
  \noindent\hrulefill
  
  \begin{minipage}{1.0\textwidth}
  \begin{align*}
    d_S(\iskip,\iskip) & \defeq 0 \\
    d_S(\itick{c_1},\itick{c_2}) & \defeq |c_1 - c_2| \\
    d_S(\iassign{x}{E_1},\iassign{x}{E_2}) & \defeq d_E(E_1,E_2) \\
    d_S(\isample{x}{D_1},\isample{x}{D_2}) & \defeq d_D(D_1,D_2) \\
    d_S(\iinvoke{f},\iinvoke{f} ) & \defeq 0 \\
    d_S(\iprob{p_1}{S_{11}}{S_{12}},\iprob{p_2}{S_{21}}{S_{22}}) & \defeq |p_1-p_2| + d_S(S_{11},S_{21}) + d_S(S_{12},S_{22}) \\
    d_S(\icond{L_1}{S_{11}}{S_{12}},\icond{L_2}{S_{21}}{S_{22}}) & \defeq d_L(L_1,L_2) + d_S(S_{11},S_{21}) + d_S(S_{12},S_{22}) \\
    d_S(\iloop{L_1}{S_1}, \iloop{L_2}{S_2}) & \defeq d_L(L_1,L_2)+d_S(S_1,S_2) \\
    d_S(S_{11};S_{12} , S_{21};S_{22}) & \defeq d_S(S_{11},S_{21}) + d_S(S_{12},S_{22}) \\
    d_S(S_1,S_2) & \defeq \infty ~\text{otherwise}
  \end{align*}
  \end{minipage}
  
  \noindent\hrulefill
  
  \begin{minipage}{1.0\textwidth}
  \begin{align*}
    d_K(\kstop, \kstop ) & \defeq 0 \\
    d_K(\kloop{L_1}{S_1}{K_1},\kloop{L_2}{S_2}{K_2}) & \defeq d_L(L_1,L_2) + d_S(S_1,S_2) + d_K(K_1,K_2) \\
    d_K(\kseq{S_1}{K_1},\kseq{S_2}{K_2}) & \defeq d_S(S_1,S_2) + d_K(K_1,K_2) \\
    d_K(K_1,K_2) & \defeq \infty ~\text{otherwise}
  \end{align*}
  \end{minipage}
  
  \noindent\hrulefill
  
  \caption{Metrics for expressions, conditions, distributions, statements, and continuations.}
  \label{Fi:MetricsForSemantics}
\end{figure*}

\cref{Fi:OperationalSemantics} presents the rules of the evaluation relation $\sigma \mapsto \mu$ for \lang{} where $\sigma$ is a configuration and $\mu$ is a probability distribution on configurations.
Note that in \lang{}, expressions $E$ and conditions $L$ are deterministic, so we define a standard big-step evaluation relation for them, written $\gamma \vdash E \Downarrow r$ and $\gamma \vdash L \Downarrow b$, where $\gamma$ is a valuation, $r \in \bbR$, and $b \in \{ \top, \bot \}$.
Most of the rules in \cref{Fi:OperationalSemantics}, except \textsc{(E-Sample)} and \textsc{(E-Prob)}, are also deterministic as they step to a Dirac measure.

\begin{figure*}
\flushleft{\small\fbox{$\gamma \vdash E \Downarrow r $}\;\;\;\;``the expression $E$ evaluates to a real value $r$ under the valuation $\gamma$''}
\begin{mathpar}\small
  \Rule{E-Var}{ \gamma(x) = r }{ \gamma \vdash x \Downarrow r }
  \and
  \Rule{E-Const}{ }{ \gamma \vdash c \Downarrow c }
  \and
  \Rule{E-Add}{ \gamma \vdash E_1 \Downarrow r_1 \\ \gamma \vdash E_2 \Downarrow r_2 \\ r = r_1 + r_2 }{ \gamma \vdash \eadd{ E_1 }{ E_2} \Downarrow r }
  \and
  \Rule{E-Mul}{ \gamma \vdash E_1 \Downarrow r_1 \\ \gamma \vdash E_2 \Downarrow r_2 \\ r = r_1 \cdot r_2}{\gamma \vdash \emul{E_1}{ E_2 }\Downarrow r}
\end{mathpar}
\flushleft{\small\fbox{$\gamma \vdash L \Downarrow b$}\;\;\;\;``the condition $L$ evaluates to a Boolean value $b$ under the valuation $\gamma$''}
\begin{mathpar}\small
  \Rule{E-Top}{ }{ \gamma \vdash \ctop \Downarrow \top }
  \and
  \Rule{E-Neg}{ \gamma \vdash L \Downarrow b }{ \gamma \vdash \cneg{L} \Downarrow \neg b }
  \and
  \Rule{E-Conj}{ \gamma \vdash L_1 \Downarrow b_1 \\ \gamma \vdash L_2 \Downarrow b_2 }{ \gamma \vdash \cconj{L_1}{L_2} \Downarrow b_1 \wedge b_2 }
  \and
  \Rule{E-Le}{ \gamma \vdash E_1 \Downarrow r_1 \\ \gamma \vdash E_2 \Downarrow r_2 }{ \gamma \vdash \cle{E_1}{E_2} \Downarrow [r_1 \le r_2] }
\end{mathpar}
\flushleft{\small\fbox{$\tuple{\gamma,S,K,\alpha} \mapsto \mu$}\;\;\;\;``the configuration $\tuple{\gamma,S,K,\alpha}$ steps to a probability distribution $\mu$ on $\tuple{\gamma',S',K',\alpha'}$'s''}
\begin{mathpar}\small
  \Rule{E-Skip-Stop}{ }{ \tuple{\gamma,\iskip,\kstop,\alpha} \mapsto \delta(\tuple{\gamma,\iskip,\kstop,\alpha}) }
  \and
  \Rule{E-Skip-Loop}{ \gamma \vdash L \Downarrow b }{ \tuple{\gamma, \iskip , \kloop{S}{L}{K} ,\alpha} \mapsto [b] \cdot \delta(\tuple{\gamma,S, \kloop{S}{L}{K} ,\alpha}) + [\neg b] \cdot \delta( \tuple{\gamma,\iskip, K,\alpha} ) }
  \and
  \Rule{E-Skip-Seq}{ }{ \tuple{\gamma,\iskip,\kseq{S}{K},\alpha} \mapsto \delta( \tuple{\gamma,S,K,\alpha} )  }
  \and
  \Rule{E-Tick}{  }{ \tuple{\gamma,\itick{c}  ,K,\alpha} \mapsto \delta(\tuple{\gamma, \iskip ,K,\alpha+c} ) }
  \and
  \Rule{E-Assign}{ \gamma \vdash E \Downarrow r }{ \tuple{\gamma, \iassign{x}{E} ,K,\alpha} \mapsto \delta( \tuple{\gamma[x \mapsto r], \iskip ,K,\alpha} ) }
  \and
  \Rule{E-Sample}{  }{ \tuple{\gamma,\isample{x}{D} ,K,\alpha} \mapsto \mu_D \bind \lambda r. \delta( \tuple{ \gamma[x \mapsto r],  \iskip ,K,\alpha } ) }
  \and
  \Rule{E-Call}{  }{ \tuple{\gamma, \iinvoke{f}  ,K,\alpha} \mapsto \delta( \tuple{\gamma , \scrD(f) , K,\alpha } ) }
  \and
  \Rule{E-Prob}{ }{ \tuple{\gamma, \iprob{p}{S_1}{S_2} ,K,\alpha} \mapsto p \cdot \delta(\tuple{\gamma,S_1,K,\alpha}) + (1-p) \cdot \delta(\tuple{\gamma,S_2,K,\alpha}) }
  \and
  \Rule{E-Cond}{  \gamma \vdash L \Downarrow b }{ \tuple{\gamma,\icond{L}{S_1}{S_2} ,K,\alpha} \mapsto [b] \cdot \delta(\tuple{\gamma,S_1,K,\alpha}) + [\neg b] \cdot \delta(\tuple{\gamma,S_2,K,\alpha}) }
  \and
  \Rule{E-Loop}{ }{ \tuple{\gamma,\iloop{L}{S},K,\alpha} \mapsto \delta( \tuple{ \gamma,\iskip,\kloop{L}{S}{K},\alpha } ) }
  \and
  \Rule{E-Seq}{  }{ \tuple{\gamma,S_1;S_2, K,\alpha } \mapsto \delta( \tuple{ \gamma,S_1,\kseq{S_2}{K},\alpha } ) }
\end{mathpar}  
\caption{Rules of the operational semantics of \lang{}.}
\label{Fi:OperationalSemantics}
\end{figure*}

The evaluation relation $\mapsto$ can be interpreted as a \emph{distribution transformer}.
Indeed, $\mapsto$ can be seen as a probability kernel.

\begin{lemma}\label{Lem:SoundnessForExpCond}
  Let $\gamma :\mathsf{VID} \to \bbR$ be a valuation.
  \begin{itemize}
    \item Let $E$ be an expression. Then there exists a unique $r \in \bbR$ such that $\gamma \vdash E \Downarrow r$.
    \item Let $L$ be a condition. Then there exists a unique $b \in \{ \top, \bot \}$ such that $\gamma \vdash L \Downarrow b$.
  \end{itemize}
\end{lemma}
\begin{proof}
  By induction on the structure of $E$ and $L$.
\end{proof}

\begin{lemma}\label{Lem:EvaluationUniqueness}
  For every configuration $\sigma \in \Sigma$, there exists a unique $\mu\in\bbD(\Sigma,\calO)$ such that $\sigma \mapsto \mu$.
\end{lemma}
\begin{proof}
  Let $\sigma = \tuple{\gamma,S,K,\alpha}$.
  Then by case analysis on the structure of $S$, followed by a case analysis on the structure of $K$ if $S = \iskip$.
  The rest of the proof appeals to \cref{Lem:SoundnessForExpCond}. 
\end{proof}

\begin{theorem}\label{The:EvaluationIsKernel}
  The evaluation relation $\mapsto$ defines a probability kernel on program configurations.
\end{theorem}
\begin{proof}
  \cref{Lem:EvaluationUniqueness} tells us that $\mapsto$ can be seen as a function $\hat{\mapsto}$ defined as follows:
  \[
  \hat{\mapsto}(\sigma,A) \defeq \mu(A) \quad \text{where} \quad \sigma \mapsto \mu.
  \]
  It is clear that $\lambda A. \hat{\mapsto}(\sigma,A)$ is a probability measure.
  On the other hand, to show that $\lambda\sigma. \hat{\mapsto}(\sigma,A)$ is measurable for any $A \in \calO$, we need to prove that for $B \in \calB(\bbR)$, it holds that $\scrO(A,B) \defeq (\lambda\sigma. \hat{\mapsto}(\sigma,A))^{-1}(B) \in \calO$.
  
  We introduce \emph{skeletons} of programs to separate real numbers and discrete structures.
  \begin{align*}
    \hat{S} & \Coloneqq \iskip \mid \itick{\square_\ell} \mid \iassign{x}{\hat{E}} \mid \isample{x}{\hat{D}} \mid \iinvoke{f} \\
    & \mid \iprob{\square_\ell}{\hat{S}_1}{\hat{S}_2} \mid \icond{\hat{L}}{\hat{S}_1}{\hat{S}_2} \\
    & \mid \iloop{\hat{L}}{\hat{S}} \mid \hat{S}_1; \hat{S}_2 \\
    \hat{L} & \Coloneqq \ctop \mid \cneg{\hat{L}} \mid \cconj{\hat{L}_1}{\hat{L}_2} \mid \cle{\hat{E}_1}{\hat{E}_2} \\
    \hat{E} & \Coloneqq x \mid \square_\ell \mid \eadd{\hat{E}_1}{\hat{E}_2} \mid \emul{\hat{E}_1}{\hat{E}_2} \\
    \hat{D} & \Coloneqq \kw{uniform}(\square_{\ell_a},\square_{\ell_b}) \\
    \hat{K} & \Coloneqq \kstop \mid \kloop{\hat{L}}{\hat{S}}{\hat{K}} \mid \kseq{\hat{S}}{\hat{K}}
  \end{align*}
  The \emph{holes} $\square_\ell$ are placeholders for real numbers parameterized by \emph{locations} $\ell \in \m{LOC}$.
  We assume that the holes in a program structure are always pairwise distinct.
  Let $\eta : \m{LOC} \to \bbR$ be a map from holes to real numbers and $\eta(\hat{S})$ (resp., $\eta(\hat{L})$, $\eta(\hat{E})$, $\eta(\hat{D})$, $\eta(\hat{K})$) be the instantiation of a statement (resp., condition, expression, distribution, continuation) skeleton by substituting $\eta(\ell)$ for $\square_\ell$.
  One important property of skeletons is that the ``distance'' between any concretizations of two different skeletons is always infinity with respect to the metrics in \cref{Fi:MetricsForSemantics}.
  
  Observe that
  \[
  \scrO(A,B) = \bigcup_{\hat{S},\hat{K}} \scrO(A,B) \cap \{ \tuple{\gamma,\eta(\hat{S}),\eta(\hat{K}),\alpha} \mid \text{any}~\gamma,\alpha,\eta  \}
  \]
  and that $\hat{S},\hat{K}$ are countable families of statement and continuation skeletons.
  Thus it suffices to prove that every set in the union, which we denote by $\scrO(A,B) \cap \scrC(\hat{S},\hat{K})$ later in the proof, is measurable.
  Note that $\scrC(\hat{S},\hat{K})$ itself is indeed measurable.
  Further, the skeletons $\hat{S}$ and $\hat{K}$ are able to determine the evaluation rule for all concretized configurations.
  Thus we can proceed by a case analysis on the evaluation rules.
  
  To aid the case analysis, we define a deterministic evaluation relation $\xmapsto{\text{det}}$ by getting rid of the $\delta(\cdot)$ notations in the rules in \cref{Fi:OperationalSemantics} except probabilistic ones \textsc{(E-Sample)} and \textsc{(E-Prob)}.
  Obviously, $\xmapsto{\text{det}}$ can be interpreted as a measurable function on configurations.
  
  \begin{itemize}
    \item If the evaluation rule is deterministic, then we have
    \begin{align*}
      & \scrO(A,B) \cap \scrC(\hat{S},\hat{K}) \\
      ={} & \{ \sigma \mid \sigma \mapsto \mu, \mu(A) \in B \} \cap \scrC(\hat{S},\hat{K}) \\
      ={} & \{ \sigma \mid \sigma \xmapsto{\text{det}} \sigma', [\sigma' \in A] \in B \} \cap \scrC(\hat{S},\hat{K}) \\
      ={} &
      \begin{dcases*}
          \scrC(\hat{S},\hat{K}) & if $\{0,1\} \subseteq B$ \\
          {\xmapsto{\text{det}}}^{-1}(A) \cap \scrC(\hat{S},\hat{K}) & if $1 \in B$ and $0 \not\in B$ \\
          {\xmapsto{\text{det}}}^{-1}(A^c) \cap \scrC(\hat{S},\hat{K}) & if $0 \in B$ and $1 \not\in B$ \\
          \emptyset & if $\{0,1\} \cap B = \emptyset$.
      \end{dcases*} 
    \end{align*}
    The sets in all the cases are measurable, so is the set $\scrO(A,B) \cap \scrC(\hat{S},\hat{K})$.
    
    \item \textsc{(E-Prob)}: Consider $B$ with the form $(-\infty,t]$ with $t \in \bbR$.
    If $t \ge 1$, then $\scrO(A,B) = \Sigma$.
    Otherwise, let us assume $t < 1$.
    Let $\hat{S} = \iprob{\square}{\hat{S}_1}{\hat{S}_2}$.
    Then we have
    \begin{align*}
      & \scrO(A,B) \cap \scrC(\hat{S},\hat{K}) \\
      ={} & \{ \sigma \mid \sigma \mapsto \mu, \mu(A) \in B \} \cap \scrC(\hat{S},\hat{K}) \\
      ={} & \{ \sigma \mid \sigma \mapsto p \cdot \delta(\sigma_1) + (1-p) \cdot \delta(\sigma_2),\\
      & \quad p \cdot [\sigma_1 \in A] + (1-p) \cdot [\sigma_2 \in A] \in B \} \cap \scrC(\hat{S},\hat{K}) \\
      ={} & \scrC(\hat{S},\hat{K}) \cap 
       \{  \tuple{\gamma,\iprob{p}{S_1}{S_2},K,\alpha} \mid \\
      & \quad p \cdot [\tuple{\gamma,S_1,K,\alpha} \in A] + (1-p) \cdot [\tuple{\gamma,S_2,K,\alpha} \in A] \le t  \}  \\
      ={} & \scrC(\hat{S},\hat{K}) \cap {} \\
      & \{ \tuple{\gamma,\iprob{p}{S_1}{S_2},K,\alpha} \mid \\
      & \quad (\tuple{\gamma,S_1,K,\alpha} \in A, \tuple{\gamma,S_2,K,\alpha}\not\in A, p \le t)  \vee {} \\
      & \quad (\tuple{\gamma,S_2,K,\alpha} \in A, \tuple{\gamma,S_1,K,\alpha}\not\in A, 1 - p \le t) \}. 
    \end{align*}
    The set above is measurable because $A$ and $A^c$ are measurable, as well as $\{ p \le t\}$ and $\{ p \ge 1 - t \}$ are measurable in $\bbR$.
    
    \item \textsc{(E-Sample)}: Consider $B$ with the form $(-\infty,t]$ with $t \in \bbR$.
    Similar to the previous case, we assume that $t < 1$.
    Let $\hat{S} = \isample{x}{\kw{uniform}(\square_{\ell_a},\square_{\ell_b})}$, without loss of generality.
    Then we have
    \begin{align*}
      & \scrO(A,B) \cap \scrC(\hat{S},\hat{K}) \\
      ={} & \{\sigma \mid \sigma \mapsto \mu, \mu(A) \in B \} \cap \scrC(\hat{S},\hat{K}) \\
      ={} & \{ \sigma \mid \sigma \mapsto \mu_D \bind \kappa_{\sigma}, \int \kappa_{\sigma}(r)(A) \mu_D(dr) \le t \} \cap \scrC(\hat{S},\hat{K}) \\
      ={} & \scrC(\hat{S},\hat{K}) \cap  \{ \sigma \mid \sigma \mapsto \mu_D \bind \kappa_{\sigma}, \\
      & \quad \mu_D(\{ r \mid \tuple{\gamma[x \mapsto r], \iskip,K,\alpha} \in A \}) \le t\} \\
      ={} & \scrC(\hat{S},\hat{K}) \cap \{ \tuple{\gamma,\isample{x}{\kw{uniform}(a,b)},K,\alpha} \mid a<b, \\
      & \quad \mu_{\kw{uniform}(a,b)}(\{ r \mid \tuple{\gamma[x \mapsto r], \iskip,K,\alpha} \in A \}) \le t  \}, 
    \end{align*}
    where $\kappa_{\tuple{\gamma, S ,K,\alpha}} \defeq \lambda r. \delta(\tuple{\gamma[x \mapsto r], \iskip, K, \alpha})$.
    For fixed $\gamma,K,\alpha$, the set $\{r \mid \tuple{\gamma[x \mapsto r],\iskip,K,\alpha} \in A\}$ is measurable in $\bbR$.
    For the distributions considered in this paper, there is a sub-probability kernel $\kappa_D : \bbR^{\mathrm{ar}(D)} \rightsquigarrow \bbR$.
    For example, $\kappa_{\kw{uniform}}(a,b)$ is defined to be $\mu_{\kw{uniform}(a,b)}$ if $a < b$, or $\mathbf{0}$ otherwise.
    Therefore, $\lambda(a,b). \kappa_{\kw{uniform}}(a,b)(\{r \mid \tuple{\gamma[x \mapsto r],\iskip,K,\alpha} \in A\})$ is measurable, and its inversion on $(-\infty,t]$ is a measurable set on distribution parameters $(a,b)$.
    Hence the set above is measurable.
  \end{itemize}
\end{proof}

\section{Trace-Based Cost Semantics}
\label{Se:TraceBasedCostSemantics}

To reason about moments of the accumulated cost, we follow the Markov-chain-based reasoning~\cite{ESOP:KKM16,LICS:OKK16} to develop a trace-based cost semantics for \lang{}.
Let $(\Omega,\calF) \defeq \bigotimes_{n=0}^\infty (\Sigma,\calO)$ be a measurable space of infinite traces on program configurations.
Let $\{\calF_n\}_{n \in \bbZ^+}$ be a filtration, i.e., an increasing sequence $\calF_0 \subseteq \calF_1 \subseteq \cdots \subseteq \calF$ of sub-$\sigma$-algebras in $\calF$, generated by coordinate maps $X_n(\omega) \defeq \omega_n$ for $n \in \bbZ^+$.
Let $\tuple{\scrD,S_{\mathsf{main}}}$ be an \lang{} program.
Let $\mu_0 \defeq \delta(\tuple{ \lambda\_.0 ,S_{\mathsf{main}},\kstop ,0 })$ be the initial distribution.
Let $\prob$ be the probability measure on infinite traces induced by \cref{Prop:IonescuTulcea} and \cref{The:EvaluationIsKernel}.
Then $(\Omega,\calF,\prob)$ forms a probability space on infinite traces of the program.



\section{Trace-Based Reasoning on Expectations}
\label{Se:MarkovChainSemantics}

Recall that we define a \emph{stopping time} $T : \Omega \to \bbZ^+ \cup \{\infty\}$ as
\[
T(\omega) \defeq \inf\{n \in \bbZ^+ \mid \omega_n = \tuple{\_,\iskip,\kstop,\_} \},
\]
and random variables $\{A_n\}_{n \in \bbZ^+}$, $\{\Phi_n\}_{n \in \bbZ^+}$, $\{Y_n\}_{n \in \bbZ^+}$ as
\begin{align*}
  A_n(\omega) \defeq \alpha_n, 
  \Phi_n(\omega)  \defeq \phi(\omega_n), 
  Y_n(\omega) \defeq A_n(\omega) + \Phi_n(\omega),  
\end{align*}
where $\omega_n=\tuple{\_,\_,\_,\alpha_n}$, and $\phi : \Sigma \to \bbR$ is an expected-potential function for expected-cost bound analysis.
Taking the stopping time into consideration, we define the stopped version for these random variables as $A_T(\omega) \defeq A_{T(\omega)}(\omega)$, $\Phi_T(\omega) \defeq \Phi_{T(\omega)}(\omega)$, $Y_T(\omega) \defeq Y_{T(\omega)}(\omega)$.
Note that $A_T,\Phi_T,Y_T$ are not guaranteed to be well-defined everywhere.

\begin{lemma}\label{Lem:StoppedProcessConvergeAS}
 If $\prob[T < \infty] = 1$, i.e., the program terminates \emph{almost surely}, then $A_T$ is well-defined almost surely and $\prob[\lim_{n \to \infty} A_n = A_T] = 1$.
 Further, if $\{ A_n\}_{n \in \bbZ^+}$ is pointwise non-decreasing, then $\lim_{n \to \infty} \expe[A_n] = \expe[A_T]$.
\end{lemma}
\begin{proof}
  By the property of the operational semantics, for $\omega \in \Omega$ such that $T(\omega)<\infty$, we have $A_n(\omega) = A_T(\omega)$ for all $n \ge T(\omega)$.
  Then we have
  \begin{align*}
    & \prob[\lim_{n \to \infty} A_n = A_T] \\
    ={} & \prob( \{ \omega \mid \lim_{n \to \infty} A_n(\omega) = A_T(\omega) \}) \\
    {}\ge{} & \prob(\{ \omega \mid \lim_{n \to \infty} A_n(\omega) = A_T(\omega) \wedge T(\omega) < \infty\}) \\
    ={} & \prob(\{ \omega \mid A_{T(\omega)}(\omega) = A_T(\omega) \wedge T(\omega) < \infty \}) \\
    ={} & \prob(\{ \omega \mid T(\omega) < \infty \}) \\
    ={} & 1. 
  \end{align*}
  
  Now let us assume that $\{A_n\}_{n \in \bbZ^+}$ is pointwise non-decreasing.
  By the property of the operational semantics, we know that $A_0 = 0$.
  Therefore, $A_n$'s are nonnegative random variables, and their expectations $\expe[A_n]$'s are well-defined.
  By \cref{Prop:MON}, we have $\lim_{n \to \infty} \expe[A_n] = \expe[\lim_{n \to \infty} A_n]$.
  We then conclude by the fact that $\lim_{n \to \infty} A_n = A_T$, a.s., which we just proved.
\end{proof}

We reformulate the martingale property using the filtration $\{\calF_n\}_{n \in \bbZ^+}$.

\begin{lemma}\label{Lem:RankingFunctionMartingale}
  For all $n \in \bbZ^+$, it holds that
  \[
  \expe[Y_{n+1} \mid \calF_n] = Y_n, \text{a.s.},
  \]
  i.e., the expectation of $Y_n$ conditioned on the execution history is an invariant for $n \in \bbZ^+$.
\end{lemma}
\begin{proof}
  We say that a sequence of random variables $\{X_n\}_{n \in \bbZ^+}$ is \emph{adapted} to a filtration $\{\calF_n\}_{n \in \bbZ^+}$ if for each $n \in \bbZ^+$, $X_n$ is $\calF_n$-measurable.
  Then $\{\Phi_n\}_{n \in \bbZ^+}$ and $\{A_n\}_{n \in\bbZ^+}$ are adapted to the coordinate-generated filtration $\{\calF_n\}_{n \in \bbZ^+}$ as $\Phi_n(\omega)$ and $A_n(\omega)$ depend on $\omega_n$.
  Then we have
  \begin{align*}
    & \expe[Y_{n+1} \mid \calF_n](\omega) \\
    ={} & \expe[A_{n+1} + \Phi_{n+1} \mid \calF_n](\omega) \\
    ={} & \expe[(A_{n+1} - A_n) +\Phi_{n+1} + A_n \mid \calF_n](\omega) \\
    ={} & \expe[(A_{n+1} - A_n) + \Phi_{n+1} \mid \calF_n](\omega) + A_n(\omega) \\
    ={} & \expe[(\alpha_{n+1} - \alpha_n) + \phi(\omega_{n+1}) \mid \calF_n] + A_n(\omega) \\
    ={} & \expe_{\sigma' \sim {\mapsto}(\omega_n)}[(\alpha'-\alpha_n) + \phi(\sigma')] + A_n(\omega) \\
    ={} & \phi(\omega_n) + A_n(\omega) \\
    ={} & \Phi_n(\omega) + A_n(\omega) \\
    ={} & Y_n(\omega), \text{a.s.} 
  \end{align*}
  Furthermore, we have the following corollary:
  \[
  \expe[Y_{n+1}] = \expe[\expe(Y_{n+1} \mid \calF_n)] = \expe[Y_n],
  \]
  for each $n \in \bbZ^+$, thus $\expe[Y_n] = \expe[Y_0]$ for all $n \in \bbZ^+$.
\end{proof}

Now we can prove soundness of the extended OST.

\begin{theorem}\label{The:ExtendedOST}
  If $\expe[|Y_n|] < \infty$ for all $n \in \bbZ^+$, then $\expe[Y_T]$ exists and $\expe[Y_T] = \expe[Y_0]$ in the following situation:
  
  There exist $\ell \in \bbN$ and $C \ge 0$ such that $\expe[T^\ell] < \infty$ and for all $n \in \bbZ^+$, $|Y_n| \le C \cdot (n+1)^{\ell}$ almost surely.
\end{theorem}
\begin{proof}
  By $\expe[T^\ell] < \infty$ where $\ell \ge 1$, we know that $\prob[T < \infty] = 1$.
  Then similar to the proof of \cref{Lem:StoppedProcessConvergeAS}, we know that $Y_T$ is well-defined almost surely and $\prob[\lim_{n \to \infty} Y_n = Y_T] = 1$.
  On the other hand, we have
  \begin{align*}
    |Y_n| & = |Y_{\min(T,n)}| \le C \cdot (\min(T,n)+1)^\ell \le C \cdot (T + 1)^\ell, \mathrm{a.s.}
  \end{align*}
  Recall that $\expe[T^\ell] < \infty$.
  Then $\expe[(T+1)^\ell] = \expe[T^\ell + O(T^{\ell-1})] < \infty$.
  By \cref{Prop:DOM}, with the function $g$ set to $\lambda\omega. C \cdot (T(\omega)+1)^\ell$, we know that $\lim_{n \to \infty} \expe[Y_n] = \expe[Y_T]$.
  By \cref{Lem:RankingFunctionMartingale}, we have $\expe[Y_n] = \expe[Y_0]$ for all $n\in\bbZ^+$ thus we conclude that $\expe[Y_0] = \expe[Y_T]$.
\end{proof}



\section{Trace-Based Reasoning on Moments}
\label{Se:ProofHigherMoments}

We start with the fundamental composition property for moment semirings.

\begin{lemma*}[\cref{Lem:EisnerBasicProperty}]
  For all $u,v \in \calR$, it holds that
  \[
  \begin{split}
 &  \tuple{ 1, (u + v), (u + v)^2, \cdots, (u + v)^m } \\
 = {} & \tuple{ 1, u, u^2, \cdots,u^m) \otimes (1, v, v^2, \cdots, v^m },
  \end{split}
  \]
  where $u^n$ is an abbreviation for $\prod_{i=1}^n u$ for $n \in \bbZ^+, u \in \calR$.
\end{lemma*}
\begin{proof}
  Observe that
  \begin{align*}
    RHS_k & = \sum_{i=0}^k \binom{k}{i} \times (u^i \cdot v^{k-i}).
  \end{align*}
  We prove by induction on $k$ that $(u + v)^k = RHS_k$.
  \begin{itemize}
    \item $k=0$: Then $(u + v)^0 = 1$.
    On the other hand, we have
    \[
    RHS_0 = \binom{0}{0} \times (u^0 \cdot v^0) = 1 \times (1 \cdot 1) = 1.
    \]
    
    \item Suppose that $(u + v)^k = RHS_k$.
    Then
    \begin{align*}
    & (u + v)^{k+1} \\
    ={} & (u + v) \cdot (u + v)^k \\
    ={} & (u + v) \cdot \sum_{i=0}^k \binom{k}{i} \times (u^i \cdot v^{k-i}) \\
    ={} & \sum_{i=0}^k \binom{k}{i} \times (u^{i+1} \cdot v^{k-i}) +  \sum_{i=0}^k \binom{k}{i} \times (u^i \cdot v^{k-i+1}) \\
    ={} & \sum_{i=1}^{k+1} \binom{k}{i-1} \times (u^i \cdot v^{k-i+1}) + \sum_{i=0}^k \binom{k}{i} \times (u^i \cdot v^{k-i+1}) \\
    ={} & \sum_{i=0}^{k+1} (\binom{k}{i-1} + \binom{k}{i}) \times (u^i \cdot v^{k-i+1}) \\
    ={} & \sum_{i=0}^{k+1} \binom{k+1}{i} \times (u^k \cdot v^{k-i+1}) \\
    ={} & RHS_{k+1}. 
    \end{align*}
  \end{itemize}
\end{proof}

We also show that $\otimes$ and $\oplus$ are monotone if the operations of the underlying semiring are monotone.

\begin{lemma}\label{Lem:MomentMonoidMonotone}
  Let $\calR=(|\calR|,{\le},{+},{\cdot},0,1)$ be a partially ordered semiring.
  If $+$ and $\cdot$ are monotone with respect to $\le$, then $\otimes$ and $\oplus$ in the moment semiring $\calM_\calR^{(m)}$ are also monotone with respect to $\aord$.
\end{lemma}
\begin{proof}
  It is straightforward to show $\oplus$ is monotone.
  For the rest of the proof, without loss of generality, we show that $\many{\tuple{u_k}_{0 \le k \le m}} \otimes \many{\tuple{v_k}_{0 \le k \le m}} \aord \many{\tuple{u_k}_{0 \le k \le m}} \otimes \many{\tuple{w_k}_{0 \le k \le m}}$ if $\many{\tuple{v_k}_{0 \le k \le m}} \aord \many{\tuple{w_k}_{0 \le k \le m}}$.
  By the definition of $\aord$, we know that $v_k \le w_k$ for all $k=0,1,\cdots,m$.
  Then for each $k$, we have
  \begin{align*}
    & (\many{\tuple{u_k}_{0 \le k \le m}} \otimes \many{\tuple{v_k}_{0 \le k \le m}})_k \\
     ={} & \sum_{i=0}^k \binom{k}{i} \times (u_i \cdot v_{k - i}) \\
    {}\le{} & \sum_{i=0}^k \binom{k}{i} \times (u_i \cdot w_{k - i}) \\
    ={} & (\many{\tuple{u_k}_{0 \le k \le m}} \otimes \many{\tuple{w_k}_{0 \le k \le m}})_k. 
  \end{align*}
  Then we conclude by the definition of $\aord$.
\end{proof}

As we allow potential functions to be \emph{interval}-valued, we show that the interval semiring $\calI$ satisfies the monotonicity required in \cref{Lem:MomentMonoidMonotone}.

\begin{lemma}\label{Lem:IntervalSemiringMonotone}
  The operations $+_\calI$ and $\cdot_\calI$ are monotone with respect to $\le_\calI$.
\end{lemma}
\begin{proof}
  It is straightforward to show $+_\calI$ is monotone.
  For the rest of the proof, it suffices to show that $[a,b] \cdot_\calI [c,d] \le_\calI [a',b'] \cdot_\calI [c,d]$ if $[a,b] \le_\calI [a',b']$, i.e., $[a,b] \subseteq [a',b']$ or $a \ge a', b\le b'$.
  
  We claim that $\min S_{a,b,c,d} \ge \min S_{a',b',c,d}$, i.e., $\min \{ ac, ad, bc, bd\} \ge \min \{ a'c,a'd, b'c,b'd \}$.
  \begin{itemize}
    \item If $0 \le c \le d$: Then $ac \le bc$, $ad \le bd$, $a'c \le b'c$, $a'd \le b' d$.
     It then suffices to show that $\min \{ ac, ad\} \ge \min\{ a'c, a'd\}$.
     Because $d \ge c \ge 0$ and $a \ge a'$, we conclude that $ac \ge a'c$ and $ad \ge a'd$.
     
     \item If $c < 0 \le d$: Then $ac \ge bc$, $ad \le bd$, $a'c \ge b'c$, $a'd \ge b'd$.
     It then suffices to show that $\min \{ bc,ad\} \ge \min\{ b'c,a'd\}$.
     Because $d \ge 0 > c$ and $a \ge a', b \le b'$, we conclude that $bc \ge b'c$ and $ad \le a'd$.
      
     \item If $c \le d < 0$: Then $ac \ge bc$, $ad \ge bd$, $a'c \ge b'c$, $a'd \ge b'd$.
     It then suffices to show that $\min\{ bc,bd\} \ge \min\{b'c,b'd\}$.
     Because $0 > d \ge c$ and $b \le b'$, we conclude that $bc \ge b'c$ and $bd \ge b'd$.
  \end{itemize}
  
  In a similar way, we can also prove that $\max S_{a,b,c,d} \le \max S_{a',b',c,d}$.
  Therefore, we show that $\cdot_\calI$ is monotone.
\end{proof}

\begin{lemma}\label{Lem:IntervalOmegaChain}
  If $\{[a_n,b_n]\}_{n \in \bbZ^+}$ is a montone sequence in $\calI$, i.e., $[a_0,b_0] \le_\calI [a_1,b_1] \le_\calI \cdots \le_\calI [a_n,b_n] \le_\calI \cdots$, and $[a_n,b_n] \le_\calI [c,d]$ for all $n \in \bbZ^+$.
  Let $[a,b] = \lim_{n \to \infty} [a_n,b_n]$ (the limit is well-defined by the monotone convergence theorem for series).
  Then $[a,b] \le_\calI [c,d]$.
\end{lemma}
\begin{proof}
  By the definition of $\le_\calI$, we know that $\{a_n\}_{n \in \bbZ^+}$ is non-increasing and $\{b_n\}_{n \in \bbZ^+}$ is non-decreasing.
  Because $a_n \ge c$ for all $n \in \bbZ^+$, we conclude that $\lim_{n \to \infty} a_n \ge c$.
  Because $b_n \le d$ for all $n \in \bbZ^+$, we conclude that $\lim_{n \to \infty} b_n \le d$.
  Thus we conclude that $[a,b] \le_\calI [c,d]$.
\end{proof}

Recall that we extend the notions of $A_n,\Phi_n,Y_n$ with intervals for higher moments as follows:
\begin{align*}
  A_n(\omega) & \defeq \many{\tuple{[\alpha_n^k,\alpha_n^k]}_{0 \le k \le m}} ~ \text{where}~\omega_n =\tuple{\_,\_,\_,\alpha_n}, \\
  \Phi_n(\omega) & \defeq \phi(\omega_n), \\
  Y_n(\omega) & \defeq A_n(\omega) \otimes \Phi_n(\omega)  .
\end{align*}
Note that in the definition of $Y_n$, we use $\otimes$ to compose the powers of the accumulated cost at step $n$ and the potential function that stands for the moments of the accumulated cost of the rest of the computation.

We now extend some of the previous results on first moments to higher moments with intervals.

\begin{lemma}\label{Lem:CondExpeLinear}
  Let $X,Y : \Omega \to \calM_\calI^{(m)}$ be integrable.
  Then $\expe[X \oplus Y] = \expe[X] \oplus \expe[Y]$.
\end{lemma}
\begin{proof}
  Appeal to linearity of expectations and the fact that $\oplus$ is the pointwise extension of $+_\calI$, as well as $+_\calI$ is the pointwise extension of numeric addition.
\end{proof}

\begin{lemma}\label{Lem:CondExpeMoveOutside}
  If $X:\Omega \to \calM_\calI^{(m)}$ is $\calG$-measurable and bounded, a.s., as well as $X(\omega)=\many{\tuple{[a_k(\omega),a_k(\omega)]}_{0 \le k \le m}}$ for all $\omega \in \Omega$, then $\expe[X \otimes Y \mid \calG] = X \otimes \expe[Y \mid \calG]$ almost surely.
\end{lemma}
\begin{proof}
  Fix $\omega \in \Omega$. Let $Y(\omega) = \many{\tuple{[b_k(\omega),c_k(\omega)]}_{0 \le k \le m}}$.
  Then we have
  \begin{align*}
    & \expe[X \otimes Y \mid \calG](\omega) \\
     ={} & \expe[\many{\tuple{[a_k,a_k]}_{0 \le k \le m}} \otimes \many{\tuple{[b_k,c_k]}_{0 \le k \le m}} \mid \calG](\omega) \\
    ={} & \expe[\many{\tuple{ \textstyle \mathop{\sum_\calI}_{i=0}^k \binom{k}{i} \times_\calI ([a_i,a_i] \cdot_\calI [b_{k-i},c_{k-i}])}_{0 \le k \le m} } \mid \calG](\omega) \\
    ={} & \many{ \tuple{ \expe[ \textstyle \mathop{\sum_\calI}_{i=0}^k \binom{k}{i} \times_\calI ([a_i,a_i] \cdot_\calI [b_{k-i},c_{k-i}]) \mid \calG](\omega)  }_{0 \le k \le m}} \\
    ={} & \many{\tuple{ \textstyle \mathop{\sum_\calI}_{i=0}^k \binom{k}{i} \times_\calI  \expe[[a_i,a_i] \cdot_\calI [b_{k-i},c_{k-i}] \mid \calG](\omega) }_{0 \le k \le m} }. 
  \end{align*}
  On the other hand, we have
  \begin{align*}
    & X(\omega) \otimes \expe[Y \mid \calG](\omega) \\
    ={} & \tuple{X_0(\omega),\cdots,X_m(\omega)} \otimes \expe[\tuple{Y_0,\cdots,Y_m} \mid \calG](\omega) \\
    ={} & \tuple{X_0(\omega),\cdots,X_m(\omega)} \otimes \tuple{\expe[Y_0 \mid \calG](\omega), \cdots,\expe[Y_m\mid\calG](\omega)} \\
    ={} & \many{ \tuple{ \textstyle \mathop{\sum_\calI}_{i=0}^k \binom{k}{i} \times_\calI (X_i(\omega) \cdot_\calI \expe[Y_{k-i} \mid \calG](\omega) ) }_{0 \le k \le m} } \\
    ={} & \many{ \tuple{ \textstyle \mathop{\sum_\calI}_{i=0}^k \binom{k}{i} \! \times_\calI \! ([a_i(\omega),a_i(\omega)] \! \cdot_\calI \! \expe[[b_{k-i},c_{k-i}] \mid \calG](\omega) ) }_{0 \le k \le m} }. 
  \end{align*}
  Thus, it suffices to show that for each $i$, $\expe[[a_i,a_i] \cdot_\calI [b_{k-i},c_{k-i}] \mid \calG] = [a_i,a_i] \cdot_\calI \expe[[b_{k-i},c_{k-i}] \mid \calG]$ almost surely.
  
  For $\omega$ such that $a_i(\omega) \ge 0$:
  \begin{align*}
    & \expe[[a_i,a_i] \cdot_\calI [b_{k-i},c_{k-i}] \mid \calG](\omega) \\
    ={} & \expe[[a_ib_{k-i}, a_ic_{k-i}] \mid \calG](\omega) \\
    ={} & [\expe[a_ib_{k-i} \mid \calG](\omega), \expe[a_ic_{k-i} \mid \calG](\omega)] \\
    ={} & [a_i(\omega) \cdot \expe[b_{k-i} \mid \calG](\omega), a_i(\omega) \cdot \expe[c_{k-i} \mid \calG](\omega)], a.s., \\
    ={} & [a_i(\omega),a_i(\omega)] \cdot_\calI \expe[[b_{k-i},c_{k-i}] \mid \calG](\omega). 
  \end{align*}
  For $\omega$ such that $a_i(\omega) < 0$:
  \begin{align*}
    & \expe[[a_i,a_i] \cdot_\calI [b_{k-i},c_{k-i}] \mid \calG](\omega) \\
    ={} & \expe[[a_ic_{k-i}, a_ib_{k-i}] \mid \calG](\omega) \\
    ={} & [\expe[a_ic_{k-i} \mid \calG](\omega), \expe[a_ib_{k-i} \mid \calG](\omega)] \\
    ={} & [a_i(\omega) \cdot \expe[c_{k-i} \mid \calG](\omega), a_i(\omega) \cdot \expe[b_{k-i} \mid \calG](\omega)], a.s., \\
    ={} & [a_i(\omega),a_i(\omega)] \cdot_\calI \expe[[b_{k-i},c_{k-i}] \mid \calG](\omega). 
  \end{align*}
\end{proof}

\begin{lemma*}[\cref{Lem:MomentInvariant}]
  For all $n \in \bbZ^+$, it holds that
  \[
  \expe[Y_{n+1} \mid \calF_n] \aord Y_n, \text{a.s.}
  \]
\end{lemma*}
\begin{proof}
  Similar to the proof of \cref{Lem:RankingFunctionMartingale}, we know that $\{A_n\}_{n \in \bbZ^+}$ and $\{\Phi_n\}_{n \in \bbZ^+}$ are adapted to $\{\calF_n\}_{n \in \bbZ^+}$.
  By the property of the operational semantics, we know that $\alpha_n(\omega) \le C \cdot n$ almost surely for some $C \ge 0$.
  Then using \cref{Lem:CondExpeMoveOutside}, we have
  \begin{align*}
    & \expe[Y_{n+1} \mid \calF_n](\omega) \\
    = {} &\expe[A_{n+1} \otimes \Phi_{n+1} \mid \calF_n](\omega) \\
    ={} &  \expe[A_n \! \otimes \! \many{\tuple{[(\alpha_{n+1}\!-\!\alpha_n)^k,(\alpha_{n+1}\!-\!\alpha_n)^k]}_{0 \le k \le m}} \!\otimes\! \Phi_{n+1} \mid \calF_n ](\omega) \\
    ={} & A_n(\omega) \!\otimes\! \expe[\many{\tuple{[(\alpha_{n+1}\!-\!\alpha_n)^k,(\alpha_{n+1}\!-\!\alpha_n)^k]}_{0 \!\le\! k \!\le\! m}} \!\otimes\! \Phi_{n+1} \!\mid\! \calF_n ](\omega), \\
    & a.s., \\
    ={} & A_n(\omega) \!\otimes\! \expe[\many{\tuple{[(\alpha_{n+1}\!-\!\alpha_n)^k,(\alpha_{n+1}\!-\!\alpha_n)^k]}_{0 \!\le\! k \!\le\! m}} \!\otimes\! \phi(\omega_{n+1}) \!\mid\! \calF_n ]. 
  \end{align*}
  Recall the property of the expected-potential function $\phi$ in \cref{De:IntervalRankingFunction}.
  Then by \cref{Lem:MomentMonoidMonotone} with \cref{Lem:IntervalSemiringMonotone}, we have
  \begin{align*}
    & \expe[Y_{n+1} \mid \calF_n](\omega) \\
     {}\aord{} &  A_n(\omega) \otimes \phi(\omega_n), a.s., \\
    ={} & A_n(\omega) \otimes \Phi_n(\omega) \\
    ={} & Y_n. 
  \end{align*}
  As a corollary, we have $\expe[Y_n] \aord \expe[Y_0]$ for all $n \in \bbZ^+$.
\end{proof}

Now we prove the following extension of OST to deal with interval-valued potential functions. 
Let $\norm{\many{\tuple{[a_k,b_k]}_{0 \le k \le m}}}_\infty \defeq \max_{0 \le k \le m} \{ \max\{ |a_k|,|b_k| \} \} $.

\begin{theorem}\label{The:IntervalOST}
If $\expe[\norm{Y_n}_\infty] < \infty$ for all $n \in \bbZ^+$, then $\expe[Y_T]$ exists and $\expe[Y_T] \aord \expe[Y_0]$ in the following situation:

There exist $\ell \in \bbN$ and $C \ge 0$ such that $\expe[T^\ell] < \infty$ and for all $n \in \bbZ^+$, $\norm{Y_n}_\infty \le C \cdot (n+1)^{\ell}$ almost surely.
\end{theorem}
\begin{proof}
  By $\expe[T^\ell] < \infty$ where $\ell \ge 1$, we know that $\prob[T < \infty] = 1$.
  Then similar to the proof of \cref{Lem:StoppedProcessConvergeAS}, we know that $Y_T$ is well-defined almost surely and $\prob[\lim_{n \to \infty} Y_n = Y_T] = 1$.
  On the other hand, $Y_n(\omega)$ can be treated as a vector of real numbers.
  Let $a_n : \Omega \to \bbR$ be a real-valued component in $Y_n$.
  Because $\expe[\norm{Y_n}_\infty] < \infty$ and $\norm{Y_n}_\infty \le C \cdot (n+1)^\ell$ almost surely, we know that $\expe[|a_n|] \le \expe[\norm{Y_n}_\infty] < \infty$ and $|a_n| \le \norm{Y_n}_\infty \le C \cdot (n+1)^\ell$ almost surely.
  Therefore,
  \[
  |a_n| = |a_{\min(T,n)}| \le C \cdot (\min(T,n) + 1)^\ell \le C \cdot (T+1)^\ell, \text{a.s.}
  \]
  Recall that $\expe[T^\ell] < \infty$.
  Then $\expe[(T+1)^\ell] = \expe[T^\ell + O(T^{\ell-1})] < \infty$.
  By \cref{Prop:DOM}, with the function $g$ set to $\lambda\omega. C \cdot (T(\omega) + 1)^\ell$, we know that $\lim_{n \to \infty} \expe[a_n] = \expe[a_T]$.
  Because $a_n$ is an arbitrary real-valued component in $Y_n$, we know that $\lim_{n \to \infty} \expe[Y_n] = \expe[Y_T]$.
  By \cref{Lem:MomentInvariant}, we know that $\expe[Y_n] \aord \expe[Y_0]$ for all $n \in \bbZ^+$.
  By \cref{Lem:IntervalOmegaChain}, we conclude that $\lim_{n \to \infty} \expe[Y_n] \aord \expe[Y_0]$, i.e., $\expe[Y_T] \aord \expe[Y_0]$.
\end{proof}

\paragraph{Discussion}
The classic OST from probability theory has been used to reason about probabilistic programs:
\begin{proposition}\label{Prop:OST}
  Let $\{X_n\}_{n \in \bbZ^+}$ be a super-martingale (resp., sub-martingale).
  If $\expe[|X_n|] < \infty$ for all $n \in \bbZ^+$, then $\expe[X_T]$ exists and $\expe[X_T] \le \expe[X_0]$ (resp., $\expe[X_T] \ge \expe[X_0]$) in each of the following situations:
  \begin{enumerate}[(a)]
    \item\label{Item:OSTBT} $T$ is almost-surely bounded;
    \item\label{Item:OSTPAST} $\expe[T]\!<\!\infty$ and for some $C \ge 0$, $\expe[|X_{n+1} -X_n| \mid X_n] \le C$ almost surely, for all $n \in \bbZ^+$;
    \item\label{Item:OSTAST} for some $C \ge 0$, $|X_n| \le C$ almost surely, for all $n \in \bbZ^+$.
  \end{enumerate}
\end{proposition}

Note that from \cref{Item:OSTBT} to \cref{Item:OSTAST} in \cref{Prop:OST}, the constraint on the stopping time $T$ is getting weaker while the constraint on $\{X_n\}_{n \in \bbZ^+}$ is getting stronger.
However, the classic OST is not suitable for higher-moment analysis:
\cref{Item:OSTBT} cannot handle almost-surely terminating programs;
\cref{Item:OSTPAST} relies on bounded difference, which is too restrictive, because the change in second moments can be non-constant even if the change in first moments can be bounded by a universal constant;
and \cref{Item:OSTAST} basically requires the accumulated cost can be bounded by a universal constant, which is also not practical.

\begin{counterexample}\label{Exa:PolynomialBound}
  Recall the random-walk program in \cref{Fi:AnnotatedRecursiveRandomWalk}.
  Below, we present the pre- and post-condition of the random-sampling statement:
  \begin{center}
  \begin{small}
    \begin{pseudo}
      ${\color{ACMDarkBlue} \{  \; \langle 1,2(d\!-\!x)\!+\!4, 4(d\!-\!x)^2\!+\!22(d\!-\!x)\!+\!28 \rangle \; \}  }$ \\
      $\isample{t}{ \kw{uniform}({-1},2)}$; \\
      ${\color{ACMDarkBlue} \{  \; \langle 1, 2(d\!-\!x\!-\!t)\!+\!5,4(d\!-\!x\!-\!t)^2\!+\!26(d\!-\!x\!-\!t)\!+\!37 \rangle \; \}}$ 
    \end{pseudo}
  \end{small}
  \end{center}
  Suppose that we want to apply the classic OST to ensuring the second-moment component is sound.
  Let us use $\Phi_n'$ and $Y_n'$ to denote the over-approximations of the second moment in $\Phi_n$ and $Y_n$, respectively.
  It has been shown that for such a random walk, the stopping time $T$ is not bounded but $\expe[T] < \infty$~\cite{PLDI:NCH18,TACAS:KUH19}.
  Thus, the applicable OST criterion could be \cref{Prop:OST}\ref{Item:OSTPAST}, which requires some $C>0$ such that $|Y'_{n+1}-Y'_n| \le C$ for all $n \in \bbZ^+$.
  
  Consider the case where the $n$-th evaluation step of a trace $\omega$ executes the random-sampling statement.
  Let $\alpha_o \defeq A_n(\omega)$.
  By the definition that $Y_n(\omega) \defeq \many{ ([\alpha_o^k,\alpha_o^k])_{0 \le k \le m} } \otimes \Phi_n(\omega)$, we have\footnote{When we use a program variable $x$ in $Y_n(\omega)$, it represents $\gamma(x)$ where $\omega_n = \tuple{\gamma,\_,\_,\_}$.}
  \begin{align*}
    & Y_n'(\omega) \\
    ={} & \alpha_o^2 \!+\! 2 \cdot \alpha_o \!\cdot\! (2(d\!-\!x)\!+\!4) \!+\! (4(d\!-\!x)^2\!+\!22(d\!-\!x)\!+\!28), \\
    & Y_{n+1}'(\omega) \\
    ={} & \alpha_o^2 \!+\! 2 \!\cdot\! \alpha_o \!\cdot\! (2(d\!-\!x\!-\!t)\!+\!5) \!+\! (4(d\!-\!x\!-\!t)^2\!+\!26(d\!-\!x\!-\!t) \!+\!37),  
  \end{align*}
  thus
  \begin{align*}
    & Y_{n+1}'(\omega)\!-\!Y_n'(\omega) \\
    ={} & 2 \!\cdot\! \alpha_o \!\cdot\! (-2t\!+\!1) \!+\! (-4t(2d\!-\!2x\!-\!t) \!+\! 4(d\!-\!x) \!-\!26t\! +\!9), 
  \end{align*}
  where $t \in [{-1},2]$, $d$ is a constant, but $\alpha_o$ and $x$ are generally unbounded,
  thus $|Y_{n+1}'-Y_n'|$ cannot be bounded by a constant.
\end{counterexample}



\section{Soundness of Bound Inference}
\label{Se:SoundnessProof}

\begin{figure*}
\begin{mathpar}\small
  \Rule{Valid-Ctx}{ \Forall{f \in \mathrm{dom}(\Delta)}\Forall{(\Gamma;Q,\Gamma;Q') \in \Delta(f)} \Delta \vdash \{\Gamma;Q\}~\scrD(f)~\{\Gamma';Q'\} }{ \vdash \Delta }
  \and
  \Rule{Q-Skip}{ }{ \Delta \vdash \{ \Gamma;Q \}~ \iskip ~\{ \Gamma;Q \} }
  \and
  \Rule{Q-Tick}{    Q = \many{\tuple{[c^k,c^k]}_{0 \le k \le m}} \otimes Q' }{ \Delta \vdash \{ \Gamma;Q \}~\itick{c}~\{ \Gamma;Q' \} }
  \and
  \Rule{Q-Assign}{ \Gamma = [E/x]\Gamma' \\   Q = [E/x]Q' }{ \Delta \vdash \{ \Gamma;Q  \}~\iassign{x}{E}~\{ \Gamma';Q' \} }
  \and
  \Rule{Q-Sample}{ \Gamma =\Forall{x \in \mathrm{supp}(\mu_D)} \Gamma' \\  Q = \expe_{x \sim \mu_D}[Q'] }{ \Delta \vdash \{  \Gamma;Q   \}~\isample{x}{D}~\{ \Gamma' ;Q' \} }
  \and
  \Rule{Q-Loop}{ \Delta \vdash \{ \Gamma \wedge L;Q \}~S_1~\{ \Gamma;Q \}  }{ \Delta \vdash \{ \Gamma;Q \}~\iloop{L}{S_1}~\{ \Gamma \wedge \neg L; Q \} }
  \and
  \Rule{Q-Cond}{ \Delta \vdash \{ \Gamma \wedge L; Q \}~S_1~\{\Gamma';Q' \} \\\\ \Delta \vdash \{ \Gamma \wedge \neg L; Q \}~S_2~\{\Gamma';Q' \}  }{ \Delta \vdash \{ \Gamma;Q \}~\icond{L}{S_1}{S_2}~\{ \Gamma';Q' \} }
  \and
  \Rule{Q-Seq}{ \Delta \vdash \{\Gamma;Q \}~S_1~\{\Gamma';Q' \} \\\\ \Delta \vdash \{ \Gamma';Q' \}~S_2~\{\Gamma'';Q''\}  }{ \Delta \vdash \{ \Gamma;Q \}~S_1;S_2~\{\Gamma'';Q''\} }
  \and
  \Rule{Q-Call}{  \Forall{i} (\Gamma;Q_i,\Gamma';Q_i') \in \Delta(f)   }{ \Delta \vdash \{ \Gamma; \textstyle \bigoplus_i Q_i \}~\iinvoke{f}~\{ \Gamma'; \textstyle \bigoplus_i Q_i' \} }
  \and
  \Rule{Q-Prob}{  \Delta \vdash \{ \Gamma;Q_1 \}~S_1~\{ \Gamma';Q' \} \\ \Delta \vdash \{ \Gamma;Q_2 \} ~S_2~\{ \Gamma';Q' \} \\\\ Q = P \oplus R \\ P = \tuple{[p,p],[0,0],\cdots,[0,0]} \otimes Q_1 \\ R = \tuple{[1-p,1-p],[0,0],\cdots,[0,0]} \otimes Q_2 }{ \Delta \vdash \{ \Gamma;Q \} ~\iprob{p}{S_1}{S_2}~\{ \Gamma'; Q' \} }
  \and
  \Rule{Q-Weaken}{ \Delta \vdash \{ \Gamma_0;Q_0 \}~S~\{ \Gamma_0';Q_0' \} \\ \Gamma \models \Gamma_0 \\ \Gamma_0' \models \Gamma' \\ \Gamma \models Q \sqsupseteq Q_0 \\ \Gamma_0' \models Q_0' \sqsupseteq Q'  }{ \Delta \vdash \{\Gamma;Q \}~S~\{\Gamma';Q'\} }
  \and
  \Rule{Valid-Cfg}{ \gamma \models \Gamma \\ \Delta \vdash \{\Gamma;Q\}~S~\{\Gamma';Q'\} \\ \Delta \vdash \{ \Gamma'; Q' \}~K }{ \Delta \vdash \{ \Gamma; Q \}~\tuple{\gamma,S,K,\alpha} }
  \and
  \Rule{QK-Stop}{ }{ \Delta \vdash \{ \Gamma ; Q \}~\kstop }
  \and
  \Rule{QK-Loop}{  \Delta \vdash \{ \Gamma \wedge L; Q\}~S~\{\Gamma;Q\} \\ \Delta \vdash \{ \Gamma \wedge \neg L ;Q\}~K }{ \Delta \vdash \{ \Gamma;Q \}~\kloop{L}{S}{K} }
  \and
  \Rule{QK-Seq}{ \Delta \vdash \{ \Gamma;Q\}~S~\{\Gamma';Q'\} \\ \Delta \vdash \{\Gamma';Q'\}~K }{ \Delta \vdash \{ \Gamma;Q \}~\kseq{S}{K} }
  \and
  \Rule{QK-Weaken}{ \Delta \vdash \{\Gamma';Q'\}~K \\ \Gamma \models \Gamma' \\ \Gamma \models Q \sqsupseteq Q' }{ \Delta \vdash \{\Gamma;Q\}~K }
\end{mathpar}
\caption{Inference rules of the derivation system.}
\label{Fi:CompleteInferenceRules}
\end{figure*}

\cref{Fi:CompleteInferenceRules} presents the inference rules of the derivation system.
Note that we prove the soundness on a more declarative system where the two rules for function calls are merged into one.
%
%
%
The validity of a context $\Delta$ for function specifications is then established by the validity of all specifications in $\Delta$, denoted by $\vdash \Delta$.
%

In addition to rules of the judgments for statements and function specifications, we also include rules for continuations and configurations that are used in the operational semantics.
A continuation $K$ is valid with a pre-condition $\{\Gamma;Q\}$, written $\Delta \vdash \{\Gamma;Q\}~K$, if $\phi_Q$ describes a bound on the moments of the accumulated cost of the computation represented by $K$ on the condition that the valuation before $K$ satisfies $\Gamma$.
Validity for configurations, written $\Delta \vdash \{\Gamma;Q\}~\tuple{\gamma,S,K,\alpha}$, is established by validity of the statement $S$ and the continuation $K$, as well as the requirement that the valuation $\gamma$ satisfies the pre-condition $\Gamma$.
Here $\phi_Q$ also describes an interval bound on the moments of the accumulated cost of the computation that continues from the configuration $\tuple{\gamma,S,K,\alpha}$.

The rule \textsc{(Q-Weaken)} and \textsc{(QK-Weaken)} are used to strengthen the pre-condition and relax the post-condition.
In terms of the bounds on moments of the accumulated cost, if the triple $\{\cdot;Q\}~S~\{\cdot;Q'\}$ is valid, then we can safely widen the intervals in the pre-condition $Q$ and narrow the intervals in the post-condition $Q'$.
To handle the judgment $\Gamma \models Q \sqsupseteq Q'$ in constraint generation,
we adapt the idea of \emph{rewrite functions} \cite{PLDI:NCH18,CAV:CHR17}.
Intuitively, to ensure that $[L_1,U_1] \sqsupseteq_{\calP\calI} [L_2,U_2]$, i.e., $L_1 \le L_2$ and $U_2 \le U_1$, under the logical context $\Gamma$, we generate constraints indicating that there exist two polynomials $T_1,T_2$ that are always nonnegative under $\Gamma$, such that $L_1=L_2+T_1$ and $U_1=U_2-T_2$.
%
%
In our implementation, $\Gamma$ is a set of linear constraints over program variables of the form $\calE \ge 0$ , then we can represent $T_1,T_2$ by \emph{conical} combinations (i.e., linear combinations with nonnegative scalars) of expressions $\calE$ in $\Gamma$.


To reduce the soundness proof to the extended OST for interval-valued bounds, we construct an \emph{annotated transition kernel} from validity judgements $\vdash \Delta$ and $\Delta \vdash \{\Gamma;Q\}~S_{\mathsf{main}}~\{\Gamma';Q'\}$.

\begin{lemma}\label{Lem:AnnotatedKernel}
  Suppose $\vdash \Delta$ and $\Delta \vdash \{\Gamma;Q\}~S_{\mathsf{main}}~\{\Gamma';Q'\}$.
  An \emph{annotated program configuration} has the form $\tuple{\Gamma,Q ,\gamma,S,K,\alpha}$ such that $\Delta \vdash \{\Gamma;Q\}~\tuple{\gamma,S,K,\alpha}$.
  Then there exists a probability kernel $\kappa$ over annotated program configurations such that:
  
  For all $\sigma=\tuple{\Gamma,Q,\gamma,S,K,\alpha} \in \mathrm{dom}(\kappa)$, it holds that
  \begin{enumerate}[(i)]
    \item $\kappa$ is the same as the evaluation relation $\mapsto$ if the annotations are omitted, i.e.,
    \[
    \kappa(\sigma) \bind \lambda\tuple{\_,\_,\gamma',S',K',\alpha'}. \delta(\tuple{\gamma',S',K',\alpha'}) = {\mapsto}(\tuple{\gamma,S,K,\alpha}),
    \]
    and
    
    \item $\phi_Q(\gamma) \sqsupseteq \expe_{\sigma' \sim \kappa(\sigma)} [\many{ \tuple{[(\alpha'\!-\!\alpha)^k,(\alpha'\!-\!\alpha)^k]}_{0 \le k \le m}} \otimes \phi_{Q'}(\gamma') ]$ where $\sigma'=\tuple{\_,Q',\gamma',\_,\_,\alpha'}$.
  \end{enumerate}
\end{lemma}

Before proving the soundness, we show that the derivation system for bound inference admits a \emph{relaxation} rule.

\begin{lemma}\label{Lem:TypingRelax}
  Suppose that $\vdash \Delta$.
  If $\Delta \vdash \{\Gamma;Q_1\}~S~\{\Gamma';Q_1'\}$ and $\Delta \vdash \{ \Gamma;Q_2 \}~S~\{\Gamma';Q_2'\}$, then the judgment $\Delta \vdash \{\Gamma;Q_1 \oplus Q_2 \}~S~\{\Gamma'; Q_1' \oplus Q_2' \}$ is derivable.
\end{lemma}
\begin{proof}
  By nested induction on the derivation of the judgment $\Delta \vdash \{\Gamma;Q_1\}~S~\{\Gamma';Q_1'\}$, followed by inversion on $\Delta \vdash \{ \Gamma;Q_2 \}~S~\{\Gamma';Q_2'\}$.
  We assume the derivations have the same shape and the same logical contexts; in practice, we can ensure this by explicitly inserting weakening positions, e.g., all the branching points, and by doing a first pass to obtain logical contexts.
  
  \begin{itemize}
    \item $\small\Rule{Q-Skip}{ }{ \Delta \vdash \{\Gamma;Q_1\}~\iskip~\{\Gamma;Q_1\} }$
    
    By inversion, we have $Q_2=Q_2'$.
    By \textsc{(Q-Skip)}, we immediately have $\Delta \vdash \{\Gamma;Q_1 \oplus Q_2 \}~\iskip~\{\Gamma;Q_1 \oplus Q_2 \}$.
    
    \item $\small\Rule{Q-Tick}{ Q_1 = \many{\tuple{[c^k,c^k]}_{0 \le k \le m}} \otimes Q_1'  }{ \Delta \vdash \{\Gamma;Q_1\}~\itick{c}~\{\Gamma;Q_1'\} }$
    
    By inversion, we have $Q_2 = \many{\tuple{[c^k,c^k]}_{0 \le k \le m}} \otimes Q_2'$.
    By distributivity, we have $\many{\tuple{[c^k,c^k]}_{0 \le k \le m}} \otimes (Q_1' \oplus Q_2') = (\many{\tuple{[c^k,c^k]}_{0 \le k \le m}} \otimes Q_1') \oplus (\many{\tuple{[c^k,c^k]}_{0 \le k \le m}} \oplus  Q_2') = Q_1 \oplus Q_2$.
    Then we conclude by \textsc{(Q-Tick)}.
    
    \item $\small\Rule{Q-Assign}{ \Gamma = [E/x]\Gamma' \\ Q_1 = [E/x]Q_1' }{ \Delta \vdash \{\Gamma;Q_1\}~\iassign{x}{E}~\{\Gamma';Q_1'\} }$
    
    By inversion, we have $Q_2 = [E/x]Q_2'$.
    Then we know that $[E/x](Q_1' \oplus Q_2') = [E/x]Q_1' \oplus [E/x]Q_2' = Q_1 \oplus Q_2$.
    Then we conclude by \textsc{(Q-Assign)}.
    
    \item $\small\Rule{Q-Sample}{ \Gamma = \Forall{x \in \mathrm{supp}(\mu_D)} \Gamma' \\ Q_1 = \expe_{x \sim \mu_D}[Q_1'] }{ \Delta \vdash \{\Gamma;Q_1\}~\isample{x}{D}~\{\Gamma';Q_1'\} }$
    
    By inversion, we have $Q_2 = \expe_{x \sim \mu_D}[Q_2']$.
    By \cref{Lem:CondExpeLinear}, we know that $\expe_{x \sim \mu_D}[Q_1' \oplus Q_2'] = \expe_{x \sim \mu_D}[Q_1'] \oplus \expe_{x \sim \mu_D}[Q_2'] = Q_1 \oplus Q_2$.
    Then we conclude by \textsc{(Q-Sample)}.
    
    \item $\small\Rule{Q-Call}{ \Forall{i} (\Gamma;Q_{1i},\Gamma';Q_{1i}') \in \Delta(f)  }{ \Delta \vdash \{\Gamma; \textstyle \bigoplus_i Q_{1i} \}~\iinvoke{f}~\{\Gamma'; \textstyle \bigoplus_i Q_{1i}' \} }$
    
    By inversion, $Q_2 = \bigoplus_j Q_{2j}$ and $Q_2' = \bigoplus_j Q_{2j}'$ where $(\Gamma;Q_{2j}, \Gamma';Q_{2j}') \in \Delta(f)$ for each $j$.
    Then by \textsc{(Q-Call)}, we have
    $\Delta \vdash \{ \Gamma; \bigoplus_i Q_{1i} \oplus \bigoplus_j Q_{2j} \} ~ \iinvoke{f} ~ \{ \Gamma'; \bigoplus_i Q_{1i}' \oplus \bigoplus_j Q_{2j}' \}$.
    
    \item $\small\Rule{Q-Prob}{ \Delta \vdash \{\Gamma;Q_{11}\}~S_1~\{\Gamma';Q_1'\} \\ \Delta \vdash \{\Gamma;Q_{12}\}~S_2~\{\Gamma';Q_1'\} \\ Q_1 = P_1 \oplus R_1 \\ P_1 = \tuple{[p,p],[0,0],\cdots,[0,0]} \otimes Q_{11} \\ R_1 = \tuple{[1-p,1-p],[0,0],\cdots,[0,0]} \otimes Q_{12} }{ \Delta \vdash \{\Gamma;Q_1\}~\iprob{p}{S_1}{S_2}~\{\Gamma';Q_1'\} }$
    
    By inversion, we know that $Q_2 = P_2 \oplus R_2$ for some $Q_{21},Q_{22}$ such that
    $\Delta \vdash \{ \Gamma; Q_{21} \}~S_1~\{\Gamma';Q_{2}' \}$,
    $\Delta \vdash \{ \Gamma; Q_{22} \}~S_2~\{\Gamma';Q_2' \}$,
    $P_2 = \tuple{[p,p],[0,0],\cdots,[0,0]} \otimes Q_{21}$,
    and $R_2 = \tuple{[1-p,1-p],[0,0],\cdots,[0,0]} \otimes Q_{22}$.
    By induction hypothesis, we have $\Delta \vdash \{\Gamma;Q_{11} \oplus Q_{21} \}~S_1~\{\Gamma'; Q_1' \oplus Q_2' \}$ and $\Delta \vdash \{\Gamma; Q_{12} \oplus Q_{22} \}~S_2~\{\Gamma'; Q_1' \oplus Q_2' \}$.
    Then
    \begin{align*}
      & \tuple{[p,p],\cdots,[0,0]} \otimes (Q_{11} \oplus Q_{21}) \\
      ={} & (\tuple{[p,p],\cdots,[0,0]} \otimes Q_{11}) \oplus (\tuple{[p,p],\cdots,[0,0]} \oplus Q_{21}) \\
      ={} & P_1 \oplus P_2, \\
      & \tuple{[1-p,1-p],\cdots,[0,0]} \otimes (Q_{12} \oplus Q_{22}) \\
      ={} & (\tuple{[1\!-\!p,\!1\!-\!p],\!\cdots\!,\![0,0]} \!\otimes\! Q_{12}) \!\oplus\! (\tuple{[1\!-\!p,\!1\!-\!p],\!\cdots\!,\![0,0]} \!\otimes\! Q_{22}) \\
      ={} & R_1 \oplus R_2, \\
      & (P_1 \oplus P_2) \oplus (R_1 \oplus R_2) \\
      ={} & (P_1 \oplus R_1) \oplus (P_2 \oplus R_2) \\
      ={} & Q_1 \oplus Q_2. 
    \end{align*}
    Thus we conclude by \textsc{(Q-Prob)}.
    
    \item $\small\Rule{Q-Cond}{ \Delta \vdash \{\Gamma \wedge L; Q_1\}~S_1~\{\Gamma';Q_1'\} \\ \Delta \vdash \{\Gamma \wedge \neg L;Q_1\} ~S_2~\{\Gamma';Q_1'\}  }{ \Delta \vdash \{\Gamma;Q_1\}~\icond{L}{S_1}{S_2}~\{\Gamma';Q_1'\} }$
    
    By inversion, we have
    $\Delta \vdash \{ \Gamma \wedge L; Q_2 \}~S_1~\{\Gamma';Q_2'\}$, and
    $\Delta \vdash \{ \Gamma \wedge \neg L; Q_2\} ~ S_2 ~ \{\Gamma'; Q_2'\}$.
    By induction hypothesis, we have $\Delta \vdash \{\Gamma \wedge L; Q_1 \oplus Q_2 \}~S_1~\{\Gamma';Q_1' \oplus Q_2'\}$ and $\Delta \vdash \{\Gamma \wedge \neg L; Q_1 \oplus Q_2 \}~S_2~\{\Gamma'; Q_1' \oplus Q_2' \}$.
    Then we conclude by \textsc{(Q-Cond)}.
    
    \item $\small\Rule{Q-Loop}{ \Delta \vdash \{\Gamma \wedge L;Q_1\}~S~\{\Gamma;Q_1\} }{ \Delta \vdash \{\Gamma;Q_1\}~\iloop{L}{S}~\{\Gamma \wedge \neg L; Q_1\} }$
    
    By inversion, we have $\Delta \vdash \{ \Gamma \wedge L; Q_2 \} ~ S ~ \{\Gamma; Q_2\}$.
    By induction hypothesis, we have $\Delta \vdash \{\Gamma \wedge L; Q_1 \oplus Q_2 \}~S~\{\Gamma; Q_1 \oplus Q_2  \}$.
    Then we conclude by \textsc{(Q-Loop)}.
    
    \item $\small\Rule{Q-Seq}{ \Delta \vdash \{\Gamma;Q_1\}~S_1~\{\Gamma'';Q_1''\} \\ \Delta \vdash \{\Gamma'';Q_1''\}~S_2~\{\Gamma';Q_1'\} }{ \Delta \vdash \{\Gamma;Q_1\}~S_1;S_2~\{\Gamma';Q_1'\} }$
    
    By inversion, there exists $Q_2''$ such that
    $\Delta \vdash \{ \Gamma; Q_2 \}~S_1 ~ \{ \Gamma''; Q_2'' \}$, and
    $\Delta \vdash \{\Gamma'';Q_2'' \} ~S_2 ~ \{ \Gamma'; Q_2' \}$.
    By induction hypothesis, we have $\Delta \vdash \{\Gamma; Q_1 \oplus Q_2 \}~S_1~\{\Gamma''; Q_1'' \oplus Q_2'' \}$ and $\Delta \vdash \{\Gamma''; Q_1'' \oplus Q_2'' \}~S_2~\{\Gamma'; Q_1' \oplus Q_2' \}$.
    Then we conclude by \textsc{(Q-Seq)}.
    
    \item $\small\Rule{Q-Weaken}{ \Delta \vdash \{\Gamma_0;Q_0\}~S~\{\Gamma_0';Q_0'\} \\ \Gamma \models \Gamma_0 \\ \Gamma_0' \models \Gamma' \\ \Gamma \models Q_1 \sqsupseteq Q_0 \\ \Gamma_0' \models Q_0' \sqsupseteq Q_1' }{ \Delta \vdash \{\Gamma;Q_1\}~S~\{\Gamma';Q_1'\} }$
    
    By inversion, there exist $Q_3,Q_3'$ such that
    $\Gamma \models Q_2 \sqsupseteq Q_3$, $\Gamma_0' \models Q_3' \sqsupseteq Q_2'$, and
    $\Delta \vdash \{ \Gamma_0; Q_3 \} ~S~ \{ \Gamma_0'; Q_3' \}$.
    By induction hypothesis, we have $\Delta \vdash \{\Gamma_0; Q_0 \oplus Q_3 \}~S~\{\Gamma_0';Q_0' \oplus Q_3' \}$.
    To apply \textsc{(Q-Weaken)}, we need to show that $\Gamma \models Q_1 \oplus Q_2  \sqsupseteq Q_0 \oplus Q_3$ and $\Gamma_0' \models Q_0' \oplus Q_3' \sqsupseteq  Q_1' \oplus Q_2'$.
    Then appeal to \cref{Lem:IntervalSemiringMonotone,Lem:MomentMonoidMonotone}.
  \end{itemize}
\end{proof}

Now we can construct the annotated transition kernel to reduce the soundness proof to OST.

\begin{proof}[Proof of \cref{Lem:AnnotatedKernel}]
  Let $\nu \defeq {\mapsto}(\tuple{\gamma,S,K,\alpha})$.
  By inversion on $\Delta \vdash \{\Gamma;Q\}~\tuple{\gamma,S,K,\alpha}$, we know that $\gamma \models \Gamma$, $\Delta \vdash \{\Gamma;Q\}~S~\{\Gamma';Q'\}$, and $\Delta \vdash \{\Gamma';Q'\}~K$ for some $\Gamma',Q'$.
  We construct a probability measure $\mu$ as $\kappa(\tuple{\Gamma,Q,\gamma,S,K,\alpha})$ by induction on the derivation of $\Delta \vdash \{\Gamma;Q\}~S~\{\Gamma';Q'\}$.
  
  \begin{itemize}
    \item $\small\Rule{Q-Skip}{ }{ \Delta \vdash \{ \Gamma;Q \}~\iskip~\{\Gamma;Q\} }$
    
    By induction on the derivation of $\Delta \vdash \{\Gamma;Q\}~K$.
    
    \begin{itemize}
      \item $\small\Rule{QK-Stop}{ }{ \Delta \vdash \{ \Gamma;Q \}~\kstop }$
      
      We have $\nu = \delta(\tuple{\gamma,\iskip,\kstop,\alpha})$.
      Then we set $\mu = \delta(\tuple{\Gamma,Q,\gamma,\iskip,\kstop,\alpha })$.
      It is clear that $\phi_Q(\gamma) = \aone \otimes \phi_Q(\gamma)$.
      
      \item $\small\Rule{QK-Loop}{ \Delta \vdash \{\Gamma \wedge L; Q\}~S~\{\Gamma;Q\} \\ \Delta \vdash \{\Gamma \wedge \neg L; Q\}~K }{ \Delta \vdash \{\Gamma;Q\}~\kloop{L}{S}{K} }$
      
      Let $b \in \{\bot,\top\}$ be such that $\gamma \vdash L \Downarrow b$.
      
      If $b = \top$, then $\nu = \delta(\tuple{\gamma,S,\kloop{L}{S}{K},\alpha})$.
      We set $\mu = \delta(\tuple{ \Gamma \wedge L,Q,\gamma,S,\kloop{L}{S}{K},\alpha })$.
      In this case, we know that $\gamma \models \Gamma \wedge L$.
      By the premise, we know that $\Delta \vdash \{\Gamma\wedge L;Q\}~S~\{\Gamma;Q\}$.
      It then remains to show that $\Delta \vdash \{ \Gamma;Q\}~\kloop{L}{S}{K}$.
      By \textsc{(QK-Loop)}, it suffices to show that $\Delta \vdash \{\Gamma \wedge L; Q\}~S~\{\Gamma;Q\}$ and $\Delta \vdash \{\Gamma \wedge \neg L; Q\}~K$.
      Then appeal to the premise.
      
      If $b = \bot$, then $\mu = \delta(\tuple{\gamma,\iskip,K,\alpha})$.
      We set $\mu = \delta(\tuple{\Gamma \wedge \neg L, Q, \gamma,\iskip,K,\alpha})$.
      In this case, we know that $\gamma \models \Gamma \wedge \neg L$.
      By \textsc{(Q-Skip)}, we have $\Delta \vdash \{\Gamma \wedge \neg L;Q\}~\iskip~\{\Gamma \wedge \neg L; Q\}$.
      It then remains to show that $\Delta \vdash \{\Gamma \wedge \neg L; Q\}~K$.
      Then appeal to the premise.
      
      In both cases, $\gamma$ and $Q$ do not change, thus we conclude that $\phi_Q(\gamma) = \aone \otimes \phi_Q(\gamma)$.
      
      \item $\small\Rule{QK-Seq}{ \Delta \vdash \{\Gamma;Q\}~S~\{\Gamma';Q'\} \\ \Delta \vdash \{\Gamma';Q'\}~K }{ \Delta \vdash \{\Gamma;Q\}~\kseq{S}{K} }$
      
      We have $\nu = \delta(\tuple{\gamma,S,K,\alpha})$.
      Then we set $\mu = \delta(\tuple{\Gamma,Q,\gamma,S,K,\alpha})$.
      By the premise, we know that $\Delta \vdash \{\Gamma;Q\}~S~\{\Gamma';Q'\}$ and $\Delta \vdash \{\Gamma';Q'\}~K$.
      Because $\gamma$ and $Q$ do not change, we conclude that $\phi_Q(\gamma) = \aone \otimes \phi_Q(\gamma)$.
      
      \item $\small\Rule{QK-Weaken}{ \Delta \vdash \{\Gamma';Q'\}~K \\ \Gamma \models \Gamma' \\ \Gamma \models Q \sqsupseteq Q' }{ \Delta \vdash \{\Gamma;Q\}~K }$
      
      Because $\gamma \models \Gamma$ and $\Gamma \models \Gamma'$, we know that $\gamma \models \Gamma'$.
      Let $\mu'$ be obtained from the induction hypothesis on $\Delta \vdash \{\Gamma';Q'\}~K$.
      Then $\phi_{Q'}(\gamma) \sqsupseteq \expe_{\sigma' \sim \mu'}[\many{\tuple{[(\alpha'-\alpha)^k,(\alpha'-\alpha)^k]}_{0 \le k \le m}} \otimes \phi_{Q''}(\gamma')]$, where $\sigma'=\tuple{\_,Q'',\gamma',\_,\_,\alpha'}$.
      We set $\mu = \mu'$.
      Because $\Gamma \models Q \sqsupseteq Q'$ and $\gamma \models \Gamma$, we conclude that $\phi_Q(\gamma) \sqsupseteq \phi_{Q'}(\gamma)$.
    \end{itemize}

  \item $\small\Rule{Q-Tick}{ Q = \many{\tuple{[c^k,c^k]}_{0 \le k \le m}} \otimes Q' }{ \Delta \vdash \{\Gamma;Q\}~\itick{c}~\{\Gamma;Q'\} }$
  
  We have $\nu = \delta(\tuple{\gamma,\iskip,K,\alpha+c})$.
  Then we set $\mu = \delta(\tuple{\Gamma,Q',\gamma,\iskip,K,\alpha+c})$.
  By \textsc{(Q-Skip)}, we have $\Delta \vdash \{\Gamma;Q'\}~\iskip~\{\Gamma;Q'\}$.
  Then by the assumption, we have $\Delta \vdash \{\Gamma;Q'\}~K$.
  It remains to show that $\phi_Q(\gamma) \sqsupseteq \many{\tuple{[c^k,c^k]}_{0 \le k \le m}} \otimes \phi_{Q'}(\gamma)$.
  Indeed, we have $\phi_Q(\gamma) = \many{\tuple{[c^k,c^k]}_{0 \le k \le m}} \otimes \phi_{Q'}(\gamma)$ by the premise.
  
  \item $\small\Rule{Q-Assign}{ \Gamma = [E/x]\Gamma' \\ Q = [E/x]Q' }{ \Delta \vdash \{\Gamma;Q\}~\iassign{x}{E}~\{\Gamma';Q'\} }$
  
  Let $r \in \bbR$ be such that $\gamma \vdash E \Downarrow r$.
  We have $\nu = \delta(\tuple{\gamma[x \mapsto r],\iskip,K,\alpha})$.
  Then we set $\mu = \delta(\tuple{\Gamma',Q',\gamma[x \mapsto r],\iskip,K,\alpha})$.
  Because $\gamma \vdash \Gamma$, i.e., $\gamma \vdash [E/x]\Gamma'$, we know that $\gamma[x \mapsto r] \vdash \Gamma'$.
  By \textsc{(Q-Skip)}, we have $\Delta \vdash \{\Gamma';Q'\}~\iskip~\{\Gamma';Q'\}$.
  Then by the assumption, we have $\Delta \vdash \{\Gamma';Q'\}~K$.
  It remains to show that $\phi_Q(\gamma) = \aone \otimes \phi_{Q'}(\gamma[x \mapsto r])$.
  By the premise, we have $Q = [E/x]Q'$, thus $\phi_Q(\gamma) = \phi_{[E/x]Q'}(\gamma) = \phi_{Q'}(\gamma[x \mapsto r])$.
  
  \item $\small\Rule{Q-Sample}{ \Gamma = \Forall{x \in \mathrm{supp}(\mu_D)} \Gamma' \\ Q = \expe_{x \sim \mu_D}[Q'] }{ \Delta \vdash \{\Gamma;Q\}~\isample{x}{D}~\{\Gamma';Q'\} }$
  
  We have $\nu = \mu_D \bind \lambda r. \delta(\tuple{\gamma[x \mapsto r],\iskip,K,\alpha})$.
  Then we set $\mu = \mu_D \bind \lambda r. \delta(\tuple{\Gamma',Q',\gamma[x \mapsto r],\iskip,K,\alpha})$.
  For all $r \in \mathrm{supp}(\mu_D)$, because $\gamma \models \Forall{x \in \mathrm{supp}(\mu_D)} \Gamma'$, we know that $\gamma[x \mapsto r] \models \Gamma'$.
  By \textsc{(Q-Skip)}, we have $\Delta \vdash \{\Gamma';Q'\}~\iskip~\{\Gamma';Q'\}$.
  Then by the assumption, we have $\Delta \vdash \{\Gamma';Q'\}~K$.
  It remains to show that $\phi_Q(\gamma) \sqsupseteq \expe_{r \sim \mu_D}[\aone \otimes \phi_{Q'}(\gamma[x \mapsto r])]$.
  Indeed, because $Q = \expe_{x \sim \mu_D}[Q']$, we know that $\phi_Q(\gamma) = \phi_{\expe_{x \sim \mu_D}[Q']}(\gamma) = \expe_{r \sim \mu_D}[\phi_{Q'}(\gamma[x \mapsto r])]$.
  
  \item $\small\Rule{Q-Call}{ \Forall{i} (\Gamma;Q_i,\Gamma';Q_i') \in \Delta(f) }{ \Delta \vdash \{\Gamma; \textstyle \bigoplus_i Q_i \}~\iinvoke{f}~\{\Gamma'; \textstyle \bigoplus_i Q_i' \} }$
  
  We have $\nu = \delta(\tuple{\gamma,\scrD(f),K,\alpha})$.
  Then we set $\mu = \delta(\tuple{\Gamma, \bigoplus_i Q_i ,\gamma,\scrD(f),K,\alpha})$.
  By the premise, we know that $\Delta \vdash \{\Gamma;Q_i\}~\scrD(f)~\{\Gamma';Q_i'\}$ for each $i$.
  By \cref{Lem:TypingRelax} and simple induction, we know that $\Delta \vdash \{\Gamma; \bigoplus_i Q_i \}~\scrD(f)~\{\Gamma'; \bigoplus_i Q_i' \} $.
  Because $\gamma$ and $\bigoplus_i Q_i$ do not change, we conclude that $\phi_{\bigoplus_i Q_i}(\gamma) = \aone \otimes \phi_{\bigoplus_i Q_i}(\gamma)$.
  
  \item $\small\Rule{Q-Prob}{ \Delta \vdash \{\Gamma;Q_1\}~S_1~\{\Gamma';Q'\} \\ \Delta \vdash \{\Gamma;Q_2\}~S_2~\{\Gamma';Q'\} \\ Q = P \oplus R \\ P = \tuple{[p,p],[0,0],\cdots,[0,0]} \otimes Q_1 \\ R = \tuple{[1-p,1-p],[0,0],\cdots,[0,0]} \otimes Q_2 }{ \Delta \vdash \{\Gamma;Q\}~\iprob{p}{S_1}{S_2} ~\{\Gamma';Q'\} }$
  
  We have $\nu = p \cdot \delta(\tuple{\gamma,S_1,K,\alpha}) + (1-p) \cdot \delta(\tuple{\gamma,S_2,K,\alpha})$.
  Then we set $\mu = p \cdot \delta(\tuple{\Gamma,Q_1,\gamma,S_1,K,\alpha} + (1-p) \cdot \delta(\tuple{\Gamma,Q_2,\gamma,S_2,K,\alpha})$.
  From the assumption and the premise, we know that $\gamma \models \Gamma$, $\Delta \vdash \{\Gamma';Q'\}~K$, and $\Delta \vdash \{\Gamma;Q_1\}~S_1~\{\Gamma';Q'\}$, $\Delta \vdash \{\Gamma;Q_2\}~S_2~\{\Gamma';Q'\}$.
  It remains to show that $\phi_Q(\gamma)_k \ge_{\calI} (p \cdot \phi_{Q_1}(\gamma)_k) +_\calI ((1-p) \cdot \phi_{Q_2}(\gamma)_k)$, where the scalar product $p \cdot [a,b] \defeq [pa,pb]$ for $p \ge 0$.
  On the other hand, from the premise, we have $Q_k = P_k +_{\calP\calI} R_k$ and $P_k = (([p,p],\cdots,[0,0]) \otimes Q_1)_k = \binom{k}{0} \times_{\calP\calI} ([p,p] \cdot_{\calP\calI} (Q_1)_k) = p \cdot (Q_1)_k$, as well as $R_k = (1-p) \cdot (Q_2)_k$.
  Therefore, we have $\phi_Q(\gamma)_k = \phi_{\many{\tuple{P_k +_{\calP\calI} R_k}_{0 \le k \le m}}}(\gamma)_k = p \cdot \phi_{Q_1}(\gamma)_k +_\calI (1-p) \cdot \phi_{Q_2}(\gamma)_k$.
  
  \item $\small\Rule{Q-Cond}{ \Delta \vdash \{\Gamma \wedge L;Q\}~S_1~\{\Gamma';Q'\} \\ \Delta \vdash \{\Gamma \wedge \neg L; Q\}~S_2~\{\Gamma';Q'\} }{ \Delta \vdash \{\Gamma;Q\}~\icond{L}{S_1}{S_2}~\{\Gamma';Q'\} }$
  
  Let $b \in \{\top,\bot\}$ be such that $\gamma \vdash L \Downarrow b$.
  
  If $b = \top$, then $\nu = \delta(\tuple{\gamma,S_1,K,\alpha})$.
  We set $\mu = \delta(\tuple{\Gamma \wedge L,Q,\gamma,S_1,K,\alpha})$.
  In this case, we know that $\gamma \models \Gamma \wedge L$.
  By the premise and the assumption, we know that $\Delta \vdash \{\Gamma \wedge L; Q\}~S_1~\{\Gamma';Q'\}$ and $\Delta \vdash \{\Gamma';Q'\}~K$.
  
  If $b = \bot$, then $\nu = \delta(\tuple{\gamma,S_2,K,\alpha})$.
  We set $\mu = \delta(\tuple{\Gamma \wedge \neg L,Q,\gamma,S_2,K,\alpha})$.
  In this case, we know that $\gamma \models \Gamma \wedge \neg L$.
  By the premise and the assumption, we know that $\Delta \vdash \{\Gamma \wedge \neg L; Q\}~S_2~\{\Gamma';Q'\}$ and $\Delta \vdash \{\Gamma';Q'\}~K$.
  
  In both cases, $\gamma$ and $Q$ do not change, thus we conclude that $\phi_Q(\gamma) = \aone \otimes \phi_Q(\gamma)$.
  
  \item $\small\Rule{Q-Loop}{ \Delta \vdash \{\Gamma \wedge L;Q\}~S~\{\Gamma;Q\} }{ \Delta \vdash \{\Gamma;Q\}~\iloop{L}{S_1}~\{\Gamma \wedge \neg L;Q \} }$
  
  We have $\nu = \delta(\tuple{\gamma,\iskip,\kloop{L}{S}{K},\alpha})$.
  Then we set $\mu = \delta(\tuple{\Gamma,Q,\gamma,\iskip,\kloop{L}{S}{K},\alpha})$.
  By \textsc{(Q-Skip)}, we have $\Delta \vdash \{\Gamma;Q\}~\iskip~\{\Gamma;Q\}$.
  Then by the assumption $\Delta \vdash \{\Gamma \wedge \neg L;Q\}~K$ and the premise, we know that $\Delta \vdash \{\Gamma;Q\}~\kloop{L}{S}{K}$ by \textsc{(QK-Loop)}.
  Because $\gamma$ and $Q$ do not change, we conclude that $\phi_Q(\gamma) = \aone \otimes \phi_Q(\gamma)$.
  
  \item $\small\Rule{Q-Seq}{ \Delta \vdash \{\Gamma;Q\}~S_1~\{\Gamma';Q'\} \\ \Delta \vdash \{\Gamma';Q'\}~S_2~\{\Gamma'';Q''\} }{ \Delta \vdash \{\Gamma;Q\}~S_1;S_2~\{\Gamma'';Q''\} }$
  
  We have $\nu = \delta(\tuple{\gamma,S_1,\kseq{S_2}{K},\alpha})$.
  Then we set $\mu = \delta(\tuple{\Gamma,Q,\gamma,S_1,\kseq{S_2}{K},\alpha})$.
  By the first premise, we have $\Delta \vdash \{\Gamma;Q\}~S_1~\{\Gamma';Q'\}$.
  By the assumption $\Delta \vdash \{\Gamma'';Q''\}~K$ and the second premise, we know that $\Delta \vdash \{\Gamma';Q'\}~\kseq{S_2}{K}$ by \textsc{(QK-Seq)}.
  Because $\gamma$ and $Q$ do not change, we conclude that $\phi_Q(\gamma) = \aone \otimes \phi_Q(\gamma)$. 
  
  \item $\small\Rule{Q-Weaken}{ \Delta \vdash \{\Gamma_0;Q_0\}~S~\{\Gamma_0';Q_0'\} \\ \Gamma \models \Gamma_0 \\ \Gamma_0' \models \Gamma' \\ \Gamma \models Q \sqsupseteq Q_0 \\ \Gamma_0' \models Q_0' \sqsupseteq Q' }{ \Delta \vdash \{\Gamma;Q\}~S~\{\Gamma';Q'\} }$
  
  By $\gamma \models \Gamma$ and $\Gamma \models \Gamma_0$, we know that $\gamma \models \Gamma_0$.
  By the assumption $\Delta \vdash \{\Gamma';Q'\}~K$ and the premise $\Gamma_0' \models \Gamma'$, $\Gamma_0' \models Q_0' \sqsupseteq Q'$, we derive $\Delta \vdash \{\Gamma_0';Q_0'\}~K$ by \textsc{(QK-Weaken)}.
  Thus let $\mu_0$ be obtained by the induction hypothesis on $\Delta \vdash \{\Gamma_0;Q_0\}~S~\{\Gamma_0';Q_0'\}$.
  Then $\phi_{Q_0}(\gamma) \sqsupseteq \expe_{\sigma' \sim \mu_0}[\many{\tuple{[(\alpha'-\alpha)^k,(\alpha'-\alpha)^k]}_{0 \le k \le m}} \otimes \phi_{Q''}(\gamma')]$, where $\sigma' = \tuple{\_,Q'',\gamma',\_,\_,\alpha'}$.
  We set $\mu = \mu_0$. 
  By the premise $\Gamma \models Q \sqsupseteq Q_0$ and $\gamma \models \Gamma$, we conclude that $\phi_{Q}(\gamma) \sqsupseteq \phi_{Q_0}(\gamma)$.
  \end{itemize}
\end{proof}

Therefore, we can use the annotated kernel $\kappa$ above to re-construct the trace-based moment semantics in \cref{Se:TraceBasedCostSemantics}.
Then we can define the potential function on annotated program configurations as $\phi(\sigma) \defeq \phi_Q(\gamma)$ where $\sigma=\tuple{\_,Q,\gamma,\_,\_,\_}$.

The next step is to apply the extended OST for interval bounds (\cref{The:IntervalOST}).
Recall that the theorem requires that for some $\ell \in \bbN$ and $C \ge 0$, $\norm{Y_n}_\infty \le C \cdot (n+1)^{\ell}$ almost surely for all $n \in \bbZ^+$.
One sufficient condition for the requirement is to assume the \emph{bounded-update} property, i.e., every (deterministic or probabilistic) assignment to a program variable updates the variable with a bounded change.
As observed by~\citet{PLDI:WFG19}, bounded updates are common in practice.
We formulate the idea as follows.

\begin{lemma}\label{Lem:BoundedUpdate}
  If there exists $C_0 \ge 0$ such that for all $n \in \bbZ^+$ and $x \in \mathsf{VID}$, it holds that $\prob[\abs{\gamma_{n+1}(x)- \gamma_n(x) } \le C_0] = 1$ where $\omega$ is an infinite trace, $\omega_n = \tuple{\gamma_n,\_,\_,\_}$, and $\omega_{n+1} = \tuple{\gamma_{n+1},\_,\_,\_}$,
  then there exists $C \ge 0$ such that  for all $n \in \bbZ^+$, $\norm{Y_n}_\infty \le C \cdot (n+1)^{md}$ almost surely.
\end{lemma}
\begin{proof}
  Let $C_1 \ge 0$ be such that for all $\itick{c}$ statements in the program, $|c| \le C_1$.
  Then for all $\omega$, if $\omega_n = \tuple{\_,\_,\_,\alpha_n}$, then $|\alpha_n| \le n \cdot C_1$.
  On the other hand, we know that $\prob[|\gamma_n(x) - \gamma_0(x)| \le C_0 \cdot n] = 1$ for any variable $x$.
  As we assume all the program variables are initialized to zero, we know that $\prob[|\gamma_n(x)| \le C_0 \cdot n] = 1$.
  From the construction in the proof of \cref{Lem:AnnotatedKernel}, we know that all the templates used to define the interval-valued potential function should have almost surely bounded coefficients.
  Let $C_2 \ge 0$ be such a bound.
  Also, the $k$-th component in a template is a polynomial in $\bbR_{kd}[\mathsf{VID}]$.
  Therefore, $\Phi_n(\omega) = \phi(\omega_n) = \phi_{Q_n}(\gamma_n)$, and
  \[
  |\phi_{Q_n}(\gamma_n)_k| \le \sum_{i=0}^{kd} C_2 \cdot |\mathsf{VID}|^{i} \cdot |C_0 \cdot n|^i \le C_3 \cdot (n+1)^{kd}, \textrm{a.s.},
  \]
  for some sufficiently large constant $C_3$.
  Thus
  \begin{align*}
  & |(Y_n)_k| = |(A_n \otimes \Phi_n)_k| \\
  ={} & |\mathop{{\sum}_\calI}_{i=0}^k \binom{k}{i} \times_\calI ((A_n)_i \cdot_\calI (\Phi_n)_{k-i})| \\
  {}\le{} & \sum_{i=0}^k \binom{k}{i} \cdot (n \cdot C_1)^i \cdot (C_3 \cdot (n+1))^{(k-i)d} \\
  {}\le{} & C_4 \cdot (n+1)^{kd}, \textrm{a.s.}, 
  \end{align*}
  for some sufficiently large constant $C_4$.
  Therefore $\norm{Y_n}_\infty \le C_5 \cdot (n+1)^{md}$, a.s., for some sufficiently large constant $C_5$.
\end{proof}

Now we prove the soundness of bound inference.

\begin{theorem*}[\cref{The:Soundness}]
  Suppose $\Delta \vdash \{\Gamma;Q\}~S_{\mathsf{main}}~\{\Gamma';\aone\}$ and $\vdash \Delta$.
  Then $\expe[A_T] \aord \phi_Q(\lambda\_.0)$, i.e., the moments $\expe[A_T]$ of the accumulated cost upon program termination are bounded by intervals in $\phi_Q(\lambda\_.0)$ where $Q$ is the quantitative context and $\lambda\_.0$ is the initial valuation, if \emph{both} of the following properties hold:
  \begin{enumerate}[(i)]
    \item $\expe[T^{md}] < \infty$, and
    
    \item there exists $C_0 \ge 0$ such that for all $n \in \bbZ^+$ and $x \in \mathsf{VID}$, it holds almost surely that $|\gamma_{n+1}(x) - \gamma_n(x)| \le C_0$ where $\tuple{\gamma_n,\_,\_,\_} = \omega_n$ and $\tuple{\gamma_{n+1},\_,\_,\_} = \omega_{n+1}$ of an infinite trace $\omega$.
  \end{enumerate}
\end{theorem*}
\begin{proof}
  By \cref{Lem:BoundedUpdate}, there exists $C \ge 0$ such that $\norm{Y_n}_\infty \le C \cdot (n+1)^{md}$ almost surely for all $n \in \bbZ^+$.
  By the assumption, we also know that $\expe[T^{md}] < \infty$.
  Thus by \cref{The:IntervalOST}, we conclude that $\expe[Y_T] \aord \expe[Y_0]$, i.e., $\expe[A_T] \aord \expe[\Phi_0] = \phi_Q(\lambda\_.0)$.
\end{proof}



\section{Termination Analysis}
\label{Se:TerminationAnalysis}

In this section, we develop a technique to reason about \emph{upper} bounds on higher moments $\expe[T^m]$ of the stopping time $T$.
We adapt the idea of expected-potential functions, but rely on a simpler convergence proof.
In this section, we assume $\calR = ([0,\infty], {\le},{+},{\cdot},0,1)$ to be a partially ordered semiring on extended nonnegative real numbers.

\begin{definition}
  A map $\psi : \Sigma \to \calM_\calR^{(m)}$ is said to be an \emph{expected-potential function} for upper bounds on stopping time if
  \begin{enumerate}[(i)]
    \item $\psi(\sigma)_0 = 1$ for all $\sigma \in \Sigma$,
    \item $\psi(\sigma) = \aone$ if $\sigma = \tuple{\_,\iskip,\kstop,\_}$, and
    \item $\psi(\sigma) \sqsupseteq \expe_{\sigma' \sim {\mapsto}(\sigma)}[\tuple{1,1,\cdots,1} \otimes \psi(\sigma')]$ for all non-terminating configuration $\sigma \in \Sigma$.
  \end{enumerate}
\end{definition}

Intuitively, $\psi(\sigma)$ is an upper bound on the moments of the evaluation steps upon termination for the computation that \emph{continues from} the configuration $\sigma$.
We define $A_n$ and $\Psi_n$ where $n \in \bbZ^+$ to be random variables on the probability space $(\Omega,\calF,\prob)$ of the trace semantics as
$A_n(\omega) \defeq \many{\tuple{n^k}_{0 \le k \le m}}$ and $\Psi_n(\omega) \defeq \psi(\omega_n)$.
Then we define $A_T(\omega) \defeq A_{T(\omega)}(\omega)$.
Note that $A_T = \many{\tuple{T^k}_{0 \le k \le m}}$.

We now show that a valid potential function for stopping time \emph{always} gives a sound upper bound.

\begin{theorem}
  $\expe[A_T] \aord \expe[\Psi_0]$.
\end{theorem}
\begin{proof}
  Let $C_n(\omega) \defeq \tuple{1,1,\cdots,1}$ if $n < T(\omega)$, otherwise $C_n(\omega) \defeq \tuple{1,0,\cdots,0}$.
  Then $A_T = \bigotimes_{i=0}^\infty C_i$.
  By \cref{Prop:MON}, we know that $\expe[A_T] = \lim_{n \to \infty} \expe[\bigotimes_{i=0}^n C_i]$.
  Thus it suffices to show that for all $n \in \bbZ^+$, $\expe[\bigotimes_{i=0}^n C_i] \aord \expe[\Psi_0]$.
  
  Observe that $\{C_n\}_{n \in \bbZ^+}$ is adapted to $\{\calF_n\}_{n \in \bbZ^+}$, because the event $\{ T \le n\}$ is $\calF_n$-measurable.
  Then we have
  \begin{align*}
    & \expe[\bigotimes\nolimits_{i=0}^n C_i \otimes \Psi_{n+1} \mid \calF_n] \\
    ={} & \bigotimes\nolimits_{i=0}^{n-1} C_i \otimes \expe[C_n \otimes \Psi_{n+1} \mid \calF_n], a.s., \\
    {}\aord{} & \bigotimes\nolimits_{i=0}^{n-1} C_i \otimes \Psi_n. 
  \end{align*}
  Therefore, $\expe[\bigotimes_{i=0}^n C_i \otimes \Psi_{n+1}] \aord \expe[\Psi_0]$ for all $n \in \bbZ^+$ by a simple induction.
  Because $\Psi_{n+1} \sqsupseteq \aone$, we conclude that
  \[
  \expe[\bigotimes\nolimits_{i=0}^n C_i] \aord \expe[\bigotimes\nolimits_{i=0}^n C_i \otimes \Psi_{n+1}] \aord \expe[\Psi_0].
  \]
\end{proof}



\section{Experimental Evaluation}
\label{Se:ExperimentDetails}

\cref{Ta:FullComparisonWithTACAS} and \cref{Fi:CompleteComparisonWithTACAS} are the complete versions of the statistics in \cref{Ta:ComparisonWithTACAS} and the plots in \cref{Fi:ComparisonWithTACAS}, respectively.

\tikzexternalenable
\begin{figure*}
  \centering
  \begin{tabular}{c@{\hspace{1ex}}c@{\hspace{1ex}}c@{\hspace{1ex}}c}
    \includegraphics{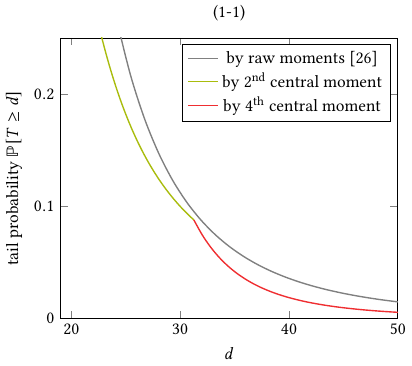}
    &
    \includegraphics{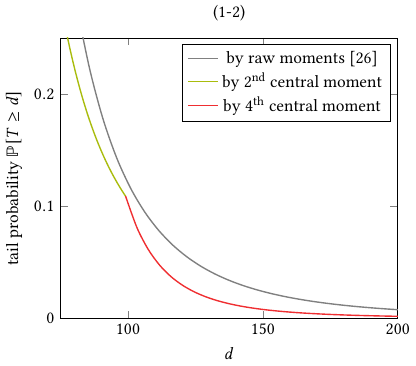}
    &
    \includegraphics{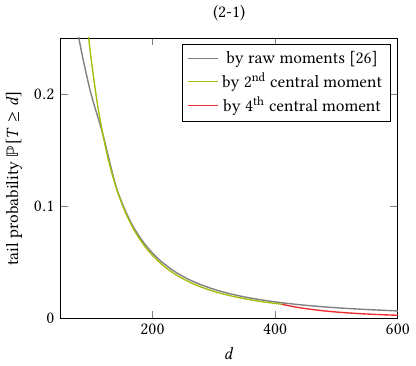}
    &
    \includegraphics{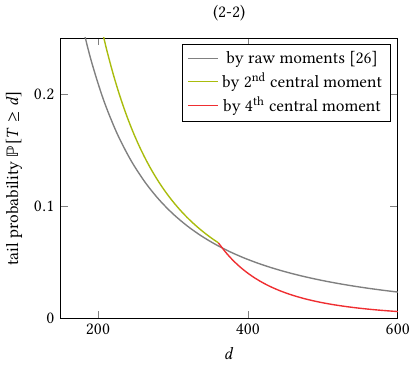}
    \\
    \includegraphics{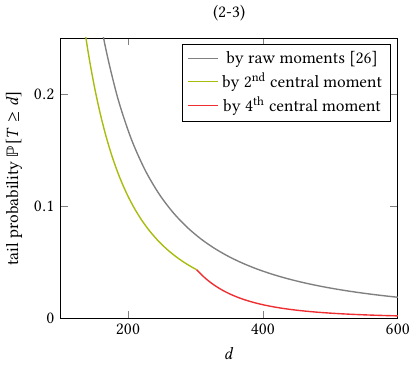}
    &
    \includegraphics{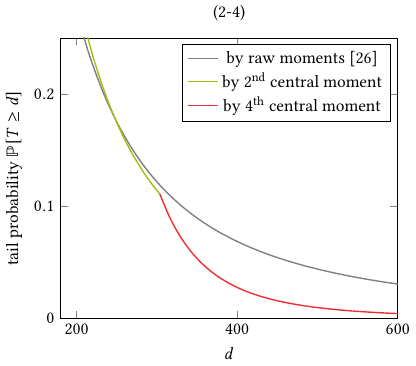}
    &
    \includegraphics{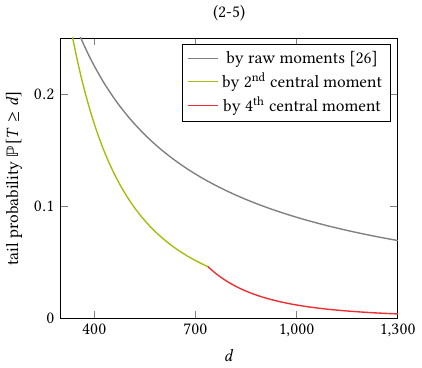}
  \end{tabular}
  \caption{Upper bound of the tail probability $\prob[T \ge d]$ as a function of $d$, with comparison to~\citet{TACAS:KUH19}. Each gray line is the minimum of tail bounds obtained from the raw moments of degree up to four inferred by~\citet{TACAS:KUH19}. Green lines and red lines are the tail bounds given by 2\textsuperscript{nd} and 4\textsuperscript{th} central moments inferred by our tool, respectively.}
  \label{Fi:CompleteComparisonWithTACAS}
\end{figure*}
\tikzexternaldisable

\cref{Tab:CompareToTACAS} compares our tool with~\citet{TACAS:KUH19} on their benchmarks for upper bounds on the first moments of runtimes.
The programs (1-1) and (1-2) are coupon-collector problems with a total of two and four coupons, respectively.
The other five are variants of random walks.
The first three are one-dimensional random walks: (2-1) is integer-valued, (2-2) is real-valued with continuous sampling, and (2-3) exhibits adversarial nondeterminism.
The programs (2-4) and (2-5) are two-dimensional random walks.
In the parentheses after the bounds, we record the degree of polynomials for the templates and the running time of the analysis.
%
All analyses completed is less than half a second, and our tool can derive better bounds than the compared tool~\cite{TACAS:KUH19}.

\begin{table}
  \centering
  \caption{Upper bounds of the expectations of runtimes, with comparison to~\citet{TACAS:KUH19}.}
  \label{Tab:CompareToTACAS}
  \resizebox{\columnwidth}{!}{%
  \begin{tabular}{cH||c||c}
    \hline
    program & pre-condition & upper bound by our tool & upper bound by~\citet{TACAS:KUH19} \\ \hline
    (1-1) & $\top$ & $13$ (deg 2, 0.106s) & $13$ (deg 1, 0.011s)  \\
    (1-2) & $\top$ & $44.6667$ (deg 4, 0.417s) & $68$ (deg 1, 0.020s) \\
    (2-1) & $x \ge 1$ & $20$ (deg 1, 0.096s) & $20$ (deg 1, 0.011s)  \\
    (2-2) & $x \ge 0$ & $75$ (deg 1, 0.116s) & $75$ (deg 1, 0.015s) \\
    (2-3) & $x > 0$ & $42$ (deg 1, 0.260s) & $62$ (deg 1, 0.017s) \\
    (2-4) & $x > 0 \wedge y > 0$ & $73$ (deg 1, 0.143s) & $96$ (deg 1, 0.013s) \\
    (2-5) & $x > y$ & $90$ (deg 1, 0.244s) & $90$ (deg 1, 0.017s) \\ \hline
  \end{tabular}%
  }
\end{table}

\begin{table}
  \centering
  \caption{Upper bounds on the raw/central moments of runtimes, with comparison to~\citet{TACAS:KUH19}. ``T/O'' stands for timeout after 30 minutes. ``N/A'' means that the tool is not applicable. ``-'' indicates that the tool fails to infer a bound. Entries with more precise results or less analysis time are marked in bold.}
  \label{Ta:FullComparisonWithTACAS}
  \resizebox{\columnwidth}{!}{%
  \begin{tabular}{@{\hspace{1pt}}c@{\hspace{1pt}}|@{\hspace{1pt}}c@{\hspace{1pt}}||@{\hspace{1pt}}c@{\hspace{1pt}}|@{\hspace{1pt}}c@{\hspace{1pt}}||@{\hspace{1pt}}c@{\hspace{1pt}}|@{\hspace{1pt}}c@{\hspace{1pt}}}
    \hline
    \multirow{2}{*}{program} & \multirow{2}{*}{moment} & \multicolumn{2}{c@{\hspace{1pt}}||@{\hspace{1pt}}}{this work} & \multicolumn{2}{c}{\citet{TACAS:KUH19}} \\ \hhline{~~----}
    & & upper bnd. & time (sec) & upper bnd. & time (sec) \\ \hline
    \multirow{5}{*}{(1-1)} & 2\textsuperscript{nd} raw & \bd{201} & 0.019 & \bd{201
    } & \bd{0.015}  \\ \hhline{~*{5}{-}}
    & 3\textsuperscript{rd} raw & \bd{3829} & \bd{0.019} & \bd{3829} & 0.020 \\ \hhline{~*{5}{-}}
    & 4\textsuperscript{th} raw & \bd{90705} & \bd{0.023} & \bd{90705} & 0.027 \\ \hhline{~*{5}{-}}
    & 2\textsuperscript{nd} central & \bd{32} & \bd{0.029} & N/A & N/A \\ \hhline{~*{5}{-}}
    & 4\textsuperscript{th} central & \bd{9728} & \bd{0.058} & N/A & N/A \\ \hline
    \multirow{5}{*}{(1-2)} & 2\textsuperscript{nd} raw & \bd{2357} & 1.068 & 3124 & \bd{0.037} \\ \hhline{~*{5}{-}}
    & 3\textsuperscript{rd} raw & \bd{148847} & 1.512 & 171932 & \bd{0.062} \\ \hhline{~*{5}{-}}
    & 4\textsuperscript{th} raw  & \bd{11285725} & 1.914 & 12049876 & \bd{0.096} \\ \hhline{~*{5}{-}}
    & 2\textsuperscript{nd} central & \bd{362} & \bd{3.346} & N/A & N/A \\ \hhline{~*{5}{-}}
    & 4\textsuperscript{th} central & \bd{955973} & \bd{9.801} & N/A & N/A \\ \hline
    \multirow{5}{*}{(2-1)} & 2\textsuperscript{nd} raw & \bd{2320} & \bd{0.016} & \bd{2320} & 11.380 \\ \hhline{~*{5}{-}}
    & 3\textsuperscript{rd} raw & \bd{691520} & \bd{0.018} & - & 16.056 \\ \hhline{~*{5}{-}}
    & 4\textsuperscript{th} raw & \bd{340107520} & \bd{0.021} & - & 23.414 \\ \hhline{~*{5}{-}}
    & 2\textsuperscript{nd} central & \bd{1920} & \bd{0.026} & N/A & N/A \\ \hhline{~*{5}{-}}
    & 4\textsuperscript{th} central & \bd{289873920} & \bd{0.049} & N/A & N/A \\ \hline
    \multirow{5}{*}{(2-2)} & 2\textsuperscript{nd} raw & \bd{8375} & \bd{0.022} & \bd{8375} & 38.463 \\ \hhline{~*{5}{-}}
    & 3\textsuperscript{rd} raw & \bd{1362813} & \bd{0.028} & - & 73.408  \\ \hhline{~*{5}{-}}
    & 4\textsuperscript{th} raw & \bd{306105209} & \bd{0.035} & - & 141.072 \\ \hhline{~*{5}{-}}
    & 2\textsuperscript{nd} central &  \bd{5875} & \bd{0.029} & N/A & N/A \\ \hhline{~*{5}{-}}
    & 4\textsuperscript{th} central & \bd{447053126} & \bd{0.086} & N/A & N/A \\ \hline
    \multirow{5}{*}{(2-3)} & 2\textsuperscript{nd} raw & \bd{3675} & \bd{0.039} & 6710 & 48.662 \\ \hhline{~*{5}{-}}
    & 3\textsuperscript{rd} raw & \bd{618584} & 0.049 & 19567045 & \bd{0.039} \\ \hhline{~*{5}{-}}
    & 4\textsuperscript{th} raw & \bd{164423336} & \bd{0.055}  & - & T/O \\ \hhline{~*{5}{-}}
    & 2\textsuperscript{nd} central & \bd{3048} & \bd{0.053} & N/A & N/A \\ \hhline{~*{5}{-}}
    & 4\textsuperscript{th} central & \bd{196748763} & \bd{0.123} & N/A & N/A \\ \hline
    \multirow{5}{*}{(2-4)} & 2\textsuperscript{nd} raw & \bd{6625} & \bd{0.035} & 10944 & 216.352 \\ \hhline{~*{5}{-}}
    & 3\textsuperscript{rd} raw & \bd{742825} & \bd{0.048} & - & 453.435 \\ \hhline{~*{5}{-}}
    & 4\textsuperscript{th} raw & \bd{101441320} & \bd{0.072} & - & 964.579 \\ \hhline{~*{5}{-}}
    & 2\textsuperscript{nd} central & \bd{6624} & \bd{0.051} & N/A & N/A \\ \hhline{~*{5}{-}}
    & 4\textsuperscript{th} central & \bd{313269063} &  \bd{0.215} & N/A & N/A \\ \hline
    \multirow{5}{*}{(2-5)} & 2\textsuperscript{nd} raw & \bd{21060} & \bd{0.045} & - & 216.605 \\ \hhline{~*{5}{-}}
    & 3\textsuperscript{rd} raw & \bd{9860940} & \bd{0.063} & - & 467.577 \\ \hhline{~*{5}{-}}
    & 4\textsuperscript{th} raw & \bd{7298339760} & \bd{0.101} & - & 1133.947 \\ \hhline{~*{5}{-}}
    & 2\textsuperscript{nd} central & \bd{20160} & \bd{0.068} & N/A & N/A \\ \hhline{~*{5}{-}}
    & 4\textsuperscript{th} central & \bd{8044220161} & \bd{0.271} & N/A & N/A \\ \hline
  \end{tabular}%
  }
\end{table}

\cref{Tab:CompareToAbsynth} compares our tool with \textsc{Absynth} by~\citet{PLDI:NCH18} on their benchmarks for upper bounds on the first moments of monotone cost accumulators.
Both tools are able to infer symbolic polynomial bounds.
\textsc{Absynth} uses a finer-grained set of base functions, and it supports bounds of the form $\maxp{x,y}$, which is defined as $\max(0,y-x)$.
Our tool is as precise as, and sometimes more precise than \textsc{Absynth},\footnote{We compared the precision by first simplifying the bounds derived by \textsc{Absynth} using the pre-conditions shown in \cref{Tab:CompareToAbsynth}, and then comparing the polynomials.} but it is less efficient than \textsc{Absynth}.
%
One reason for this could be that we use a full-fledged numerical abstract domain to derive the logical contexts, while \textsc{Absynth} employs less powerful but more efficient heuristics to do that step.
Nevertheless, all the analyses completed in less than one second.

\begin{table*}
  \centering
  \caption{Upper bounds of the expectations of monotone costs, with comparison to~\citet{PLDI:NCH18}.}
  \label{Tab:CompareToAbsynth}
  \resizebox{\textwidth}{!}{%
  \begin{tabular}{@{\hspace{1pt}}c@{\hspace{1pt}}|@{\hspace{1pt}}c@{\hspace{1pt}}||H@{\hspace{1pt}}c@{\hspace{1pt}}H||@{\hspace{1pt}}c}
    \hline
    program & pre-condition  & termination check & upper bound by our tool & lower bound & upper bound by \textsc{Absynth}~\cite{PLDI:NCH18} \\ \hline
    \textsf{2drdwalk} & $d<n$ & $\bbE[T] < \infty$ (deg 1, 0.267s)  & $2(n-d+1)$ (deg 1, 0.269s) & $n-d$ (deg 1, 0.256s) & $2\maxp{d,n+1}$ (deg 1, 0.170s) \\
    \textsf{C4B\_t09} & $x>0$ & $\bbE[T]<\infty$ (deg 1, 0.060s) & $17x$ (deg 1, 0.061s) & $1.6667 (x-1)$ (deg 1, 0.060s) & $17\maxp{0,x}$ (deg 1, 0.014s)  \\
    \textsf{C4B\_t13} & $x>0 \wedge y >0$ & $\bbE[T]<\infty$ (deg 1, 0.063s) & $1.25x+y$ (deg 1, 0.060s) & $x$ (deg 1, 0.061s) & $1.25 \maxp{0,x} + \maxp{0,y}$ (deg 1, 0.008s) \\
    \textsf{C4B\_t15} & $x > y \wedge y > 0$ & $\bbE[T]<\infty$ (deg 1, 0.144s) & $1.1667x$ (deg 1, 0.072s) &  $0.6667 (x -  y)$ (deg 1, 0.086s) & $1.3333\maxp{0,x}$ (deg 1, 0.014s)  \\
    \textsf{C4B\_t19} & $i>100 \wedge k >0$  & $\bbE[T] < \infty$ (deg 1, 0.059s) & $k+2i-49$ (deg 1, 0.059s) &  $k+2i-49$ (deg 1, 0.058s) & $\maxp{0,51+i+k} + 2 \maxp{0,i}$ (deg 1, 0.010s) \\
    \textsf{C4B\_t30} & $x>0 \wedge y > 0$  & $\bbE[T]<\infty$ (deg 1, 0.063s) & $0.5x+0.5y+2$ (deg 1, 0.044s) & $0$ (deg 1, 0.051s) & $0.5\maxp{0,x+2} + 0.5\maxp{0,y+2}$ (deg 1, 0.007s) \\
    \textsf{C4B\_t61} & $l \ge 8$ & $ \bbE[T] < \infty$ (deg 1, 0.058s) & $1.4286l$ (deg 1, 0.046s) & $ 1.4286 l -3$ (deg 1, 0.054s) & $\maxp{0,l} + 0.5 \maxp{1,l}$ (deg 1, 0.007s) \\
    \textsf{bayesian\_network} & $n>0$ &   $\bbE[T]<\infty$ (deg 1, 0.309s)  & $4n$ (deg 1, 0.264s) & $4n$ (deg 1, 0.292s) & $4\maxp{0,n}$ (deg 1, 0.057s) \\
    \textsf{ber} & $x<n$  & $\bbE[T]<\infty$ (deg 1, 0.041s) & $2(n-x)$ (deg 1, 0.035s) & $2(n-x)$ (deg 1, 0.041s) & $2\maxp{x,n}$ (deg 1, 0.004s) \\
    \textsf{bin} & $n>0$ & $\bbE[T]<\infty$ (deg 1, 0.042s) & $0.2(n+9)$ (deg 1, 0.036s) & $0.2n$ (deg 1, 0.039s) & $0.2 \maxp{0,n+9}$ (deg 1, 0.030s) \\
    \textsf{condand} & $n>0\wedge m>0$ &  $\bbE[T]<\infty$ (deg 1, 0.044s) & $2m$ (deg 1, 0.042s) & $0$ (deg 1, 0.043s) & $2\maxp{0,m}$ (deg 1, 0.004s) \\
    \textsf{cooling} & $mt > st \wedge t >0$ &  $\bbE[T]<\infty$ (deg 1, 0.080s) & $mt-st+0.42t+2.1$ (deg 1, 0.071s) & $0$ (deg 1, 0.076s) & $0.42\maxp{0,t+5} + \maxp{st,mt}$ (deg 1, 0.017s) \\
    \textsf{coupon} & $\top$ &  $\bbE[T^4]<\infty$ (deg 4, 0.118s) & $11.6667$ (deg 4, 0.066s) & $6.0624$ (deg 4, 0.063s) & $15$ (deg 1, 0.016s) \\
    \textsf{cowboy\_duel} & $\top$ &  $\bbE[T]<\infty$ (deg 1, 0.041s) & $1.2$ (deg 1, 0.030s) & $1.2$ (deg 1, 0.036s) & $1.2$ (deg 1, 0.004s) \\
    \textsf{cowboy\_duel\_3way} & $\top$ &  $\bbE[T]<\infty$ (deg 1, 0.127s) & $2.0833$ (deg 1, 0.110s) & $1$ (deg 1, 0.130s) & $2.0833$ (deg 1, 0.142s) \\
    \textsf{fcall} & $x < n$ &  $\bbE[T^2]<\infty$ (deg 2, 0.066s) & $2(n-x)$ (deg 1, 0.053s) & $2(n-x)$ (deg 1, 0.059s) & $2\maxp{x,n}$ (deg 1, 0.004s) \\
    \textsf{filling\_vol} & $volTF > 0$ &  $ \bbE[T]<\infty$ (deg 1, 0.096s) & $0.6667 volTF + 7$ (deg 1, 0.082s) & $ 0.1 volTF+0.05$ (deg 1, 0.092s) & $0.3333 \maxp{0,volTF+10} + 0.3333 \maxp{0, volTF + 11}$ (deg 1, 0.079s) \\
    \textsf{geo} & $\top$ &  $\bbE[T]<\infty$ (deg 1, 0.072s) & $5$ (deg 1, 0.061s) & $5$ (deg 1, 0.068s) & $5$ (deg 1, 0.003s) \\
    \textsf{hyper} & $x<n$ &  $\bbE[T]<\infty$ (deg 1, 0.041s) & $5(n-x)$ (deg 1, 0.035s) & $5 (n-x-1)$ (deg 1, 0.042s) & $5 \maxp{x,n}$ (deg 1, 0.005s) \\
    \textsf{linear01} & $x>2$ &  $\bbE[T]<\infty$ (deg 1, 0.040s) & $0.6x$ (deg 1, 0.034s) & $0.6(x-1)$ (deg 1, 0.038s) & $0.6\maxp{0,x}$ (deg 1, 0.008s) \\
    \textsf{prdwalk} & $x<n$ &  $\bbE[T]<\infty$ (deg 1, 0.047s) & $1.1429(n-x+4)$ (deg 1, 0.037s) & $1.1429 (n -x)$ (deg 1, 0.051s) & $1.1429 \maxp{x,n+4}$ (deg 1, 0.011s) \\
    \textsf{prnes} & $y>0 \wedge n <0$ &  $\bbE[T]<\infty$ (deg 1, 0.082s) & $- 68.4795n+0.0526(y-1)$ (deg 1, 0.071s)  & $-10 n$ (deg 1, 0.084s) & $68.4795 \maxp{n,0} + 0.0526\maxp{0,y}$ (deg 1, 0.016s) \\
    \textsf{prseq} & $y > 9 \wedge x - y > 2$ &  $\bbE[T]<\infty$ (deg 1, 0.073s) & $1.65x-1.5y$ (deg 1, 0.061s) & $1.65 x - 1.5 y -4.65$ (deg 1, 0.072s) & $1.65 \maxp{y,x} + 0.15 \maxp{0,y}$ (deg 1, 0.011s) \\
    \textsf{prspeed} & $x+1<n \wedge y<m$ &  $\bbE[T]<\infty $ (deg 1, 0.076s) & $2(m-y)+0.6667(n-x)$ (deg 1, 0.065s)  & $0.6667(n-x)-1.3333$ (deg 1, 0.074s) & $2\maxp{y,m} + 0.6667 \maxp{x,n}$ (deg 1, 0.010s) \\
    \textsf{race} & $h<t$ &  $\bbE[T]<\infty $ (deg 1, 0.052s) & $0.6667(t-h+9)$ (deg 1, 0.039s) & $0.6667 (t-h+1) $ (deg 1, 0.048s) & $0.6667 \maxp{h, t+9}$ (deg 1, 0.027s) \\
    \textsf{rdseql} & $x>0 \wedge y>0$ &  $\bbE[T]<\infty$ (deg 1, 0.069s) & $2.25x+y$ (deg 1, 0.059s)&  $2.25x$ (deg 1, 0.069s) & $2.25\maxp{0,x}+\maxp{0,y}$ (deg 1, 0.008s) \\
    \textsf{rdspeed} & $x+1<n \wedge y <m$ &  $\bbE[T]<\infty$ (deg 1, 0.065s) & $2(m-y)+0.6667(n-x)$ (deg 1, 0.062s) &  $0.6667(n-x)-1.3333$ (deg 1, 0.077s) & $2\maxp{y,m} + 0.6667\maxp{x,n}$ (deg 1, 0.010s) \\
    \textsf{rdwalk} & $x<n$ &  $\bbE[T]<\infty$ (deg 1, 0.039s) & $2(n-x+1)$ (deg 1, 0.035s) & $2(n-x)$ (deg 1, 0.041s) & $2 \maxp{x,n+1}$ (deg 1, 0.004s) \\
    \textsf{rejection\_sampling} & $n>0$ &  $\bbE[T^2]<\infty$ (deg 3, 0.171s) & $2n$ (deg 2, 0.071s) & $n$ (deg 2, 0.080s) & $2\maxp{0,n}$ (deg 1, 0.350s) \\
    \textsf{rfind\_lv} & $\top$ &  $\bbE[T]<\infty$ (deg 1, 0.033s) & $2$ (deg 1, 0.033s) & $2$ (deg 1, 0.035s) & $2$ (deg 1, 0.004s) \\
    \textsf{rfind\_mc} & $k>0$ &  $\bbE[T]<\infty$ (deg 1, 0.048s) & $2$ (deg 1, 0.047s) & $0$ (deg 1, 0.052s) & $\maxp{0,k}$ (deg 1, 0.007s) \\
    \textsf{robot} & $n>0$ &  $\bbE[T] < \infty$ (deg 1, 0.054s) & $0.2778(n+7)$ (deg 1, 0.047s) & $0.2778 n$ (deg 1, 0.051s) & $0.3846 \maxp{0,n+6}$ (deg 1, 0.017s) \\
    \textsf{roulette} & $n<10$ &  $ \bbE[T]<\infty $ (deg 1, 0.067s) & $-4.9333n+98.6667$ (deg 1, 0.057s) & $ - 4.9333 n + 54.2667 $ (deg 1, 0.064s) & $4.9333 \maxp{n,20}$ (deg 1, 0.073s) \\
    \textsf{sprdwalk} & $x<n$ &  $\bbE[T]<\infty$ (deg 1, 0.042s) & $2(n-x)$ (deg 1, 0.036s) & $2(n-x)$ (deg 1, 0.042s) & $2\maxp{x,n}$ (deg 1, 0.004s) \\
    \textsf{trapped\_miner} & $n>0$ &  $\bbE[T]<\infty$ (deg 1, 0.095s) & $7.5n$ (deg 1, 0.081s) & $7.5n$ (deg 1, 0.097s) & $7.5\maxp{0,n}$ (deg 1, 0.015s) \\
    \textsf{complex} & $y,w,n,m>0$ & $\bbE[T]<\infty$ (deg 2, 0.145s) & $4mn+2n+w+0.6667(y+1)$ (deg 2, 0.142s) & $6n+0.6667y$ (deg 1, 0.142s) & $(4\maxp{0,m}+2)\maxp{0,n}+(\maxp{0,w} +0.3333\maxp{0,y})\maxp{0,y+1}+0.6667$ (deg 2, 0.451s) \\
    \textsf{multirace} & $n>0 \wedge m>0$ &  $\bbE[T^2] < \infty$ (deg 4, 0.506s) & $2mn+4n$ (deg 2, 0.080s) &    $2mn+2n$ (deg 2, 0.083s) & $2\maxp{0,m}\maxp{0,n} + 4\maxp{0,n}$ (deg 2, 0.692s) \\
    \textsf{pol04} & $x>0$ &  $ \bbE[T^2] < \infty$ (deg 4, 0.062s)  & $4.5x^2+10.5x$ (deg 2, 0.052s) & $4.5 x^2 + 1.5 x $ (deg 2, 0.052s) & $1.5\maxp{0,x}^2+3\maxp{1,x}\maxp{0,x} + 10.5 \maxp{0,x}$ (deg 2, 0.101s) \\
    \textsf{pol05} & $x>0$ &  $ \bbE[T^2] < \infty$ (deg 4, 0.085s)  & $0.6667(x^2+3x)$ (deg 2, 0.059s) &  $ 0.6667 (x^2- x)$ (deg 2, 0.059s) & $2\maxp{0,x+1}\maxp{0,x}$ (deg 2, 0.133s) \\
    \textsf{pol07} & $n>1$ &  $\bbE[T^2] < \infty$ (deg 4, 0.100s) & $1.5n^2-4.5n+3$ (deg 2, 0.057s) & $1.5n^2-4.5n+3$ (deg 2, 0.058s) & $1.5\maxp{1,n} \maxp{2,n}$ (deg 2, 0.203s) \\
    \textsf{rdbub} & $n>0$ &  $\bbE[T^2] < \infty$ (deg 4, 0.105s) & $3n^2$ (deg 2, 0.055s) & $3 n^2$ (deg 2, 0.059s) & $3\maxp{0,n}^2$ (deg 2, 0.030s) \\
    \textsf{recursive} & $arg_l < arg_h$ & T/O & $0.25 (arg_h -arg_l)^2 + 1.75 (arg_h  - arg_l)$ (deg 2, 0.313s) & $2 (arg_h -  arg_l)$ (deg 1, 0.233s)  & $0.25\maxp{arg_l,arg_h}^2+1.75\maxp{arg_l,arg_h}$ (deg 2, 0.398s) \\
    \textsf{trader} & $mP > 0 \wedge sP > mP$ & $\bbE[T^2]<\infty$ (deg 2, 0.079) & $1.5(sP^2-mP^2+sP - mP)$ (deg 2, 0.073s) & $1.5(sP^2-mP^2+sP-mP)$ (deg 2, 0.074s) & $1.5\maxp{mP,sP}^2+3 \maxp{mP,sP} \maxp{0,mP} +1.5\maxp{mP,sP} $ (deg 2, 0.192s) \\ \hline
  \end{tabular}%
  }
\end{table*}

\cref{Tab:CompareToPLDI} compares our tool with~\citet{PLDI:WFG19} on their benchmarks for lower and upper bounds on the first moments of accumulated costs.
All the benchmark programs satisfy the bounded-update property.
To ensure soundness, our tool has to perform an extra termination check required by \cref{The:Soundness}.
Our tool derives similar symbolic lower and upper bounds, compared to~the results of \citet{PLDI:WFG19}.
Meanwhile, our tool is more efficient than the compared tool.
One source of slowdown for~\citet{PLDI:WFG19} could be that they use Matlab as the LP backend, and the initialization of Matlab has significant overhead.

\begin{table*}
  \centering
  \caption{Upper and lower bounds of the expectation of (possibly) non-monotone costs, with comparison to~\citet{PLDI:WFG19}.}
  \label{Tab:CompareToPLDI}
  \resizebox{\textwidth}{!}{%
  \begin{tabular}{@{\hspace{1pt}}c@{\hspace{1pt}}|@{\hspace{1pt}}c@{\hspace{1pt}}||@{\hspace{1pt}}c@{\hspace{1pt}}|@{\hspace{1pt}}c@{\hspace{1pt}}|@{\hspace{1pt}}c@{\hspace{1pt}}||@{\hspace{1pt}}c@{\hspace{1pt}}}
    \hline
    program & pre-cond.  & termination &  & bounds by our tool & bounds by~\citet{PLDI:WFG19} \\ \hline
    Bitcoin & \multirow{2}{*}{$x \ge 1$} &  $\bbE[T]<\infty$  & ub. & $-1.475x$ (deg 1, 0.078s)   & $-1.475 x + 1.475$ (deg 2, 3.705s)  \\ \hhline{~|~||~|-|-||-}
    Mining & & (deg 1, 0.080s) & lb. & $-1.5x$ (deg 1, 0.088s)  & $-1.5x$ (deg 2, 3.485s) \\ \hline
    Bitcoin & \multirow{2}{*}{$y \ge 0$} &  $\bbE[T^2] < \infty$  & ub. & $-7.375y^2 -66.375y$ (deg 2, 0.127s)  & $-7.375 y^2 - 41.625 y + 49$ (deg 2, 5.936s)  \\ \hhline{~|~||~|-|-||-}
    Mining Pool & & (deg 4, 0.168s) & lb. & $- 7.5y^2-67.5 y$ (deg 2, 0.134s)  & $-7.5y^2 -67.5y$ (deg 2, 6.157s) \\ \hline
    Queueing & \multirow{2}{*}{$ n > 0$}  & $\bbE[T^2] < \infty$  & ub. &  $0.0531 n$ (deg 2, 0.134s)  & $0.0492n$ (deg 3, 69.669s)  \\ \hhline{~|~||~|-|-||-}
    Network & & (deg 2, 0.208s) & lb. & $0.028n$ (deg 2, 0.139s)  & $0.0384n$ (deg 3, 68.849s)\\ \hline
    Running & \multirow{2}{*}{$ x \ge 0$} &  $\bbE[T^2] < \infty$  & ub. & $0.3333 (x^2 + x)$ (deg 2, 0.063s)  & $0.3333 x^2 + 0.3333 x$ (deg 2, 3.766s)  \\ \hhline{~|~||~|-|-||-}
    Example & &  (deg 2, 0.064s)& lb. & $0.3333( x^2 + x)$ (deg 2, 0.059s)  & $0.3333 x^2+0.3333x -0.6667$ (deg 2, 3.555s) \\ \hline
    Nested & \multirow{2}{*}{$i \ge 0$} &  $\bbE[T^2] < \infty$  & ub. & $0.3333i^2+i$ (deg 2, 0.117s)  & $0.3333 i^2 + i$ (deg 2, 28.398s)  \\ \hhline{~|~||~|-|-||-}
    Loop & & (deg 4, 0.127s) & lb. &  $0.3333 i^2 + i$ (deg 2, 0.115s)  & $0.3333 i^2-i$ (deg 2, 7.299s) \\ \hline
    Random & \multirow{2}{*}{$x \le n$} &  $\bbE[T]<\infty$ & ub.  & $2.5x-2.5n-2.5$ (deg 1, 0.064s)   & $2.5x-2.5n$ (deg 2, 4.536s)  \\ \hhline{~|~||~|-|-||-}
    Walk & & (deg 1, 0.063s)& lb. & $2.5 x-2.5n-2.5$ (deg 1, 0.068s)  & $2.5x-2.5n-2.5$ (deg 2, 4.512s) \\ \hline
    \multirow{2}{*}{2D Robot} & \multirow{2}{*}{$ x \ge y$} & $\bbE[T^2] < \infty$  & ub. &  $1.7280(x-y)^2 +31.4539(x-y)+126.5167$ (deg 2, 0.132s)  & $1.7280(x-y)^2 + 31.4539 (x- y)+126.5167$ (deg 2, 7.133s)  \\ \hhline{~|~||~|-|-||-}
    & & (deg 2, 0.145s)& lb. & $1.7280(x-y)^2+31.4539(x -y)+29.7259$ (deg 2, 0.121s)  &  $1.7280(x-y)^2 + 31.4539 (x -y)$ (deg 2, 7.040s) \\ \hline
    Good  &  {$d \le 30 \wedge {}$} &   $\bbE[T^2] < \infty$  &  ub. & $-0.5n - 3.6667d + 117.3333$ (deg 2, 0.093s)   & $0.0067 dn -0.7n-3.8035 d+0.0022d^2 +119.4351$ (deg 2, 5.272s)  \\ \hhline{~|~||~|-|-||-}
    Discount & $n \ge 1$ & (deg 2, 0.093s)& lb. & $ - 0.005 n^2-0.5 n$ (deg 2, 0.092s)  & $0.0067 dn - 0.7133n-3.8123d + 0.0022d^2+112.3704$ (deg 2, 5.323s) \\ \hline
    Pollutant  & \multirow{2}{*}{$n \ge 0$} &    $\bbE[T^2] <\infty$  & ub. & $-0.2n^2 +50.2 n$ (deg 2, 0.092s) & $-0.2 n^2 + 50.2n$ (deg 2, 5.851s)  \\  \hhline{~|~||~|-|-||-}
    Disposal & &(deg 2, 0.091s) & lb. & $ -0.2 n^2+50.2n -435.6$ (deg 2, 0.094s)  & $-0.2n^2 +50.2n-482$ (deg 2, 5.215s) \\ \hline
    Species & {$a \ge 5 \wedge {}$} & $\bbE[T^2] < \infty$  & ub. &  $40ab-180a-180b+810$ (deg 2, 0.045s)  & $40ab-180a-180b+810$ (deg 3, 5.545s) \\ \hhline{~|~||~|-|-||-}
    Fight & $b \ge 5$ & (deg 2, 0.042s) &  lb. & N/A  & N/A \\ \hline
  \end{tabular}%
  }
\end{table*}



\section{Case Study: Timing-Attack Analysis}
\label{Se:TimingAttackAnalysis}

We motivate our work on central-moment analysis using a probabilistic program
with a \emph{timing-leak} vulnerability, and demonstrate how the results from an analysis
can be used to bound the success rate of an attack program that attempts to exploit the vulnerability.
The program is extracted and modified from a web application provided by DARPA
during engagements as part of the STAC program~\cite{misc:STAC}.
In essence, the program models a password checker that compares an input \id{guess} with
an internally stored password \id{secret}, represented as two $N$-bit vectors.
The program in \cref{Fi:PasswordChecker}(a) is the interface of the checker,
and \cref{Fi:PasswordChecker}(b) is the comparison function \textsf{compare},
which carries out most of the computation.
The statements of the form ``$\kw{tick}(\cdot)$'' represent a cost model for the
running time of \textsf{compare}, which is assumed to be observable by the
attacker.
\textsf{compare} iterates over the bits from high-index to low-index,
and the running time expended during the processing of bit $i$ depends on
the current comparison result (stored in $\id{cmp}$), as well as on the
values of the $i$-th bits of $\id{guess}$ and $\id{secret}$.
Because the running time of \textsf{compare} might \emph{leak} information about the
relationship between $\id{guess}$ and $\id{secret}$, \textsf{compare} introduces some \emph{random delays}
to add noise to its running time.
However, we will see shortly that such a countermeasure does \emph{not} protect the
program from a timing attack.

\begin{figure*}
\centering
{\small
\begin{subfigure}{0.54\textwidth}
    \begin{pseudo}
    \kw{func} \textsf{compare}(\id{guess}, \id{secret}) \kw{begin} \\+
      $i \coloneqq N; \id{cmp} \coloneqq 0;$ \\
      \kw{while} $i>0$ \kw{do} \\+
        \kw{tick}($2$); \\ 
        \kw{while} $i>0$ \kw{do} \\+
          \kw{if} \kw{prob}($0.5$) \kw{then} \kw{break} \kw{fi} \\
          \kw{tick}($5$); \\
          \kw{if} $\id{cmp} > 0 \vee (\id{cmp} = 0 \wedge \id{guess}[i] > \id{secret}[i])$ \\
          \kw{then} $\id{cmp} \coloneqq 1$ \\
          \kw{else} \\+
            \kw{tick}($5$); \\
            \kw{if} $\id{cmp} < 0 \vee (\id{cmp} = 0 \wedge \id{guess}[i] < \id{secret}[i])$ \\
            \kw{then} $\id{cmp} \coloneqq {-1}$ \\
            \kw{fi} \\-
          \kw{fi} \\
          \kw{tick}($1$); \\
          $i \coloneqq i - 1$ \\-
        \kw{od} \\-
      \kw{od}; \\
      \kw{return} \id{cmp} \\-
    \kw{end}
  \end{pseudo}
  \caption*{(b)}
\end{subfigure}
\hfill
\begin{subfigure}{0.42\textwidth}
  \centering
  \begin{tabular}{l}
    \begin{pseudo}
      \kw{func} \textsf{check}(\id{guess}) \kw{begin} \\+
        $\id{cmp} \coloneqq \mathsf{compare}(\id{guess}, \id{secret});$ \\
        \kw{if} $\id{cmp} = 0$ \kw{then} \\+
          $\mathsf{login}()$ \\-
        \kw{fi} \\-
      \kw{end}
    \end{pseudo}
  \\
   \multicolumn{1}{c}{(a)} \\
   \\
    \begin{pseudo}
    $\id{guess} \coloneqq \vec{0}$; \\
    $i \coloneqq N$; \\
    \kw{while} $i > 0$ \kw{do} \\+
      $\id{next} \coloneqq \id{guess};$ \\
      $\id{next}[i] \coloneqq 1;$ \\
      $\id{est} \coloneqq \kw{estimateTime}(K,\mathsf{check}(\id{next}));$ \\
      \kw{if} $\id{est} \le 13N - 1.5i$ \kw{then} \\+
        $\id{guess}[i] \coloneqq 0$ \\-
      \kw{else} \\+
        $\id{guess}[i] \coloneqq 1$ \\-
      \kw{fi}; \\
      $i \coloneqq i - 1$ \\-
    \kw{od}
  \end{pseudo}
   \\
  \multicolumn{1}{c}{(c)}
  \end{tabular}
  \caption*{}
\end{subfigure}
}
\caption{\label{Fi:PasswordChecker}
(a) The interface of the password checker.
(b) A function that compares two bit vectors, adding some random noise.
(c) An attack program that attempts to exploit the timing properties of \textsf{compare} to
find the value of the password stored in $\id{secret}$.
}
\end{figure*}

We now show how the moments of the running time of \textsf{compare}---the
kind of information provided by our central-moment analysis (\cref{Se:DerivationSystem})---are
useful for analyzing the success probability of the attack program given in
\cref{Fi:PasswordChecker}(c).
Let $T$ be the random variable for the running time of \textsf{compare}.
A standard timing attack for such programs is to guess the bits of $\mathit{secret}$ successively.
%
The idea is the following: Suppose that we have successfully obtained the bits $\mathit{secret}[i+1]$ through
$\mathit{secret}[N]$;
we now guess that the next bit, $\mathit{secret}[i]$, is $1$ and set $\mathit{guess}[i] \coloneqq 1$.
Theoretically, if the following two conditional expectations
\begin{align}
  \begin{split} \expe[T_1] & \defeq \expe[T \mid \textstyle{\bigwedge_{j=i+1}^N (\id{secret}[j] = \id{guess}[j])} \\ & \qquad\qquad {} \wedge (\id{secret}[i] = 1 \wedge \id{guess}[i] = 1)] \end{split} \label{Eq:GuessCorrect} \\
  \begin{split} \expe[T_0] & \defeq \expe[T \mid \textstyle{\bigwedge_{j=i+1}^N (\id{secret}[j] = \id{guess}[j])}\\ & \qquad \qquad {} \wedge (\id{secret}[i] = 0 \wedge \id{guess}[i] = 1)] \end{split} \label{Eq:GuessWrong}
\end{align}
have a significant difference, then there is an opportunity to check our guess by
running the program multiple times, using the \emph{average} running time as
an estimate of $\expe[T]$, and choosing the value of $\mathit{guess}[i]$ according
to whichever of \eqref{Eq:GuessCorrect} and \eqref{Eq:GuessWrong} is closest to our
estimate.
However, if the difference between $\expe[T_1]$ and $\expe[T_0]$ is not significant enough,
or the program produces a large amount of noise in its running time, the attack might not be
realizable in practice.
To determine whether the timing difference represents an exploitable vulnerability,
we need to reason about the attack program's success rate.

Toward this end, we can analyze the failure probability for setting $\mathit{guess}[i]$
incorrectly, which happens when, due to an unfortunate fluctuation, the running-time
estimate $\id{est}$ is closer to one of $\expe[T_1]$ and $\expe[T_0]$, but the truth is actually
the other.
For instance, suppose that $\expe[T_0] < \expe[T_1]$ and $est < \frac{\expe[T_0]+\expe[T_1]}{2}$;
the attack program would pick $T_0$ as the truth, and set $\mathit{guess}[i]$ to $0$.
If such a choice is \emph{incorrect}, then the actual distribution of $\id{est}$
on the $i$-th round of the attack program satisfies $\expe[\id{est}] = \expe[T_1]$,
and the probability of this failure event is
\begin{align*}
& \prob\lrsq{\id{est} < \frac{\expe[T_0]+\expe[T_1]}{2}} \\
 ={} & \prob\lrsq{ \id{est} - \expe[T_1] < \frac{\expe[T_0]-\expe[T_1] }{2} } \\ 
={} & \prob\lrsq{ \id{est} - \expe[\id{est}] < \frac{\expe[T_0]-\expe[T_1] }{2} } 
\end{align*}
under the condition given by the conjunction in \eqref{Eq:GuessCorrect}.
This formula has exactly the same shape as a \emph{tail probability},
which makes it possible to utilize moments and \emph{concentration-of-measure} inequalities~\cite{book:Dubhashi09} to bound the probability.

The attack program is parameterized by $K > 0$, which represents the number of trials it performs
for each bit position to obtain an estimate of the running time.
Assume that we have applied our central-moment-analysis technique
(\cref{Se:DerivationSystem}),
and obtained the following inequalities on the
mean (i.e., the first moment), the second moment, and the variance (i.e., the second
central moment) of the quantities \eqref{Eq:GuessCorrect} and
\eqref{Eq:GuessWrong}.
\Omit{
\twr{Missing: add the second moment for use later to contrast the attack-program bounds we
obtain via variances versus the weaker bounds one would obtain using Markov's inequality.}
}
\begin{align}
  \expe[T_1] & \ge 13N,      & \expe[T_1] & \le 15N,    \notag  \\
   \vari[T_1] & \le 26N^2 + 42N, \label{Eq:TimeCorrect}  \\
  \expe[T_0] & \ge 13N - 5i, & \expe[T_0] & \le 13N - 3i, \notag \\
   \vari[T_0] & \le 8N-36i^2+52Ni+24i. \label{Eq:TimeWrong}  
\end{align}
To bound the probability that the attack program makes an incorrect
guess for the $i$-th bit, we do case analysis:

\begin{itemize}
 \item
  Suppose that $\id{secret}[i] = 1$, but the attack program assigns $\id{guess}[i] \coloneqq  0$.
  The truth---with respect to the actual distribution of the running time $T$ of \textsf{compare}
  for the $i$-th bit---is that $\expe[\id{est}] = \expe[T_1]$,
  but the attack program in \cref{Fi:PasswordChecker}(c)
  executes the then-branch of the conditional statement.
  Thus, our task reduces to that of bounding $\prob[\id{est} < 13N - 1.5i]$.
  The estimate $\id{est}$ is the average of $K$ i.i.d.\ random variables drawn from a distribution with
  mean $\expe[T_1]$ and variance $\vari[T_1]$.
  We derive the following, using the inequalities from \eqref{Eq:TimeCorrect}:
  \begin{align}
    \expe[\id{est}] & = \expe[T_1] \ge 13N, \notag \\
     \vari[\id{est}] & = \frac{\vari[T_1]}{K} \le \frac{26N^2+42N}{K}. \label{Eq:EstMoments}
  \end{align}
  We now use Cantelli's inequality.
  \begin{proposition*}[Cantelli's inequality]
  If $X$ is a random variable and $a>0$,
  then we have $\prob[X - \expe[X] \le -a] \le \frac{\vari[X]}{\vari[X]+a^2}$
  and $\prob[X - \expe[X] \ge a] \le \frac{\vari[X]}{\vari[X]+a^2}$.
  \end{proposition*}
  We are now able to derive an upper bound on $\prob[\id{est} < 13N - 1.5i]$ as follows:
  \begin{align*}
    & \prob[\id{est} \le 13N - 1.5i] \\
    ={} & \prob[\id{est} - 13N \le -1.5i]  \\
    {}\le{} & \prob[\id{est} - \expe[\id{est}] \le -1.5i]  & \reason{$\expe[\id{est}] \ge 13N$ by \eqref{Eq:EstMoments}} \\
    {}\le{} & \frac{\vari[\id{est}]}{\vari[\id{est}]+(1.5i)^2} & \reason{Cantelli's inequality} \\
    ={} & \frac{26N^2+42N}{26N^2+42N+2.25Ki^2}. 
  \end{align*}
  
\item
  The other case, in which $\id{secret}[i]=0$ but the attack program
  chooses to set $\id{guess}[i] \coloneqq 1$, can be analyzed in a similar
  way to the previous case, and the bound obtained is the following:
  \begin{align*}
    & \prob[\id{est} > 13N - 1.5i] \\
    {} \le {} & \prob[\id{est} \ge 13N - 1.5i] \\ 
    {}\le{} & \frac{8N-36i^2+52Ni+24i }{8N-36i^2+52Ni+24i + 2.25Ki^2 }. 
  \end{align*}
\end{itemize}
Let $F^i_1$ and $F^i_0$, respectively, denote the two upper bounds on the failure probabilities
for the $i$-th bit.

For the attack program to succeed, it has to succeed for all bits.
If the number of bits is $N=32$, and in each iteration the number of trials that the
attack program uses to estimate the running time is $K=10^4$, we derive a
\emph{lower} bound on the success rate of the attack program from the
upper bounds on the failure probabilities derived above:
\[
  \prob[\textsc{Success}] \ge \prod_{i=1}^{32} (1 - \max(F^i_1, F^i_0)) \ge 0.219413,
\]
which is low, but not insignificant.
However, the somewhat low probability is caused by a property of \textsf{compare}:
if $\id{guess}$ and $\id{secret}$ share a very long prefix, then the running-time behavior on different values of $guess$
becomes indistinguishable.
However, if instead we bound the success rate for all but the last six bits, we obtain:
\[
  \prob[\textsc{Success for all but the last six bits}] \ge 0.830561,
\]
which is a much higher probability!
The attack program can enumerate all the possibilities to resolve the last six bits.
(Moreover, by brute-forcing the last six bits, a total of $10,\!000 \times 26 + 64 = 260,\!064$
calls to \textsf{check} would be performed, rather than $320,\!000$.)

Overall, our analysis concludes that the \textsf{check} and \textsf{compare} procedures in
\cref{Fi:PasswordChecker} are vulnerable to a timing attack.



\fi

\end{document}